\documentclass[11pt]{article}

\usepackage[utf8]{inputenc}
\usepackage[T1]{fontenc}
\usepackage[a4paper,margin=1in]{geometry}
\usepackage{amsmath,amssymb,amsfonts,amsthm,amsopn}
\usepackage{mathtools}
\usepackage{graphicx}
\usepackage{epstopdf}
\usepackage{algorithm}
\usepackage{algorithmic}
\usepackage{float}
\usepackage{flafter}
\usepackage{placeins}
\usepackage{dsfont}
\usepackage{enumitem}
\usepackage{microtype}
\usepackage{newunicodechar}
\usepackage[hidelinks]{hyperref}
\usepackage[capitalise,nameinlink]{cleveref}

\newunicodechar{’}{'}
\newunicodechar{–}{--}
\newunicodechar{‐}{-}

\DeclareGraphicsExtensions{.pdf,.png,.jpg,.eps}
\setlist[enumerate]{leftmargin=.5in}
\setlist[itemize]{leftmargin=.5in}
\newtheorem{theorem}{Theorem}[section]
\newtheorem{lemma}[theorem]{Lemma}
\newtheorem{corollary}[theorem]{Corollary}

\theoremstyle{remark}

\crefname{hypothesis}{Hypothesis}{Hypotheses}
\crefname{fact}{Fact}{Facts}

\newenvironment{keywords}
  {\par\medskip\noindent\textbf{Keywords. }\ignorespaces}
  {\par\medskip}
\newenvironment{MSCcodes}
  {\par\noindent\textbf{MSC codes. }\ignorespaces}
  {\par\medskip}

\title{Prediction-Oriented Sequential Bayesian Design of Physical Experiments
Informed by Computer Models}
\author{
Hao Zhu\thanks{Institute of Applied Statistics, Johannes Kepler University Linz,
Linz, Austria (\texttt{hao.zhu@jku.at}).}
\and
Markus Hainy\thanks{Institute of Applied Statistics, Johannes Kepler University
Linz, Linz, Austria (\texttt{markus.hainy@jku.at}).}
}
\date{September 29, 2026}

\hypersetup{
  pdftitle={Prediction-Oriented Sequential Bayesian Design of Physical Experiments Informed by Computer Models},
  pdfauthor={Hao Zhu, Markus Hainy}
}

\begin{document}

	\maketitle
	
	\begin{abstract}
		In many scientific and engineering domains, physical experiments are often costly, time-consuming and limited to a certain number of design points. The Kennedy \& O’Hagan (KOH) model framework is widely used for combining simulator runs with limited experimental observations. Under Bayesian implementation, simulator output, model discrepancy, and observation noise are jointly modeled by coupled Gaussian processes, followed by coherent posterior inference and uncertainty quantification, thereby providing a probabilistic foundation for experimental design. Existing sequential design criteria mainly target better parameter estimation or global uncertainty reduction but may overlook the local structure of the predicted response. This work presents a sequential Bayesian experimental design (BED) framework aimed at improving the predictive performance of the KOH model. The framework introduces a mutual information (MI)–based criterion combined with a local complexity score which measures nonlinearity of the posterior predictive mean near each candidate point, leading to more efficient design decisions. To make the criterion computationally feasible in a fully Bayesian setting, we employ Gaussian mixture reduction with fast covariance update strategies. Through several challenging test problems and an example of the Duffing oscillator on a discrete design space, we demonstrate the effectiveness of the proposed method and illustrate its difference from conventional design baselines under sequential (offline) and adaptive (online) BED settings. The numerical results suggest that incorporating local response complexity into the MI criterion can substantially improve prediction-oriented design for physical experiments informed by computer models.
		
	\end{abstract}
	
	
	
	\section{Introduction}
	A computer model is often approximated by a computationally cheaper emulator, which serves as a surrogate for the expensive mechanistic model. The mechanistic model can be regarded as an abstraction of real-world phenomena but inevitably deviates from the true physical process. The deviation, also known as model discrepancy, is typically inferred from field data and used to correct the computer model. The Kennedy and O’Hagan (KOH) framework provides a well-established statistical approach for jointly modeling simulation and field data while accounting for model discrepancy \cite{KennedyOHagan2001, CoxParkSinger2001, higdon2004combining}. However, acquiring field data, either from observations or through laboratory experiments, can be expensive and challenging. For instance, the discovery of new materials with a specific property often demands a significant number of costly experiments \cite{Dehghannasiri2017}. In these cases, gathering the most informative field data through sequential experimental design can significantly expedite the modeling process and lower the cost of the experiments. \par
	Based on the KOH framework, researchers have already proposed several optimal design methods. An expected improvement criterion was developed for sequential computer experiments by comparing models with and without discrepancy over the calibration space \cite{Pratola2013}. Information-theoretic and preposterior approaches have also been explored for Bayesian design of computationally intensive models \cite{Huan2013,Arendt2015}. Several studies have focused on learning unknown calibration parameters. A sequential D-optimal design method for computer model calibration was suggested in \cite{Diao2022}. In a related effort, the expected integrated variance of the calibration parameters has been used as a sequential Bayesian design criterion \cite{Surer2023}. Another approach uses a hybrid design which combines D-optimality with space-filling designs for prediction in physical experiments \cite{Krishna2021}. A common challenge is that calibration parameters are often multivariate and confounded with Gaussian process (GP) hyperparameters, which results in identifiability issues. Even though projection methods \cite{Gu2018, Wang2020, Xie2020} can improve the identifiability of calibration parameters, computer models remain imperfect. Structural model inadequacy relative to the true physical process persists even when the model parameters are accurately calibrated. Our work therefore directly focuses on improving the overall predictive capability of the KOH model rather than on parameter calibration. Indeed, some work has already considered prediction-oriented criteria, including Maximin Latin Hypercube Design (LHD) \cite{bayarri2007framework}, integrated mean squared prediction error \cite{ranjan2011follow} and related expected improvement \cite{Williams2011} criteria for simulator runs and physical experiments. More recently, a joint integrated mean squared prediction error (IMSPE) and maximin LHD strategy was proposed to improve prediction \cite{Leatherman2017combined}. Building on this work, a closed-form IMSPE criterion was later derived for continuous optimization of computer model designs under the KOH framework \cite{Koermer2024}. These methods provide useful baselines, but they mostly focus on either global variance reduction, parameter information, or geometric coverage.\par
	In this paper, we present a prediction-oriented sequential design problem for physical experiments under the KOH framework. 
	The mutual information (MI) is commonly used as a Bayesian design criterion since it can incorporate posterior uncertainty in calibration parameters, discrepancy, and observation noise. Based on the fully Bayesian implementation of the KOH model \cite{higdon2004combining}, we use predictive MI as the baseline acquisition criterion. In the setting we consider, the predictive MI value is high if the prospective observation is informative over a target prediction set. The computation of the MI is notoriously challenging. Unlike recent approaches that optimize variational lower bounds of the MI using gradient-based methods \cite{Foster2019, Foster2020, kleinegesse2021gradientbasedbayesianexperimentaldesign}, we employ nested Monte Carlo (NMC) \cite{Ryan2016}, together with computational acceleration techniques \cite{Gramacy2020Surrogates, Schieferdecker2009}. However, in practice the designs selected through the predictive MI criterion are often not the optimal ones in terms of predictive performance. The posterior predictive mean may exhibit local nonlinear structure, and observations in such regions can be especially valuable for learning the physical response. This idea is related to local polynomial approximation, where deviations from a local affine fit reflect higher-order response structure, and adaptive design strategies that use local linear approximations to detect nonlinear regions \cite{FanGijbels1996, Crombecq2011}. We therefore augment predictive MI with a local complexity score that measures the local nonlinearity of the posterior predictive mean.\par 
	The main novelty and contributions of this paper are as follows.
	\begin{itemize}
		\item We formulate a prediction-oriented sequential Bayesian design framework for physical experiments under the KOH model, using predictive mutual information to target future physical responses rather than calibration parameters alone.
		
		\item We introduce a local complexity score for the posterior predictive mean and incorporate it with predictive MI as a hybrid design criterion.
		
		\item We consider online adaptive design of experiments (ADE), where the model is updated after collecting each new observation, and  offline greedy sequential design of experiments (SDE), where the design points are selected sequentially without response updates. We compare a range of baseline criteria including predictive mutual information, IMSPE, D-optimality, and maximin space-filling designs under both design approaches. The proposed hybrid design criterion is only evaluated under ADE. We also evaluate local complexity as a design criterion alone, providing an ablation study of the two components of the hybrid criterion.
		
		\item We use Gaussian mixture reduction and fast covariance update calculations to make the Bayesian implementation computationally feasible and efficient for the reported experiments.
	\end{itemize}\par
	The remainder of this paper is organized as follows. \Cref{sec:koh-design} reviews the Kennedy \& O’Hagan framework, summarizes its Bayesian implementation, and introduces the design objectives considered in this work together with several existing criteria under the Bayesian setting. \Cref{sec:mi-cx} introduces the predictive MI-based design criterion and our extension, where this criterion is weighted by a local complexity score. We also establish a theoretical connection between predictive MI and IMSPE when the covariance update is small. \Cref{sec:computation} summarizes the computational acceleration techniques used in the implementation. ~\Cref{sec:experiments} reports the numerical experiments on one-, two-, and three-dimensional design problems, including baseline comparisons and an ablation study of the MI and local complexity components. ~\Cref{sec:ConcluDiscuss} concludes with a discussion of the main findings, limitations, and directions for future research.

	\section{Model Framework and Sequential Design}
	\label{sec:koh-design}
	
	In this section, we first provide a brief review of the KOH framework and its Bayesian implementation. Then, we introduce the design objectives together with the sequential design protocols and give an overview of previously employed experimental design criteria under the Bayesian framework. Further details are presented in the following subsections.
	
	\subsection{Review of Kennedy \& O'Hagan Model}
	\label{subsec:KOH}
	
	We adopt the computer model calibration framework proposed by \cite{KennedyOHagan2001}, along with its subsequent adaptations for fully Bayesian inference by \cite{higdon2004combining} and \cite{bayarri2007framework}. The Kennedy \& O'Hagan model, hereafter denoted KOH model, combines simulator runs with physical observations by modeling physical response as the sum of simulator response, model discrepancy term and observational noise. In this work, the design problem is to select additional physical observations based on an initial set of simulator runs and physical observations.\par
	
	Assume that we observe $n$ responses from the physical process and $m$ responses from the simulator, and let $N  = n + m$ be the total number of design points (field + simulator). The controllable design inputs are represented by $\mathbf{X} =(\mathbf{X}^{c\top},\mathbf{X}^{p\top})^\top
	=(\mathbf{x}_1^{c\top},\ldots,\mathbf{x}_m^{c\top},\mathbf{x}_{1}^{p\top},\ldots,\mathbf{x}_{n}^{p\top})^{\top}$ $\in \mathbb{R}^{N \times d}$, with each $\mathbf{x}^c_i = (x^c_{i1}, \ldots, x^c_{id})^\top$ or $\mathbf{x}^p_i = (x^p_{i1}, \ldots, x^p_{id})^\top$ denoting a single input location for the computer model or physical experiment, respectively. Therefore, $\mathbf{X}^c \in \mathbb{R}^{m \times d}$ and $\mathbf{X}^p \in \mathbb{R}^{n \times d}$ contain the set of controllable input settings for the computer and physical experiments, respectively.
	
	In addition to the design inputs which can be controlled in the physical experiment, the computer model typically contains other inputs that need to be calibrated to closely match the actual conditions of the physical process. Assume there are $d_{\theta}$ of these calibration parameters in the computer model and let $\boldsymbol{\theta}=(\theta_1,\ldots,\theta_{d_\theta})^\top \in \mathbb{R}^{d_\theta}$ denote the ``true'' setting of these parameters, i.e., the setting that most closely corresponds to the physical conditions. Moreover, let $\mathbf{T} = (\mathbf{t}_1^{\top},\ldots,\mathbf{t}_m^{\top})^\top \in \mathbb{R}^{m \times d_\theta}$ denote the $m$ settings of the calibration parameters at which the simulator is evaluated. 
	
	The KOH model is formulated as 
	\begin{align*}
		\mathbf{y}^p &= y^c(\mathbf{X}^p, \boldsymbol{\theta}) + \delta(\mathbf{X}^p) + \boldsymbol{\epsilon}, \\[6pt]
		\mathbf{y}^s &= y^c(\mathbf{X}^c, \mathbf{T}),
	\end{align*}
	where $y^{c}(\cdot,\cdot)$ denotes the underlying computer model, $\mathbf{y}^s \in \mathbb{R}^{m}$ represents the simulator outputs evaluated at a set of design points $\mathbf{X}^c$, $\mathbf{y}^p \in\mathbb{R}^{n}$ denotes observations from the physical process, $\delta(\cdot)$ captures the discrepancy between the physical process and the computer model, and the observational errors are assumed to be independent and normally distributed as 
	$\boldsymbol{\epsilon} \sim \mathcal{N}(\mathbf{0}, \sigma^2 I_n)$, 
	where $\boldsymbol{\epsilon} = (\epsilon_1,\ldots,\epsilon_n)^\top \in \mathbb{R}^{n}$.
	
	Both $y^c(\cdot,\cdot)$ and $\delta(\cdot)$ are modeled using GPs \cite{RasmussenWilliams2006},
	\begin{alignat}{2}
		y^c(\mathbf{x},\mathbf{t}) 
		&\sim \mathrm{GP}\!\left(m_c(\mathbf{x},\mathbf{t}),\, 
		k_{\boldsymbol{\phi}_1}\!\bigl((\mathbf{x},\mathbf{t}),(\mathbf{x}',\mathbf{t}')\bigr)\right),
		\label{eq:gp1}
		\\[6pt]
		\delta(\mathbf{x}) 
		&\sim \mathrm{GP}\!\left(m_\delta(\mathbf{x}),\, 
		k_{\boldsymbol{\phi}_2}\!\bigl(\mathbf{x},\mathbf{x}'\bigr)\right).
		\label{eq:gp2}
	\end{alignat}
	Here \(m_c(\cdot,\cdot)\) and \(m_\delta(\cdot)\) denote the mean functions and 
	\(k_{\boldsymbol{\phi}_1}(\cdot,\cdot)\) and \(k_{\boldsymbol{\phi}_2}(\cdot,\cdot)\) are the corresponding covariance kernels parameterized by hyperparameters \(\boldsymbol{\phi}_1\) and \(\boldsymbol{\phi}_2\). For simplicity, we assume zero-mean functions for both GPs, so $m_c(\mathbf{x},\mathbf{t}) = \mathbf{0}$ and $m_{\delta}(\mathbf{x}) = \mathbf{0}$ in our model.
	
	Conditional on $\boldsymbol{\Omega} = (\boldsymbol{\theta}, \boldsymbol{\phi}_1, \boldsymbol{\phi}_2,\sigma^2)$, the outputs can be jointly expressed by a block-structured covariance matrix
	\begin{equation}
		\label{eq:Sigma-mu}
		\begin{pmatrix}
			\mathbf y^p \\[3pt]
			\mathbf y^s
		\end{pmatrix}
		\;\Bigm|\;\boldsymbol\Omega
		\;\sim\;
		\mathcal{N}\!\left(\boldsymbol{\mu}, \boldsymbol{\Sigma}_{oo}\right),
	\end{equation}
	where $\boldsymbol{\mu}=\mathbf{0} \in\mathbb{R}^{N}$ represents the mean and the joint covariance matrix $\boldsymbol{\Sigma}_{oo}\in\mathbb{R}^{N\times N}$ is decomposed as
	\begin{equation}
		\label{eq:Sigma-oo}
		\boldsymbol{\Sigma}_{oo} =
		\underbrace{\begin{pmatrix}
				\mathbf{K}_{c,pp} & \mathbf{K}_{c,ps}\\[3pt]
				\mathbf{K}_{c,sp} & \mathbf{K}_{c,ss}
		\end{pmatrix}}_{\text{computer GP}}
		+
		\underbrace{\begin{pmatrix}
				\mathbf{K}_{\delta,pp} & 0 \\[3pt]
				0 & 0
		\end{pmatrix}}_{\text{discrepancy GP}}
		+
		\underbrace{\begin{pmatrix}
				\sigma^2 \mathbf{I}_n & 0 \\[3pt]
				0 & 0
		\end{pmatrix}}_{\text{observational noise}}.
	\end{equation}
	
	The elements of the matrices are given by
	\begin{eqnarray*}
		\left[\mathbf{K}_{c,pp}\right]_{ij} & = & k_{\boldsymbol{\phi}_1}\!\bigl((\mathbf{x}^p_i,\boldsymbol{\theta}),(\mathbf{x}^p_j,\boldsymbol{\theta})\bigr), \quad i,j = 1,\ldots,n,\\
		\left[\mathbf{K}_{c,ps}\right]_{ij} & = & k_{\boldsymbol{\phi}_1}\!\bigl((\mathbf{x}^p_i,\boldsymbol{\theta}),(\mathbf{x}^c_j,\mathbf{t}_j)\bigr), \quad i = 1,\ldots,n, \: j = 1,\ldots,m,\\
		\left[\mathbf{K}_{c,ss}\right]_{ij} & = & k_{\boldsymbol{\phi}_1}\!\bigl((\mathbf{x}^c_i,\mathbf{t}_i),(\mathbf{x}^c_j,\mathbf{t}_j)\bigr), \quad i,j = 1,\ldots,m,\\
		\left[\mathbf{K}_{\delta,pp}\right]_{ij} & = & k_{\boldsymbol{\phi}_2}\!\bigl(\mathbf{x}^p_i,\mathbf{x}^p_j\bigr), \qquad \qquad \quad i,j = 1,\ldots,n.
	\end{eqnarray*}
	
	
	For prediction at a new set of physical design inputs $\mathbf{X}^* \in \mathbb{R}^{n^* \times d}$, the joint distribution of the future noisy outputs $\mathbf{y}^*$ and observed data $\mathbf{y}=(\mathbf{y}^{p\top},\mathbf{y}^{s\top})^{\top}$ has the covariance matrix
	\begin{equation}
		\label{eq:big_covmatrix}
		\mathrm{Cov}\!\left[
		\begin{pmatrix}
			\mathbf y^* \\
			\mathbf y
		\end{pmatrix}
		\right]
		=
		\begin{pmatrix}
			\boldsymbol{\Sigma}_{**} & \boldsymbol{\Sigma}_{*o}\\[3pt]
			\boldsymbol{\Sigma}_{o*} & \boldsymbol{\Sigma}_{oo}
		\end{pmatrix}\in\mathbb{R}^{(n^*+N)\times(n^*+N)}.
	\end{equation}
	
	The submatrix $\boldsymbol{\Sigma}_{**}$ is composed of the elements
	\begin{equation*}
		\left[\boldsymbol{\Sigma}_{**}\right]_{ij} =  k_{\boldsymbol{\phi}_1}\!\bigl((\mathbf{x}^*_i,\boldsymbol{\theta}),(\mathbf{x}^*_j,\boldsymbol{\theta})\bigr) + k_{\boldsymbol{\phi}_2}\!\bigl(\mathbf{x}^*_i,\mathbf{x}^*_j\bigr) + \sigma^2 \cdot \mathds{1}_{i=j}, \quad i,j = 1,\ldots,n^*, 
	\end{equation*}
	where $\mathds{1}_{i=j}$ denotes the Kronecker delta function which is $1$ if $i = j$ and $0$ otherwise.
	
	The structure of  $\boldsymbol{\Sigma}_{*o}$ is
	\begin{equation*}
		\boldsymbol{\Sigma}_{*o} = \begin{pmatrix} \boldsymbol{\Sigma}_{*o}^p & \boldsymbol{\Sigma}_{*o}^s \end{pmatrix},
	\end{equation*}
	with the elements of the submatrices of $\boldsymbol{\Sigma}_{*o}$ given by
	\begin{eqnarray*}
		\left[\boldsymbol{\Sigma}_{*o}^p\right]_{ij} &=& 
		k_{\boldsymbol{\phi}_1}\!\bigl((\mathbf{x}^*_i,\boldsymbol{\theta}),(\mathbf{x}^p_j,\boldsymbol{\theta})\bigr) + k_{\boldsymbol{\phi}_2}\!\bigl(\mathbf{x}^*_i,\mathbf{x}^p_j\bigr), \quad i = 1,\ldots,n^*,\:  j = 1,\ldots,n,\\
		\left[\boldsymbol{\Sigma}_{*o}^s\right]_{ij} &=& k_{\boldsymbol{\phi}_1}\!\bigl((\mathbf{x}^*_i,\boldsymbol{\theta}),(\mathbf{x}^c_j,\mathbf{t}_j)\bigr), \qquad \qquad \qquad \quad i = 1,\ldots,n^*, \: j = 1,\ldots,m.
		\\
	\end{eqnarray*}
		Also,  $\boldsymbol{\Sigma}_{o*}=\boldsymbol{\Sigma}_{*o}^\top$.
		
		Consequently, the predictive distribution of $\mathbf{y}^{*}$ is
		\begin{equation}
			\label{eq:predictive-distribution}
			\mathbf y^* \,\bigm|\, \mathbf y,\boldsymbol\Omega
			\;\sim\;
			\mathcal N\!\left(
			\boldsymbol{\mu}^*,\,\boldsymbol{\Sigma}^*\right),
		\end{equation}
		with predictive mean $\boldsymbol{\mu}^* \in \mathbb{R}^{n^*}$ and covariance $\boldsymbol{\Sigma}^* \in \mathbb{R}^{n^* \times n^*}$ given by
		\begin{equation}
			\label{eq:predictive-mean-cov}
			\boldsymbol{\mu}^* = \boldsymbol{\Sigma}_{*o}\,\boldsymbol{\Sigma}_{oo}^{-1} \,
			\mathbf{y}, \qquad
			\boldsymbol{\Sigma}^* = \boldsymbol{\Sigma}_{**}
			- \boldsymbol{\Sigma}_{*o}\,\boldsymbol{\Sigma}_{oo}^{-1}\boldsymbol{\Sigma}_{o*}.
		\end{equation}
		
		Note that the predictive mean and covariance depend on the parameters $\boldsymbol{\Omega}$, the inputs of interest $\mathbf{X}^*$, as well as on the physical design inputs $\mathbf{X}^p$ and the computer simulation inputs $\mathbf{X}^c$. Only the predictive mean also depends on the observed data $\mathbf{y}$.

\subsection{Design Problem and Design Protocols}
\label{subsec:design-protocols}

Optimal experimental design typically relies on constructing an appropriate design criterion or objective function. As discussed earlier, the objective of this study is to select additional inputs for physical experiments that can improve the predictive performance of the KOH model accounting for the uncertainty in the simulator, discrepancy, calibration parameters, and observation noise. Compared to computer experiments, physical experiments are typically constrained by a limited set of candidate points and a tight experimental budget.

We formulate the design problem on a finite candidate set
\[
\mathcal{D}_{\mathrm{cand}}
=
\{\boldsymbol{\xi}_1,\ldots,\boldsymbol{\xi}_M\},
\]
where each $\boldsymbol{\xi}_i\in\mathbb{R}^d$ represents a feasible location for a physical experiment. Let $B$ be the total experimental budget. Before round $b$, the already selected subset is
\[
\mathcal{D}_{\text{sel}}^{(b-1)}
=
\{\boldsymbol{\xi}_{(1)}, \ldots, \boldsymbol{\xi}_{(b-1)}\},
\]
and the next design point $\boldsymbol{\xi}_{(b)}$ is chosen from
$\mathcal{D}_{\text{cand}} \setminus \mathcal{D}_{\text{sel}}^{(b-1)}$. The selected design points are incorporated into the KOH model by appending them to the matrix $\mathbf{X}^p$. Therefore, the matrix $\boldsymbol{\Sigma}_{oo}$ and the predictive parameters in \cref{eq:predictive-mean-cov} depend on the selected design points.
The predictive performance of the model is evaluated based on the predictive results $\mathbf{y}^*$ at locations
$\mathcal{X}_{\text{pred}} = \mathbf{X}^* \in \mathbb{R}^{n^* \times d}$.

In this study, we consider two protocols for sequential Bayesian experimental design: sequential (offline) design of experiments (SDE) and adaptive (online) design of experiments (ADE). The SDE represents the offline greedy approach, seeking the best location at each round without observing the new physical observations. Consequently, the model posterior remains unchanged throughout the design process. In contrast, ADE conducts real physical experiments immediately after each selected input $\boldsymbol{\xi}_{(b)}$. The corresponding observations then update the posterior of the KOH model. In this way, ADE can incorporate newly observed physical responses to adapt later design decisions.

\Cref{alg:seq-adaptive-design} clearly illustrates the distinction between the two strategies.

\begin{algorithm}[htbp]
\caption{SDE and ADE under the KOH framework}
\label{alg:seq-adaptive-design}
\begin{algorithmic}
\REQUIRE Budget $B$; candidates
$\mathcal{D}_{\text{cand}}=\{\boldsymbol{\xi}_1,\ldots,\boldsymbol{\xi}_M\}$;
prediction set $\mathcal{X}_{\text{pred}}=\mathbf{X}^*$;
initial data $\{(\mathbf{X}^p,\mathbf{X}^c,\mathbf{y})\}$;
mode $\in\{\texttt{SDE},\texttt{ADE}\}$;
criterion or objective function $\mathsf{Crit}(\cdot)$
\ENSURE
$\mathcal{D}_{\text{sel}}^{(B)}
=
\{\boldsymbol{\xi}_{(1)},\ldots,\boldsymbol{\xi}_{(B)}\}$

\STATE $\mathcal{D}_{\text{sel}}^{(0)}\gets\emptyset$; fit KOH posterior

\FOR{$b=1$ \TO $B$}
\STATE
\[
\boldsymbol{\xi}_{(b)} \gets
\arg\max(\min)_{\boldsymbol{\xi}\in
\mathcal{D}_{\text{cand}}\setminus
\mathcal{D}_{\text{sel}}^{(b-1)}}
\mathsf{Crit}\!\bigl(
\boldsymbol{\xi};\,\text{current posterior},\,
\mathcal{X}_{\text{pred}}
\bigr)
\]
\IF{$\texttt{mode}=\texttt{SDE}$}
\STATE No physical experiment; keep current KOH posterior unchanged
\ELSIF{$\texttt{mode}=\texttt{ADE}$}
\STATE Run experiment at $\boldsymbol{\xi}_{(b)}$ to get
$\mathbf{y}^{\text{new}}_{(b)}$; augment data; update KOH posterior
\ENDIF
\STATE
$\mathcal{D}_{\text{sel}}^{(b)}
\gets
\mathcal{D}_{\text{sel}}^{(b-1)}
\cup\{\boldsymbol{\xi}_{(b)}\}$
\ENDFOR

\RETURN $\mathcal{D}_{\text{sel}}^{(B)}$ and final KOH posterior
\end{algorithmic}
\end{algorithm}

A globally optimal batch design which selects all $B$ points simultaneously in principle performs at least as well as the multi-round greedy approach and often yields better performance \cite{MuellerPoetscher1990}. However, the combinatorial search space grows rapidly with \(\binom{M}{B}\), making true batch design computationally infeasible. Greedy design has been shown to provide approximation guarantees for nonnegative monotone submodular objectives under a cardinality constraint \cite{NemhauserWolseyFisher1978} and achieved strong results for GP-based sensor placement in \cite{KrauseSinghGuestrin2008}. The Bayesian setting in this study does not strictly satisfy the monotonicity and submodularity conditions. Nevertheless, when these properties hold approximately, greedy selection can still perform well in practice.

\subsection{Benchmark Design Criteria}
\label{subsec:benchmark-criteria}
We first review the benchmark design criteria used for comparison in the subsequent numerical examples. Three representative design criteria are commonly used: IMSPE, which aims to reduce predictive uncertainty; D-optimality, which aims to maximize information about the calibration parameters; and the maximin space-filling criterion, which improves the geometric coverage of the design space. The model-related criteria are evaluated under the KOH framework introduced in Section \ref{subsec:KOH} by integrating over the posterior distribution $p(\boldsymbol{\Omega}\mid \mathbf y)$. At round $b$, the vector $\mathbf y$ denotes all observations available before evaluating the next candidate. In ADE, this includes the initial observations and the physical responses collected in the previous $b-1$ rounds. For a candidate $\boldsymbol{\xi}$, let
\[
\mathcal{D}_{\rm sel}^{(b)}(\boldsymbol{\xi})
=
\mathcal{D}_{\rm sel}^{(b-1)}\cup\{\boldsymbol{\xi}\}
\]
denote the candidate-augmented design set.
\paragraph{Integrated mean squared prediction error}
The integrated mean squared prediction error (IMSPE) criterion is derived from I-optimality \cite{Studden1977Optimal}. It was initially proposed in \cite{Sacks1989} and detailed formulations under the KOH framework were considered in \cite{Leatherman2017combined, Koermer2024}.\par
Using the predictive covariance matrix in \cref{eq:predictive-mean-cov}, the fully Bayesian IMSPE criterion at round $b$ is
\begin{equation}
\label{eq:imspe-bayes}
\mathrm{IMSPE}(\boldsymbol{\xi})
=
\mathbb{E}_{\boldsymbol{\Omega}\sim p(\boldsymbol{\Omega}\mid \mathbf{y})}
\!\left[
\frac{1}{n^*}
\operatorname{tr}\!\Bigl(
\boldsymbol{\Sigma}^*(\mathbf{X}^*;\, \boldsymbol{\Omega},\,
\mathcal{D}_{\rm sel}^{(b)}(\boldsymbol{\xi}))
\Bigr)
\right].
\end{equation}
Here $\operatorname{tr}(\cdot)$ denotes the matrix trace operator. The next experimental point is selected by minimizing \cref{eq:imspe-bayes} over the remaining candidate points
\begin{equation*}
\label{eq:imspe-selection}
\boldsymbol{\xi}_{(b)}
=
\arg\min_{\boldsymbol{\xi}\in
\mathcal{D}_{\rm cand}\setminus\mathcal{D}_{\rm sel}^{(b-1)}}
\mathrm{IMSPE}(\boldsymbol{\xi}).
\end{equation*}
In ADE, both $\mathbf{y}$ and $p(\boldsymbol{\Omega}\mid \mathbf{y})$ are updated at each round $b$ by the new field observation $\mathbf{y}^{\text{new}}_{(b)}$.

\paragraph{D-optimality}
D-optimality targets information about the calibration parameters rather than predictive uncertainty directly. Let
$\boldsymbol{\Omega}_{-\theta}
=
(\boldsymbol{\phi}_1,\boldsymbol{\phi}_2,\sigma^2)$
denote the components of $\boldsymbol{\Omega}$ other than
$\boldsymbol{\theta}$. Following \cite{Krishna2021,Diao2022}, conditional on $\boldsymbol{\Omega}_{-\theta}$,
the elements of the Fisher information matrix (FIM) for $\boldsymbol{\theta}$ are given by
\begin{equation}
\label{eq:FIM-theta-given-Omega}
\left[I\!
\left(
\boldsymbol{\theta}; \, \boldsymbol{\Omega}_{-\theta}, \mathcal{D}_{\rm sel}^{(b)}(\boldsymbol{\xi})
\right)\right]_{ij}
=
\frac{\partial \boldsymbol{\mu}_o^\top}{\partial \theta_i}
\boldsymbol{\Sigma}_{oo}^{-1}
\frac{\partial \boldsymbol{\mu}_o}{\partial \theta_j}
+
\frac{1}{2}
\operatorname{tr}\!\left(
\boldsymbol{\Sigma}_{oo}^{-1}
\frac{\partial \boldsymbol{\Sigma}_{oo}}{\partial \theta_i}
\boldsymbol{\Sigma}_{oo}^{-1}
\frac{\partial \boldsymbol{\Sigma}_{oo}}{\partial \theta_j}
\right),
\: i,j=1,\ldots,d_\theta.
\end{equation}
Here, $\boldsymbol{\Sigma}_{oo}$ is defined in
\cref{eq:Sigma-oo} and $\boldsymbol{\mu}_o$ is equal to $\mathbf{y}$ at points with observed data. At these points, the partial derivatives in $\partial \boldsymbol{\mu}_o / \partial \theta_j$ are $0$. However, no data are observed at the current candidate design point and at the previously selected design points $\mathcal{D}_{\rm sel}^{(b-1)}$ in SDE after at least one round. For these points, the corresponding elements of $\boldsymbol{\mu}_o$ are computed as the predictive mean \cref{eq:predictive-mean-cov} at these points given all the observed points. 

The criterion we consider is the log determinant of the FIM. To account for posterior uncertainty in all model
parameters, we average the FIM over the posterior distribution $p(\boldsymbol{\Omega}\mid\mathbf{y})$.
The next point is therefore chosen by maximizing the expected criterion:
\begin{equation}
\label{eq:D-argmax-one-step}
\boldsymbol{\xi}_{(b)}
=
\arg\max_{\boldsymbol{\xi}\in
\mathcal{D}_{\rm cand}\setminus \mathcal{D}_{\rm sel}^{(b-1)}}
\;
\mathbb{E}_{\boldsymbol{\Omega}\sim
	p(\boldsymbol{\Omega}\mid\mathbf{y})} \left[
	\log\det I\!
	\left(
	\boldsymbol{\theta}; \, \boldsymbol{\Omega}_{-\theta}, \mathcal{D}_{\rm sel}^{(b)}(\boldsymbol{\xi})
	\right)
\right].
\end{equation}

\paragraph{Maximin distance}
Under circumstances where prior information of the model or expert knowledge is not available, a space‐filling design is commonly adopted \cite{Johnson1990, PronzatoMuller2012, lin2015latin}. The maximin distance criterion selects the design point which has the largest minimum Euclidean distance to the already selected points. The design input is searched in a greedy manner:
\begin{equation*}\label{eq:maximin-adaptive}
\boldsymbol{\xi}_{(b)}
=
\arg\!\max_{\boldsymbol{\xi}\in \mathcal{D}_{\mathrm{cand}}\setminus \mathcal{D}_{\mathrm{sel}}^{(b-1)}}
\;
\min_{\boldsymbol{\xi}_{(j)}\in \mathcal{D}_{\mathrm{sel}}^{(b-1)}(1\le j\le b-1)}
\bigl\lVert \boldsymbol{\xi} - \boldsymbol{\xi}_{(j)} \rVert_2.
\end{equation*}
Since the maximin distance criterion does not depend on the KOH posterior or the observed responses, it yields the same design under SDE and ADE.

\section{Prediction-Oriented Mutual Information Design}
\label{sec:mi-cx}

In this section, we introduce the main prediction-oriented design criteria used in this work. First, we define a predictive mutual information (MI) criterion that directs the selection of physical experimental inputs by maximizing the expected information gained about $\mathbf{y}^*$ while marginalizing over the model parameters $\boldsymbol{\Omega}$. We then compare predictive MI with IMSPE to clarify how the two criteria treat posterior predictive uncertainty differently. Finally, motivated by shape changes of the predictive surface after new observations are collected in adaptive design, we incorporate a local complexity score that captures nonlinear structure in the posterior predictive mean. We then construct a complexity-weighted predictive MI criterion, denoted MI--CX.

\subsection{Predictive Mutual Information}
\label{subsec:predictive-mi}
From \cite{Lindley1972}, the utility function of BED has the following general form:
\begin{equation}
\label{eq:general-utility}
U(\boldsymbol{\xi})
=
\iint
u\!\left(
\boldsymbol{\xi},\mathbf{y}^{\mathrm{new}}(\boldsymbol{\xi});
\boldsymbol{\Omega}
\right)
p\!\left(
\mathbf{y}^{\mathrm{new}}(\boldsymbol{\xi}),\boldsymbol{\Omega}
\mid
\boldsymbol{\xi}
\right)
\,d\boldsymbol{\Omega}\,d\mathbf{y}^{\rm new}.
\end{equation}
Here $U(\boldsymbol{\xi})$ is the expected utility of the conditional utility function $u\!\left(\boldsymbol{\xi},\mathbf{y}^{\mathrm{new}}(\boldsymbol{\xi});\boldsymbol{\Omega}\right)$ and $\mathbf{y}^{\mathrm{new}}(\boldsymbol{\xi})$ represents the (as yet unobserved) outcome of an experiment conducted at $\boldsymbol{\xi}$. Since the response is unknown before performing the experiment, the expectation is integrated over both the uncertainty from $\boldsymbol{\Omega}$ and $\mathbf{y}^{\mathrm{new}}(\boldsymbol{\xi})$.\par
At round $b$, we use the mutual information between the potential physical observation at $\boldsymbol{\xi}$ and the predictive quantity $\mathbf y^*$ evaluated at the prediction inputs $\mathcal{X}_{\rm pred}=\mathbf X^*$:
\begin{equation}
\label{eq:predictive-mi}
\mathrm{MI}(\boldsymbol{\xi})
:=
I\!\left(
\mathbf y^*;
\mathbf{y}^{\mathrm{new}}_{(b)}(\boldsymbol{\xi})
\mid
\mathbf y,\mathcal{D}_{\rm sel}^{(b)}(\boldsymbol{\xi})
\right)
\end{equation}
which is defined as \cite{Ryan2016}
\begin{equation}
\label{eq:predictive-mi-integral}
\begin{aligned}
\mathrm{MI}(\boldsymbol{\xi})
&=
\iint
p\!\left(
\mathbf{y}^*,\mathbf{y}^{\mathrm{new}}_{(b)}(\boldsymbol{\xi})
\mid
\mathbf{y},\mathcal{D}_{\rm sel}^{(b)}(\boldsymbol{\xi})
\right)
\\
&\quad \times
\log
\frac{
p\!\left(
\mathbf{y}^*,\mathbf{y}^{\mathrm{new}}_{(b)}(\boldsymbol{\xi})
\mid
\mathbf{y},\mathcal{D}_{\rm sel}^{(b)}(\boldsymbol{\xi})
\right)
}{
p\!\left(
\mathbf{y}^*
\mid
\mathbf{y},\mathcal{D}_{\rm sel}^{(b)}(\boldsymbol{\xi})
\right)
p\!\left(
\mathbf{y}^{\mathrm{new}}_{(b)}(\boldsymbol{\xi})
\mid
\mathbf{y},\mathcal{D}_{\rm sel}^{(b)}(\boldsymbol{\xi})
\right)
}
\,d\mathbf{y}^*\,d\mathbf{y}^{\rm new}.
\end{aligned}
\end{equation}
Corresponding to \cref{eq:general-utility}, the utility function for predictive mutual information gain is defined as the log ratio between the posterior and prior predictive distributions of the quantity of interest:
\[
u\!\left(\boldsymbol{\xi},\mathbf{y}^{\mathrm{new}}_{(b)}(\boldsymbol{\xi});\mathbf{y}^*\right)
=
\log
\frac{
p\!\left(
\mathbf{y}^*
\mid
\mathbf{y}^{\mathrm{new}}_{(b)}(\boldsymbol{\xi}),\mathbf{y},
\mathcal{D}_{\rm sel}^{(b)}(\boldsymbol{\xi})
\right)
}{
p\!\left(
\mathbf{y}^*
\mid
\mathbf{y},
\mathcal{D}_{\rm sel}^{(b-1)}
\right)
}.
\]
This criterion can be interpreted as the expected Kullback--Leibler divergence between the posterior predictive and prior predictive distributions of future observations $\mathbf{y}^*$ updated by $\mathbf{y}^{\mathrm{new}}_{(b)}(\boldsymbol{\xi})$ \cite{KullbackLeibler1951,ChalonerVerdinelli1995}.\par
Each marginal likelihood term in \cref{eq:predictive-mi-integral} is a posterior predictive distribution marginalized over $\boldsymbol{\Omega}$. For example, 
\begin{equation}
\label{eq:mi-marginalization}
p\!\left(
\mathbf y^*,\mathbf{y}^{\mathrm{new}}_{(b)}(\boldsymbol{\xi})
\mid
\mathbf y,\mathcal{D}_{\rm sel}^{(b)}(\boldsymbol{\xi})
\right)
=
\int
p\!\left(
\mathbf y^*,\mathbf{y}^{\mathrm{new}}_{(b)}(\boldsymbol{\xi})
\mid
\mathbf y,\boldsymbol{\Omega},
\mathcal{D}_{\rm sel}^{(b)}(\boldsymbol{\xi})
\right)
p(\boldsymbol{\Omega}\mid\mathbf y)
\,d\boldsymbol{\Omega}.
\end{equation}
The design input can be selected by maximizing the expected utility
\begin{equation*}
\label{eq:mi-selection}
\boldsymbol{\xi}_{(b)}
=
\arg\max_{\boldsymbol{\xi}\in
\mathcal{D}_{\rm cand}\setminus\mathcal{D}_{\rm sel}^{(b-1)}}
\mathrm{MI}(\boldsymbol{\xi}).
\end{equation*}
In SDE, the posterior distribution $p(\boldsymbol{\Omega}\mid\mathbf y)$ in \cref{eq:mi-marginalization} is fixed, whereas in ADE, the physical response $\mathbf{y}^{\mathrm{new}}_{(b)}$ observed at $\boldsymbol{\xi}_{(b)}$ is added to the data $\mathbf{y}$, and the KOH posterior is updated before the next round.

\subsection{Connection Between Predictive MI and IMSPE}
\label{subsec:mi-imspe}

The IMSPE and predictive MI criteria both aim to reduce predictive uncertainty directly, but they characterize changes in the predictive covariance differently. IMSPE depends on the trace of the predictive covariance matrix, whereas predictive MI considers the full covariance structure. To clarify this relationship, we state the following result conditional on the parameters $\boldsymbol{\Omega}$.\par
Let $\boldsymbol{\Sigma}^*_{(b-1)}$ denote the pre-experiment predictive covariance matrix at round $b-1$. Similarly, $\boldsymbol{\Sigma}^*_{(b)}$  denotes the corresponding predictive covariance matrix after adding the design point $\boldsymbol{\xi}_{(b)}$. The following theorem summarizes the local relationship between the two criteria.
\begin{theorem}[Local asymptotic relationship between MI and IMSPE]
\label{thm:MI-IMSPE}
Assume that (i) conditional on $\boldsymbol{\Omega}$,
$(\mathbf y^*,\mathbf{y}^{\mathrm{new}}_{(b)}(\boldsymbol{\xi}))$ is jointly Gaussian, (ii) $\boldsymbol{\Sigma}^*_{(b-1)}$ is symmetric positive definite
with eigenvalues 
$0<\lambda_{\min}\le\lambda_{\max}<\infty$. Define the covariance reduction and IMSPE reduction at round $b$ as
\[
\Delta\boldsymbol{\Sigma}^*(\boldsymbol{\xi}_{(b)})
:=
\boldsymbol{\Sigma}^*_{(b-1)}
-
\boldsymbol{\Sigma}^*_{(b)}
\succeq \mathbf 0,
\qquad
\Delta\mathrm{IMSPE}(\boldsymbol{\xi}_{(b)} \mid \boldsymbol{\Omega
})
:=
\frac{1}{n^*}\,
\operatorname{tr}\big(
\Delta\boldsymbol{\Sigma}^*(\boldsymbol{\xi}_{(b)})
\big).
\]
As the sequential design rounds proceed, the following small-update regime holds:
\[
\big\|
(\boldsymbol{\Sigma}^*_{(b-1)})^{-1/2}
\,
\Delta\boldsymbol{\Sigma}^*(\boldsymbol{\xi}_{(b)})
\,
(\boldsymbol{\Sigma}^*_{(b-1)})^{-1/2}
\big\|_2 \;\longrightarrow\; 0.
\]
Then there exists a finite constant $C_{b-1}$, depending on
$\boldsymbol{\Sigma}^*_{(b-1)}$, such that
\[
\frac{U(\boldsymbol{\xi}_{(b)}\mid \boldsymbol{\Omega
})}{\Delta\mathrm{IMSPE}(\boldsymbol{\xi}_{(b)}\mid \boldsymbol{\Omega
}) }
\;\longrightarrow\;
C_{b-1}
\quad\text{as}\quad
\Delta\boldsymbol{\Sigma}^*(\boldsymbol{\xi}_{(b)})\to \mathbf 0,
\]
and this constant is bounded in terms of the eigenvalues of
$\boldsymbol{\Sigma}^*_{(b-1)}$:
\[
\frac{n^*}{2\,\lambda_{\max}\!\big(\boldsymbol{\Sigma}^*_{(b-1)}\big)}
\;\le\;
C_{b-1}
\;\le\;
\frac{n^*}{2\,\lambda_{\min}\!\big(\boldsymbol{\Sigma}^*_{(b-1)}\big)}.
\]
\end{theorem}

\begin{proof}
See Supplementary Material.
\end{proof}

\Cref{thm:MI-IMSPE} establishes a link between the two criteria. As the rounds proceed sequentially and enter a small-update regime, maximizing mutual information becomes locally equivalent to maximizing the IMSPE reduction. The ratio between the two criteria converges to a factor determined by the current predictive covariance. From another perspective, the predictive MI criterion can be interpreted as a spectrally weighted version of IMSPE. While IMSPE measures the total reduction in marginal predictive variance, predictive MI additionally accounts for how this reduction is distributed across the principal directions of the predictive covariance.\par
The theorem is stated conditional on $\boldsymbol{\Omega}$, while the Bayesian version criteria used in this paper marginalize over $p(\boldsymbol{\Omega}\mid\mathbf y)$. As the sequential design progresses, posterior changes in $\boldsymbol{\Omega}$ are expected to become smaller. Thus, the conditional relationship in \cref{thm:MI-IMSPE} also provides practical insight into the local behavior of the two criteria in a Bayesian setting.

\subsection{Local Complexity of the Posterior Predictive Mean}
\label{subsec:local-complexity}
Predictive MI can help obtain the points that reduce global predictive uncertainty most, but it may also have the drawback of neglecting critical model structure such as extrema or inflection regions. In adaptive design, newly observed physical responses can change the posterior predictive surface after each round. It is therefore worthwhile to measure how strongly the predictive mean varies around each candidate by quantifying its departure from a local linear fit. This idea is related to local polynomial approximation and curvature-based sequential design \cite{FanGijbels1996, Crombecq2011, MuDaiXu2017}.\par
Let $\mu_b(\mathbf x)$ denote the posterior predictive mean of the physical response at an arbitrary input $\mathbf x\in\mathcal X$ at round $b$, after marginalizing over the posterior uncertainty in $\boldsymbol{\Omega}$. For each candidate input $\boldsymbol{\xi}\in\mathcal{D}_{\rm cand}\setminus\mathcal{D}_{\rm sel}^{(b-1)}$, we construct a local stencil in the normalized input space,
\[
\boldsymbol{\xi}+h\mathbf v_1,\ldots,\boldsymbol{\xi}+h\mathbf v_L,
\]
where $h>0$ is a local stencil radius and $\mathbf v_1,\ldots,\mathbf v_L$ are fixed local directions. The corresponding vector of posterior predictive means is
\begin{equation}
\label{eq:local-predictive-mean-vector}
\boldsymbol{\mu}_{b,\xi}
=
\left(
\mu_b(\boldsymbol{\xi}+h\mathbf v_1),
\ldots,
\mu_b(\boldsymbol{\xi}+h\mathbf v_L)
\right)^\top .
\end{equation}

To measure local departure from a linear approximation, we define a local regression matrix
\[
\mathbf A_{\xi}
=
\begin{pmatrix}
1 & (h\mathbf v_1)^\top\\
\vdots & \vdots\\
1 & (h\mathbf v_L)^\top
\end{pmatrix}
\in\mathbb{R}^{L\times(d+1)}.
\]
Then we construct a projection matrix 
\[
\mathbf P_{\xi}
=
\mathbf A_{\xi}
(\mathbf A_{\xi}^\top\mathbf A_{\xi})^{-1}
\mathbf A_{\xi}^\top.
\]
The residual from the local linear fit is
\[
\mathbf r_{\xi}
=
(\mathbf I_L-\mathbf P_{\xi})\boldsymbol{\mu}_{b,\xi}.
\]
We define the raw local complexity score as
\begin{equation}
\label{eq:cx-raw}
\mathrm{CX}_{\rm raw}(\boldsymbol{\xi})
=
\frac{1}{L}
\left\lVert
\mathbf r_{\xi}
\right\rVert_2^2 .
\end{equation}
A larger value of $\mathrm{CX}_{\rm raw}(\boldsymbol{\xi})$ indicates a greater local complexity, as the posterior predictive mean departs farther from a local linear fit around $\boldsymbol{\xi}$.\par
Suppose that $\mu_b$ is twice continuously differentiable near $\boldsymbol{\xi}$. Then a componentwise Taylor expansion of \cref{eq:local-predictive-mean-vector} gives 
\[
\mu_b(\boldsymbol{\xi}+h\mathbf v_j)
=
\mu_b(\boldsymbol{\xi})
+
h\nabla\mu_b(\boldsymbol{\xi})^\top\mathbf v_j
+
\frac{h^2}{2}\mathbf v_j^\top
\mathbf H_{\mu_b}(\boldsymbol{\xi})\mathbf v_j
+
O(h^3),
\]
where $\mathbf H_{\mu_b}(\boldsymbol{\xi})$ is the Hessian matrix. Because the projection removes the constant and first-order terms, the residual is driven by the quadratic component. Consequently,
\[
\mathrm{CX}_{\rm raw}(\boldsymbol{\xi})
=
\frac{h^4}{4L}
\left\|
(\mathbf I_L-\mathbf P_{\xi})\mathbf q_{\xi}
\right\|_2^2
+
O(h^5),
\]
where $(\mathbf q_{\xi})_j
=
\mathbf v_j^\top\mathbf H_{\mu_b}(\boldsymbol{\xi})\mathbf v_j$.
Therefore, the score evaluates second-order deviation from local linearity.\par
We use a rank-normalized version of the raw score in \cref{eq:cx-raw} so that it can be combined more appropriately with predictive MI. Let $M_b=|\mathcal{D}_{\rm cand}\setminus\mathcal{D}_{\rm sel}^{(b-1)}|$ be the number of remaining candidates, then the normalized complexity score is given by
\begin{equation}
\label{eq:cx-rank}
\mathrm{CX}(\boldsymbol{\xi})
=
\frac{
\operatorname{rank}
\left\{\mathrm{CX}_{\rm raw}(\boldsymbol{\xi})\right\}-1
}{
M_b-1
},
\qquad
\boldsymbol{\xi}\in\mathcal{D}_{\rm cand}\setminus\mathcal{D}_{\rm sel}^{(b-1)} .
\end{equation}
After this, $\mathrm{CX}(\boldsymbol{\xi})\in[0,1]$.

\subsection{The MI--CX Criterion}
\label{subsec:mi-cx-criterion}
Rather than being used alone, the local complexity score is used to weight predictive MI. The hybrid criterion favors candidate points that are both informative and located in regions where the current posterior predictive mean exhibits strong local nonlinearity. Define the combined criterion as
\begin{equation}
\label{eq:mi-cx}
\mathrm{MI\mbox{-}CX}(\boldsymbol{\xi})
=
\mathrm{MI}(\boldsymbol{\xi})\,\mathrm{CX}(\boldsymbol{\xi}),
\end{equation}
where $\mathrm{MI}(\boldsymbol{\xi})$ is defined in Section \ref{subsec:predictive-mi} and $\mathrm{CX}(\boldsymbol{\xi})$ is the normalized local complexity score in \cref{eq:cx-rank}. The next experimental input is selected by
\begin{equation}
\label{eq:mi-cx-selection}
\boldsymbol{\xi}_{(b)}
=
\arg\max_{\boldsymbol{\xi}\in\mathcal{D}_{\rm cand}\setminus\mathcal{D}_{\rm sel}^{(b-1)}}
\mathrm{MI\mbox{-}CX}(\boldsymbol{\xi}).
\end{equation}
The multiplicative form in \cref{eq:mi-cx} comprehensively accounts for the global predictive uncertainty as well as the local structural pattern of the predictive surface. We use this MI--CX criterion in ADE because the posterior predictive mean changes at each round. $\mathrm{CX}(\boldsymbol{\xi})$ among all candidate points is updated accordingly at each design round. In contrast, under SDE, $\mathrm{CX}(\boldsymbol{\xi})$ is unchanged through the whole process and can be misleading for later selections, especially when the initial predictive mean is inaccurate. 

\section{Computational Strategies}
\label{sec:computation}
This section summarizes the strategies used to make the experimental design computationally feasible and reports their computational cost. The computational burden of the predictive MI and MI--CX is substantial because, at every design round, parameter uncertainty and the likelihood must be integrated over for each candidate. Specifically, we sample the predictive response $\mathbf y^*$ and $\mathbf{y}^{\mathrm{new}}_{(b)}(\boldsymbol{\xi})$ as well as the model parameters $\boldsymbol{\Omega}$ in order to account for all sources of uncertainty. This leads naturally to a nested Monte Carlo (NMC) estimator \cite{Rainforth2018}.

\subsection{Nested Monte Carlo Estimation of Predictive MI}
\label{subsec:mc-mi}
Let $S\in\mathbb{N}$ be the number of outer samples and $J\in\mathbb{N}$ the number of inner samples. For each outer index \(s=1,\dots,S\), we first draw $\boldsymbol{\Omega}^{(0)}_s \sim p(\boldsymbol{\Omega}\mid \mathbf{y})$ and then sample $(\mathbf{y}^*_s, {\mathbf{y}^{\mathrm{new}}_{(b)}(\boldsymbol{\xi})}_s \sim p(\mathbf{y}^*, \mathbf{y}^{\mathrm{new}}_{(b)}(\boldsymbol{\xi}) \mid \mathbf{y},\boldsymbol{\Omega}^{(0)}_s)$. In the inner layer, we independently draw \(\boldsymbol{\Omega}^{(1:J)}_s=\{\boldsymbol{\Omega}^{(j)}_s\}_{j=1}^J\) i.i.d. from \(p(\boldsymbol{\Omega}\mid \mathbf{y})\). The marginal densities are approximated by Gaussian mixture. For example, the joint predictive density of
$\mathbf y^*$ and $\mathbf{y}^{\mathrm{new}}_{(b)}(\boldsymbol{\xi})$ is
\begin{equation}
\label{eq:joint-predictive-mixture}
\widehat{p}^{\,\mathrm{joint}}_{J}
\!\left(
\mathbf{y}^*_s,
{\mathbf{y}^{\mathrm{new}}_{(b)}(\boldsymbol{\xi})}_s
\right)
\approx
\frac{1}{J}\sum_{j=1}^{J}
p\!\left(
\mathbf{y}^*_s,
{\mathbf{y}^{\mathrm{new}}_{(b)}(\boldsymbol{\xi})}_s
\,\middle|\,
\mathbf{y},
\boldsymbol{\Omega}^{(j)}_s,
\mathcal{D}_{\mathrm{sel}}^{(b)}(\boldsymbol{\xi})
\right).
\end{equation}
The NMC estimator for the computation of the MI in \cref{eq:predictive-mi-integral} is 
\begin{equation}\label{eq:NMC_estimator_MI}
	\widehat{U}_{S,J}(\boldsymbol{\xi})
	=\frac{1}{S}\sum_{s=1}^{S}
	\Bigg[
	\log \widehat{p}^{\,\mathrm{joint}}_{J}\!\big(\mathbf{y}^*_s,{\mathbf{y}^{\mathrm{new}}_{(b)}(\boldsymbol{\xi})}_s\big)
	-\log \widehat{p}^{\,*}_{J}\!\big(\mathbf{y}^*_s\big)
	-\log \widehat{p}^{\,\mathrm{new}}_{J}\!\big({\mathbf{y}^{\mathrm{new}}_{(b)}(\boldsymbol{\xi})}_s\big)
	\Bigg].
\end{equation}
The convergence of the NMC estimator has been studied in
\cite{HongJuneja2010,Rainforth2018}. For the convergence properties to hold, some smoothness properties of the transformation function, which is $f(\cdot)=\log(\cdot)$ in this case, need to be satisfied.
\begin{lemma}[Smoothness and bounded derivatives of the outer-layer function]
\label{lemma1}
Let the outer-layer function of the proposed NMC estimator be the log function
\( f : \mathbb{R}_+ \to \mathbb{R},\: x \mapsto \log(x) \). Let the argument of \( f \) be lower-bounded by a sufficiently small positive constant
\( \tau > 0 \), i.e., \( f \) is only evaluated for inputs in \( [\tau, \infty) \). Then \( f \) is thrice continuously differentiable on \( [\tau, \infty) \), and all derivatives up to order three are bounded.
\end{lemma}
\begin{proof}
See Supplementary Material.
\end{proof}
Following \cref{lemma1}, the NMC estimator in \cref{eq:NMC_estimator_MI} satisfies the established convergence properties, with the mean squared error of $\widehat{U}_{S,J}(\boldsymbol{\xi}_{(b)})$ decaying at a rate of ${O}(1/S + 1/J^2)$.

\begin{corollary}[Convergence of the NMC estimator]
\label{corollary1}
Under the conditions in \cref{lemma1} and assuming independence between the inner and outer samplers, the mean squared error of
$\widehat{U}_{S,J}(\boldsymbol{\xi})$ converges to 0 at rate $O(S^{-1}+J^{-2})$,
\[
\mathbb{E}\left[
\left(
\widehat{U}_{S,J}(\boldsymbol{\xi})
-
U(\boldsymbol{\xi})
\right)^2
\right]
=
O(S^{-1}+J^{-2}).
\]
\end{corollary}
\begin{proof}
See Supplementary Material.
\end{proof}
If the total computational budget is $T=SJ$, balancing the two terms in the rate
$O(S^{-1}+J^{-2})$ gives
$S\asymp T^{2/3}$ and $J\asymp T^{1/3}$, yielding the rate
$O(T^{-2/3})$. It is also known that the NMC estimator is generally biased for any finite $J$. As $J\to\infty$, this bias vanishes. Therefore, when both $S$ and $J$ increase, the estimator is consistent.

\begin{corollary}[Bias and almost sure convergence of the NMC estimator]
\label{corollary2}
For any finite $S,J$, $\widehat{U}_{S,J}(\boldsymbol{\xi})$ is a biased estimator of $U(\boldsymbol{\xi})$, i.e.,
\[
\mathbb{E}\left[
\widehat{U}_{S,J}(\boldsymbol{\xi})
\right]
\neq
U(\boldsymbol{\xi}).
\]
Under the regularity conditions of \cref{lemma1}, if $J\to\infty$ and $S\to\infty$,
then $\widehat{U}_{S,J}(\boldsymbol{\xi})$ converges to $U(\boldsymbol{\xi})$ almost surely:
\[
\lim_{S\to\infty}
\lim_{J\to\infty}
\widehat U_{S,J}
=
U
\qquad \text{a.s.}
\]
\end{corollary}
\begin{proof}
See Supplementary Material.
\end{proof}

\subsection{Gaussian Mixture Reduction}
\label{subsec:gmrc}

The mixture approximation in \cref{eq:joint-predictive-mixture} can be expensive to compute when the number of posterior samples $J$ is large. We therefore mitigate the computational cost by compressing the Gaussian mixture while preserving the main distributional features.\par
The method is motivated by Runnalls' algorithm \cite{Runnalls2007} and clustering-based reduction methods \cite{Schieferdecker2009}. We first whiten the conditional predictive mean so that Euclidean distance becomes more meaningful for clustering. We then apply weighted $k$-means clustering to reduce the original $J$ components to an initial smaller mixture with $J_0$ components, followed by moment matching within each cluster. Finally, we further reduce the mixture using Runnalls' method until
$J_{\rm target}$ components remain, where
$J_{\rm target}<J_0<J$. To avoid unnecessary global pairwise comparisons, we further use $k$-nearest neighbors (k-NN) to restrict candidate merging pairs. This reduces the cost of pair search from $O(J_0^2)$ to $O(J_0 \times n_{\rm knn})$. The full procedure and algorithmic details are given in the Supplementary Material.

\subsection{Covariance Precomputation and Sequential Updates}
\label{subsec:covariance-updates}

The second computational strategy accelerates the repeated covariance updates in SDE. For each round, each posterior draw $\boldsymbol{\Omega}_s^{(j)}$ and each candidate input $\boldsymbol{\xi}$, the KOH predictive covariance in \cref{eq:predictive-mean-cov} requires to compute the observation covariance block
$\boldsymbol{\Sigma}_{oo}^{(b)}$ and its inverse. Recomputing this inverse from scratch would be unnecessarily expensive.\par
The direct computational cost of the predictive covariance is
\[
O\!\left(
N_o^{\,3}
\;+\;
N_o^{\,2} n^*
\;+\;
N_o\,{n^*}^{\,2}
\right),
\]
where $N_o = n+m+b$. To reduce this burden, we precompute the full covariance matrix for each posterior draw of $\boldsymbol{\Omega}$ over all observation points, candidate points, and prediction inputs, and then obtain the required submatrices by slicing the full matrix in each design round. In SDE, the posterior sample of $\boldsymbol{\Omega}$ remains fixed, while the covariance structure expands by adding one input at each round. We therefore can apply the Schur complement together with a rank-one update \cite{ShermanMorrison1950,Gramacy2020Surrogates} to update the inverse observation block and the predictive covariance efficiently. As a result, the complexity is reduced to
\[
O\!\left(
N_o^{2} + N_o n^* + {n^*}^{2}
\right).
\]
The detailed block update formulas are provided in the Supplementary Material.

\subsection{Computational Cost}
\label{sec:Computational-Cost}

We summarize the computational savings in evaluating the predictive covariance $\boldsymbol{\Sigma}^*$, achieved by combining the Gaussian mixture reduction in Section \ref{subsec:gmrc} with the precomputation and Schur complement strategy in Section \ref{subsec:covariance-updates}.\par
Let $\mathcal{D}_{\mathrm{cand}}$ contain $M$ initial candidates and let $B$ denote the design budget. Across all design rounds, the total number of candidate scans is
\[
\sum_{b=1}^B \bigl(M-b+1\bigr)
=
MB - \frac{B(B-1)}{2},
\]
which we summarize as $O(BM)$ for notational brevity.

\paragraph{Naive evaluation}
Without compression or precomputation, the overall cost is
\[
O\!\Big(J \times BM \times
(
N_o^{\,3}
\;+\;
N_o^{\,2} n^*
\;+\;
N_o\,{n^*}^{\,2}
)
\Big).
\]

\paragraph{With mixture reduction and fast covariance updates}
After implementing the two strategies, the overall cost becomes
\[
O\!\Big(J_{\mathrm{target}} \times BM \times
(
N_o^{2} + N_o n^* + {n^*}^{2}
)
\Big).
\]

\paragraph{Complexity of different design criteria}

\begin{enumerate}
    \item \textbf{Predictive MI}
    The predictive MI criterion requires $S$ outer samples to compute the log-density of a joint Gaussian distribution of dimension $(n^*+1)$ with $J_{\mathrm{target}}$ Gaussian mixture components. Therefore, the total complexity becomes
    \[
    O\!\Big(J_{\mathrm{target}} \times BM \times
    (
    N_o^{2} + N_o n^* + {n^*}^{2}
    ) + S \times J_{\mathrm{target}} \times BM \times {(n^*+1)}^{2}
    \Big).
    \]
    \item \textbf{MI--CX}
     MI--CX has an additional cost of evaluating the local complexity score at the remaining candidate points at each round.
    
    \item \textbf{IMSPE}
    IMSPE only requires computing the trace of $\boldsymbol{\Sigma}^*_{(b)}$ for each $\boldsymbol{\Omega}^{(j)}$ sample and does not involve the outer $S$-loop. With compression, the overall complexity reduces to
    \[
    O\!\Big(J_{\mathrm{target}} \times BM \times
    (
    N_o^{2} + N_o n^* + {n^*}^{2}
    )
    \Big).
    \]

    \item \textbf{D-optimality.}
     The computation of D-optimality relies primarily on the observation covariance block $\boldsymbol{\Sigma}_{oo}^{(b)}$ and the mean vector $\boldsymbol{\mu}_o^{(b)}$ in \cref{eq:FIM-theta-given-Omega}. While the update of $\big(\boldsymbol{\Sigma}_{oo}^{(b)}\big)^{-1}$ can be accelerated by using the Schur complement, evaluating the Fisher information matrix for each $\boldsymbol{\theta}^{(j)}=(\theta_1^{(j)},\dots,\theta_{d_\theta}^{(j)})$ still requires finite-difference evaluations of $\partial_{\theta_i}\boldsymbol{\mu}_o$ and $\partial_{\theta_i}\boldsymbol{\Sigma}_{oo}$, together with $d_\theta^2$ trace terms. The resulting total complexity can be written as
    \[
    O\!\Big( J \times BM \times d_\theta^2 \times N_o^2 \Big).
    \]

    \item \textbf{Maximin distance}
    The maximin distance criterion is independent of the KOH posterior. Its total cost is approximately
    \[
    O(BMd),
    \]
    where $d$ is the input dimension.
\end{enumerate}

\section{Numerical Experiments}
\label{sec:experiments}
We evaluate the predictive MI criterion and the proposed
MI--CX criterion through three examples with increasing design spaces. The first example is a one-dimensional synthetic problem
with a localized oscillatory discrepancy. 
The second example is a sensor placement problem for predicting a two-dimensional plume field, while the third considers a
three-dimensional discrete design problem based on the Duffing system. These examples compare the round-by-round predictive performance of predictive MI, MI--CX, and the benchmark criteria. We also compare SDE and ADE to assess the effectiveness of adapting newly observed physical responses. An ablation study separately assesses the contributions of the local complexity score and predictive MI. Detailed problem descriptions and additional results can be found in the Supplementary Material.

\subsection{Experimental Preliminaries and Evaluation Metrics}
\label{subsec:experimental-protocol}
We evaluate the following design scenarios: predictive MI, IMSPE, and
D-optimality under both SDE and ADE; MI--CX under ADE; and maximin
distance. For SDE, we use the computational strategies introduced in Section \ref{subsec:covariance-updates}.\par
Each experiment is repeated 10 times with different random seeds, and the reported performance is averaged across the 10 replicates. To assess the predictive performance, we consider two evaluation metrics that measure the prediction quality. The mean squared error (MSE) evaluates the accuracy of the predictive mean, while the continuous ranked probability score (CRPS) evaluates the quality of the full posterior predictive distribution and therefore accounts for predictive uncertainty. Formal definitions are provided in \cref{app:evaluation-metrics}.\par
To reduce the computational cost, we fix the hyperparameters of the
first-stage GP, $\boldsymbol{\phi}_1$, at their posterior means. We account
for uncertainty in the calibration parameters $\boldsymbol{\theta}$ and the
second-stage GP hyperparameters $\boldsymbol{\phi}_2$ using $J=1{,}000$
post-burn-in MCMC samples. For predictive MI and MI--CX, we
use $S=10{,}000$ outer predictive samples for the NMC estimator. When Gaussian mixture reduction is applied, the number of components is reduced from $J=1{,}000$ to $J_{\rm target}=200$. Preliminary checks indicated that this choice preserves the first two moments of the mixture sufficiently well. The sampling configuration is the same across the examples.\par
All design criteria are implemented using PyTorch \cite{paszke2019} and
Pyro.\footnote{\url{https://docs.pyro.ai/en/stable/contrib.oed.html}}

\subsection{One-Dimensional Localized-Oscillation Example}
\label{subsec:one-dimensional-example}
Our first synthetic example is motivated by the one-dimensional
calibration setting of \cite{brynjarsdottirOHagan2014}, with two calibration parameters and one controllable design input. The physical process contains localized oscillations that cannot be captured by the computer model. There are two calibration parameters and the localized oscillatory discrepancy is near $x=0.72$.\par 
We use LHD to select 40 initial simulator runs over the input space $(x,a,b)$. The 7 initial physical observations are also equally spaced between $0.03$ and $0.96$. We use 160 uniformly spaced candidate points $\mathcal{D}_{\text{cand}}$
and 300 nonoverlapping prediction points $\mathcal{X}_{\text{pred}}$
over the input domain. The experimental budget is $B=15$ physical observations.

\paragraph{Predictive performance}
We select \Cref{fig:one-d-final-fit} to illustrate the final predictive fit obtained by Adaptive MI--CX, which achieves the best performance in both metrics. \cref{fig:one-d-stepwise} compares the round-by-round MSE and CRPS of all eight design strategies. From \Cref{fig:one-d-final-fit}, Adaptive MI--CX successfully captures the localized oscillations and it achieves the lowest MSE and CRPS at the final round, followed by Adaptive MI. Adaptive updating does not improve every criterion in this example, as Adaptive IMSPE and D-optimality perform less well than their sequential counterparts.

\begin{figure}[!t]
\centering
\includegraphics[width=0.80\textwidth]{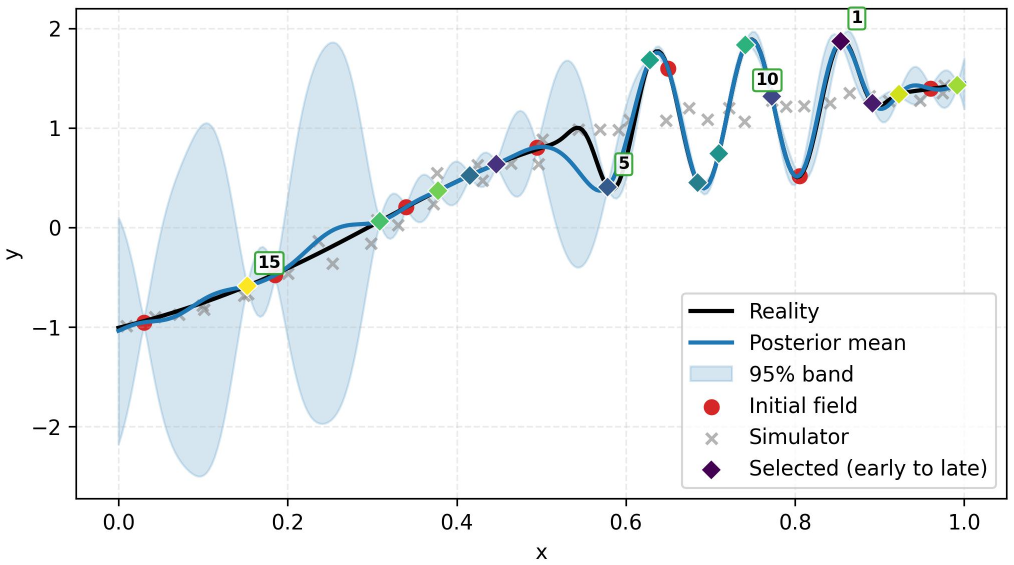}
\caption{A representative final-round predictive fit of Adaptive MI--CX after 15 design rounds. The shaded region denotes the 95\% uncertainty interval, and the diamonds indicate the selected physical observations.}
\label{fig:one-d-final-fit}
\end{figure}

\begin{figure}[!htbp]
\centering
\includegraphics[width=\textwidth]{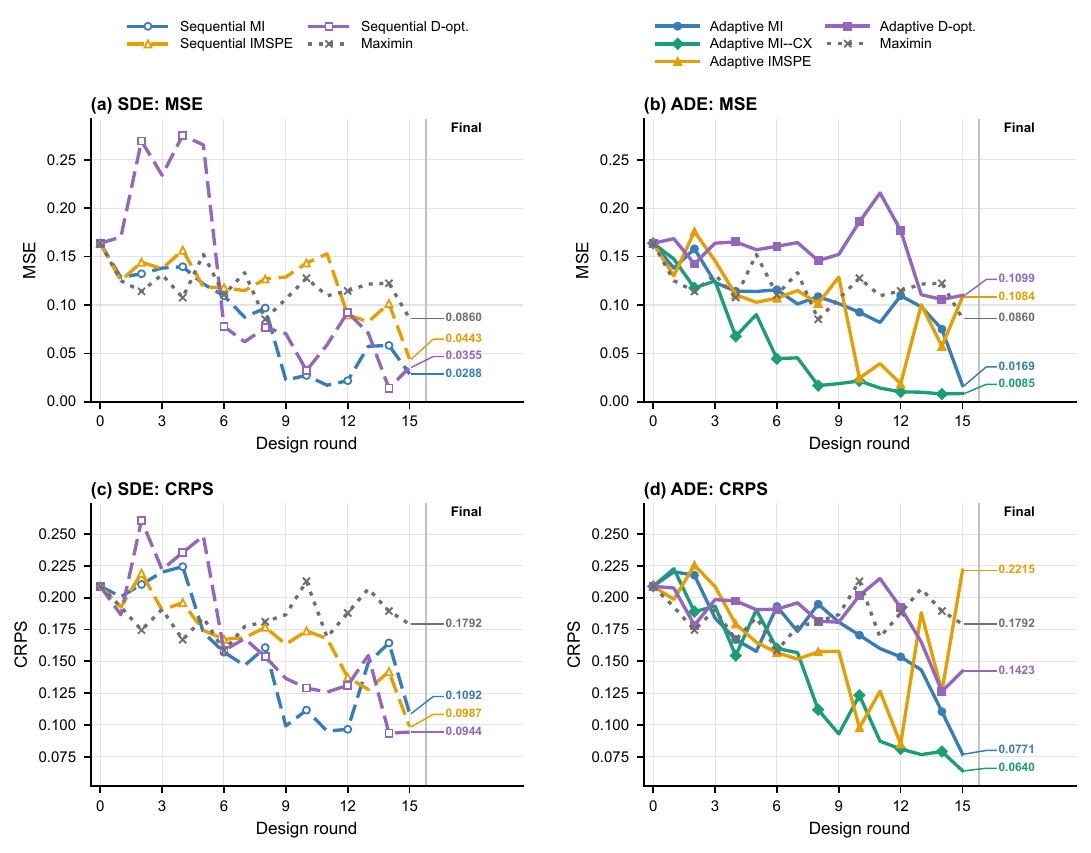}
\caption{Round-by-round predictive performance trajectories:
(a) SDE MSE, (b) ADE MSE, (c) SDE CRPS, and (d) ADE CRPS.
Maximin distance is included in both columns as a common baseline,
and the final-round values are shown on the right of each panel.}
\label{fig:one-d-stepwise}
\end{figure}

\paragraph{Effect of Gaussian mixture reduction}
In the first example, we also compare the computational cost of evaluating Sequential MI using the full NMC inner sample size and the strategy using Gaussian mixture reduction introduced in Section \ref{subsec:gmrc}. As shown in \cref{tab:one-d-mi-computation}, the reduced mixture retains 14 of the 15 physical inputs selected by the full NMC while saving  substantial running time.

\begin{table}[htbp]
\footnotesize
\caption{Computational comparison of the full NMC and Gaussian-mixture-reduced evaluations of Sequential MI.}
\label{tab:one-d-mi-computation}
\begin{center}
\begin{tabular}{lcccc}
\hline
MI evaluation & Components & Time per round (s) & Selection overlap \\
\hline
Full NMC
& $J=1{,}000$
& $93.24$
& Reference \\
Gaussian mixture reduction
& $J_{\rm target}=200$
& $18.39$
& $14/15$ \\
\hline
\end{tabular}
\end{center}
\end{table}

The model specification, complete selection order, fitting results and performance metrics for all design rounds, and acquisition times are reported in the Supplementary Material.

\subsection{Two-Dimensional Plume Sensor-Placement Example}
\label{subsec:plume-example}
Our second example is a sensor placement problem with the goal of predicting a two-dimensional plume field. This example is motivated by an atmospheric contaminant monitoring problem \cite{KeatsYeeLien2007,KeatsYeeLien2010}, which requires effective placement of the sensors. The design input is a two-dimensional sensor location, and the three calibration parameters control the source location, the wind direction, and the spreading coefficient. The simulator is a simplified Gaussian plume representation of atmospheric advection and diffusion. Different from the simulator, the physical field has multiscale centerline meandering, spatially varying plume width, and a weak secondary plume component.\par
20 initial physical observations and
160 simulator runs over
$[0,1]^2\times[0,1]^3$ are selected by a space-filling design. The candidate set contains 600 sensor locations,
The prediction space is gridded into 800 locations, and the design budget is
$B=30$.
\paragraph{Predictive performance}
\Cref{fig:plume-selected-locations} shows the locations selected by the eight
design strategies. Adaptive MI--CX places more sensors around the curved plume
core and its transition regions, whereas MI, IMSPE and maximin distance tend to distribute the sensors more uniformly over the whole space. As shown in \cref{fig:plume-stepwise}, Adaptive MI--CX achieves the lowest
final MSE, $3.97\times10^{-4}$, while Adaptive MI achieves the lowest final
CRPS, $0.0161$. The final CRPS of Adaptive MI--CX is higher, at $0.0253$. In this example, MI--CX places greater emphasis on reconstructing the curved plume structure and allocates fewer observations to smoother regions, where experiments can reduce much uncertainty but provide little improvement in predictive mean accuracy. The model specification, complete predictive fits, performance results, and running times are provided in the Supplementary Material. 
\begin{figure}[!t]
\centering
\includegraphics[width=\textwidth]{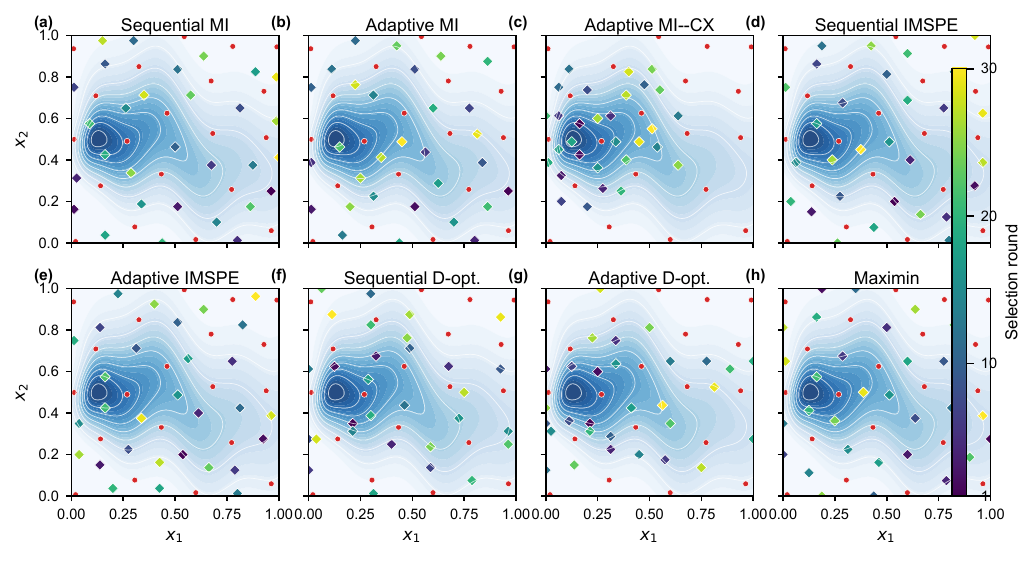}
\caption{Representative selected sensor locations in the two-dimensional plume example. The
background shows the true physical field, red circles are the initial
physical observations, and diamonds represent the selected locations, colored by
design round.}
\label{fig:plume-selected-locations}
\end{figure}
\begin{figure}[!t]
\centering
\includegraphics[width=\textwidth]{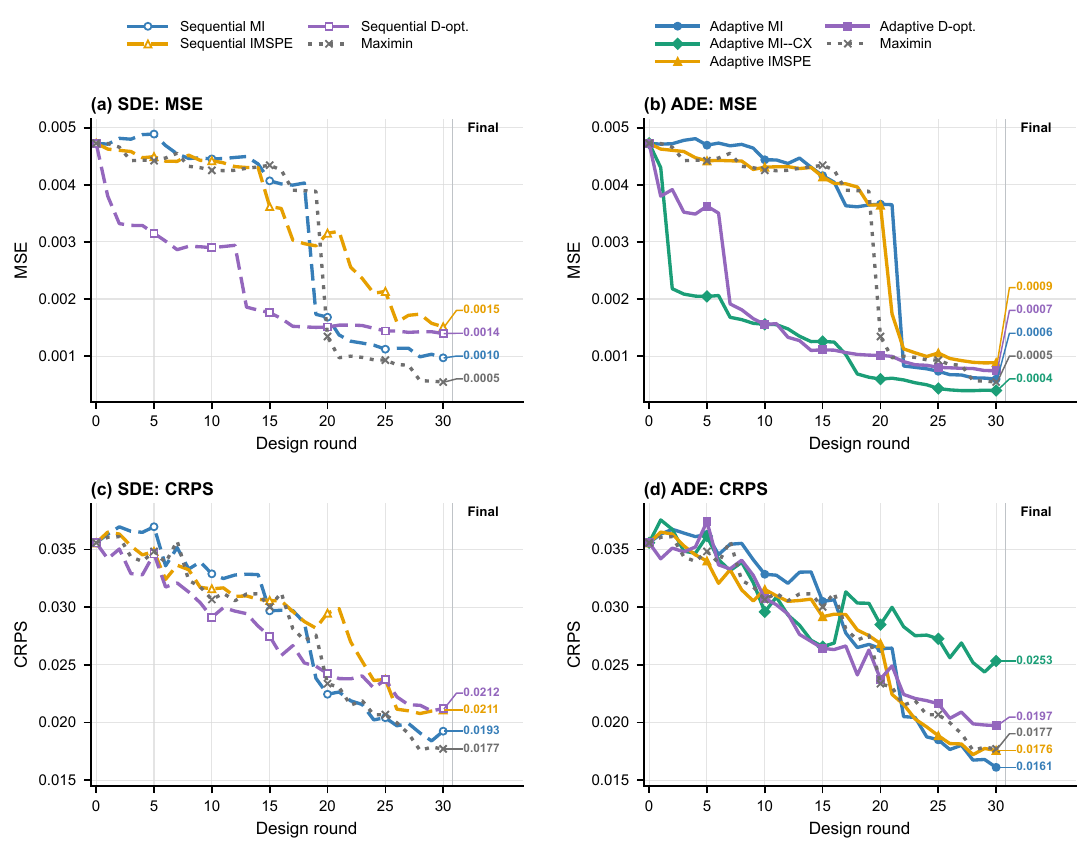}
\caption{Round-by-round predictive performance in the two-dimensional plume
example: (a) SDE MSE, (b) ADE MSE, (c) SDE CRPS, and (d) ADE CRPS.
Maximin distance is included in both columns as a common baseline, and
the final-round values are shown on the right of each panel.}
\label{fig:plume-stepwise}
\end{figure}

\subsection{Three-Dimensional Duffing-Oscillator Example}
\label{subsec:duffing-example}
Our third example is based on a
forced Duffing oscillator \cite{KovacicBrennan2011}. The controllable input is $\mathbf{x}=(\omega,F,k)^\top$ and has three dimensions, representing the forcing frequency, forcing
amplitude, and linear stiffness, respectively. One single calibration parameter $\theta$
is the damping ratio. The physical system is a hardening Duffing oscillator, while the simulator lacks the cubic stiffness term. Our target is to predict the peak-to-peak displacement response of the Duffing oscillator after it reaches steady state. For higher-dimensional design problems, searching over all continuous input combinations becomes computationally infeasible. We therefore grid the design space into an $18\times8\times6$ lattice of 864 candidate points. We start from 8 initial physical
observations and 192 computer model runs, and the design budget is $B=40$ rounds.
\paragraph{Predictive performance}
\Cref{fig:duffing-stepwise} shows that Adaptive MI--CX achieves the lowest final MSE,
$0.0158$. For CRPS, Adaptive IMSPE attains
the lowest final value, $0.0842$, while Adaptive MI--CX obtains $0.0974$. Overall, the complexity-weighted criterion performs best among all methods. The Maximin distance is also competitive, reflecting the robust performance of space-filling designs under a relatively large design budget.\par
\begin{figure}[!t]
\centering
\includegraphics[width=\textwidth]{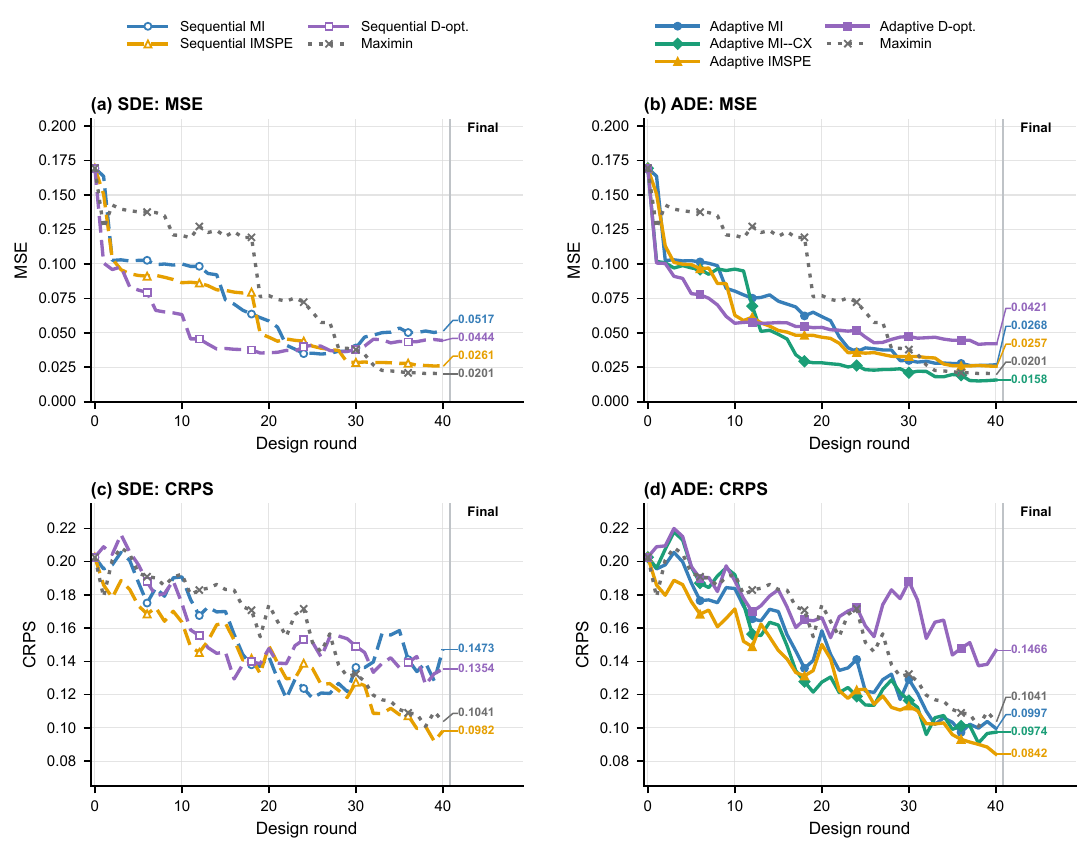}
\caption{Round-by-round predictive performance in the three-dimensional
Duffing-oscillator example: (a) SDE MSE, (b) ADE MSE, (c) SDE CRPS,
and (d) ADE CRPS. Maximin distance is included in both columns as a
common baseline, and the final-round values are shown on the right of
each panel.}
\label{fig:duffing-stepwise}
\end{figure}

\Cref{fig:duffing-fit} illustrates the Adaptive MI--CX design and its final prediction on a representative fixed-stiffness slice. The selected inputs are mostly around
the nonlinear response transition and the posterior predictive mean captures the main resonance pattern. The model equations, parameter settings and additional numerical results
are provided in the Supplementary Material. 

\begin{figure}[!t]
\centering
\includegraphics[width=\textwidth]{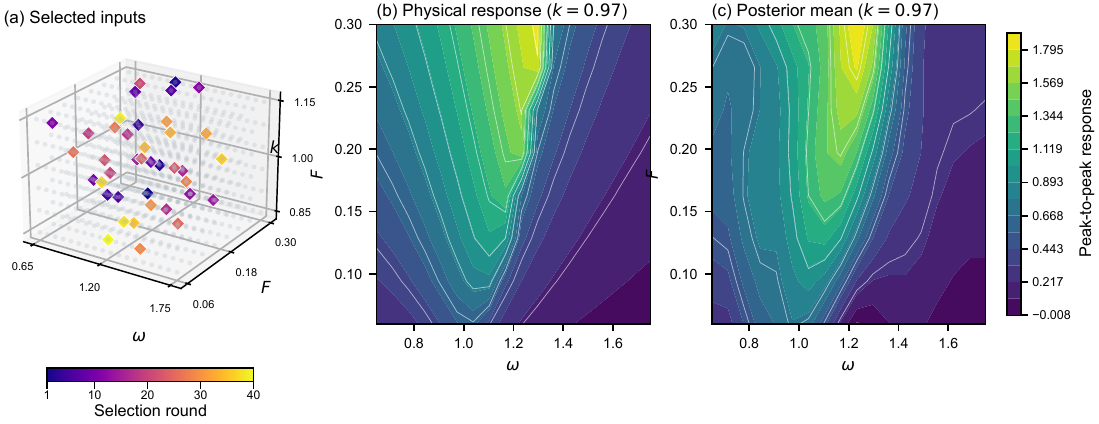}
\caption{Representative Adaptive MI--CX results for the three-dimensional Duffing-oscillator
example. (a) Selected inputs, colored
by design round. Panels (b)--(c) show the true physical response and the
posterior predictive mean at $k=0.97$.}
\label{fig:duffing-fit}
\end{figure}
\subsection{Ablation Study}
\label{subsec:ablation}

To verify the contributions of the proposed local complexity
score, we compare Adaptive MI, CX-only, and Adaptive MI--CX across all three examples.  As shown in \cref{tab:ablation-final},  Adaptive MI--CX achieves the lowest final MSE in all three examples and also the lowest final CRPS in the one-dimensional and Duffing examples, while Adaptive MI gives the lowest CRPS in the plume example.
CX-only performs worse than both Adaptive MI and Adaptive MI--CX. Thus, the local complexity score is more effective as a weighting factor of the predictive MI than as a stand-alone criterion. The complete
performance trajectories are provided in the Supplementary Material.

\begin{table}[H]
\footnotesize
\setlength{\tabcolsep}{4pt}
\caption{Final-round ablation results. Lower values are better, and the
best value in each row is shown in bold.}
\label{tab:ablation-final}
\begin{center}
\begin{tabular}{llccc}
\hline
Example & Metric & Adaptive MI & CX-only & Adaptive MI--CX \\
\hline
One-dimensional & MSE  & 0.0169 & 0.0702 & \textbf{0.0085} \\
                & CRPS & 0.0771 & 0.1652 & \textbf{0.0640} \\
Plume           & MSE  & $5.97\!\times\!10^{-4}$ & $1.01\!\times\!10^{-3}$
                         & $\mathbf{3.97\times10^{-4}}$ \\
                & CRPS & \textbf{0.0161} & 0.0424 & 0.0253 \\
Duffing         & MSE  & 0.0268 & 0.0530 & \textbf{0.0158} \\
                & CRPS & 0.0997 & 0.1549 & \textbf{0.0974} \\
\hline
\end{tabular}
\end{center}
\end{table}

\section{Conclusion and Discussion}
\label{sec:ConcluDiscuss}

We developed a prediction-oriented Bayesian experimental design framework for selecting physical observations under the KOH model. Predictive MI measures the expected information provided by potential observations from candidate sets, while MI--CX incorporates a local complexity score to emphasize regions where the predictive response shows stronger local nonlinearity. We considered both SDE without posterior updates and ADE, where each new physical observation updates the GP posterior before the next selection.\par
The three numerical examples considered one-, two-, and three-dimensional design spaces and compared the proposed criteria with IMSPE, D-optimality, and maximin distance. Overall, Adaptive MI--CX performs best across the three examples, achieving the lowest MSE in all cases while also showing substantial uncertainty reduction. The ablation study showed that CX-only is not effective  enough as a stand-alone criterion but as a weighting factor. The results also demonstrate that ADE does not always improve the predictive quality compared to SDE, suggesting that the benefit of adaptation depends in a complex way on the nature of the physical process and the computer experiment in conjunction with the criterion.\par
We also conduct a theoretical analysis that builds the connection between the two commonly used prediction-oriented criteria, predictive  MI and IMSPE. We found that the MI criterion reduces to a scaled IMSPE criterion when the covariance update becomes sufficiently small. This result explains the close relationship between the two criteria when there is an abundance of information. However, one can still observe that their different treatment of predictive covariance can have a significant effect on the design, especially at early design rounds.\par
We introduced several computational strategies for computing the design criteria while accounting for parameter and response uncertainty in \cref{sec:computation}. Gaussian mixture reduction decreases the number of components in the predictive mixture distribution, while covariance precomputation and Schur complement updates substantially reduce the computational cost of repeated matrix operations during SDE.\par
The present study is limited to greedy one-step selection, and the computational cost of the NMC estimator remains relatively high. Future work may consider batch or nonmyopic design, continuous optimization of the proposed criteria, and adapting variational and deep adaptive design methods \cite{Foster2019,Foster2021} to reduce the computational cost of predictive MI. It would also be useful to consider uncertainty from emulator construction more completely and to investigate the proposed criteria for computer experiments.

\appendix

\crefalias{section}{appendix}
\crefname{appendix}{Appendix}{Appendices}
\Crefname{appendix}{Appendix}{Appendices}
\section{Evaluation Metrics}
\label{app:evaluation-metrics}

\begin{itemize}

\item \textbf{Mean squared error (MSE).}
Given $J$ posterior samples
$\{\boldsymbol{\Omega}^{(j)}\}_{j=1}^{J}$, let
\[
\boldsymbol{\mu}_{j}^{*}
=
\boldsymbol{\mu}^{*}
\bigl(\mathbf{X}^{*};\boldsymbol{\Omega}^{(j)}\bigr)
\in \mathbb{R}^{n^{*}}
\]
denote the conditional posterior predictive mean at the prediction inputs
\[
\mathcal{X}_{\rm pred}
=
\mathbf{X}^{*}
\in \mathbb{R}^{n^{*}\times d}.
\]
The fully Bayesian posterior predictive mean is approximated by averaging over the posterior samples:
\begin{equation}
\label{eq:app-predictive-mean}
\overline{\boldsymbol{\mu}}^{*}
=
\frac{1}{J}
\sum_{j=1}^{J}
\boldsymbol{\mu}_{j}^{*}.
\end{equation}
Let $\mathbf{y}^{*}\in\mathbb{R}^{n^{*}}$ denote the true physical responses at the prediction inputs. We define the predictive MSE as
\begin{equation}
\label{eq:app-mse}
\mathrm{MSE}
=
\frac{1}{n^{*}}
\left\lVert
\overline{\boldsymbol{\mu}}^{*}
-
\mathbf{y}^{*}
\right\rVert_{2}^{2}.
\end{equation}
Thus, the MSE measures the average squared difference between the fully Bayesian posterior predictive mean and the true physical response over the prediction set.

\item \textbf{Continuous ranked probability score (CRPS).}
The CRPS \cite{GneitingRaftery2007} evaluates the quality of the full posterior predictive distribution rather than only its mean. Suppose that
\[
\widetilde{\mathbf{y}}_{1}^{*},
\ldots,
\widetilde{\mathbf{y}}_{S}^{*}
\sim
p\bigl(
\mathbf{y}^{*}
\mid
\mathbf{X}^{*},\mathbf{y}
\bigr)
\]
are $S$ posterior predictive samples. Each predictive sample is obtained by first drawing
\[
\boldsymbol{\Omega}^{(j_s)}
\sim
p(\boldsymbol{\Omega}\mid\mathbf{y})
\]
and then sampling
\[
\widetilde{\mathbf{y}}_{s}^{*}
\sim
p\bigl(
\mathbf{y}^{*}
\mid
\mathbf{X}^{*},\mathbf{y},
\boldsymbol{\Omega}^{(j_s)}
\bigr).
\]
The CRPS averaged over the prediction set is estimated by
\begin{equation}
\label{eq:app-crps}
\widehat{\mathrm{CRPS}}
=
\frac{1}{n^{*}S}
\sum_{s=1}^{S}
\left\lVert
\widetilde{\mathbf{y}}_{s}^{*}
-
\mathbf{y}^{*}
\right\rVert_{1}
-
\frac{1}{2n^{*}S^{2}}
\sum_{s=1}^{S}
\sum_{s'=1}^{S}
\left\lVert
\widetilde{\mathbf{y}}_{s}^{*}
-
\widetilde{\mathbf{y}}_{s'}^{*}
\right\rVert_{1}.
\end{equation}
The first term measures the discrepancy between the posterior predictive samples and the true physical responses, whereas the second term accounts for the dispersion of the posterior predictive distribution.

\end{itemize}

Lower MSE and CRPS values indicate better predictive performance. Considering these two metrics together evaluates both the predictive accuracy and the quality of the full predictive distribution.

\section*{Acknowledgments}
The authors were supported by the Linz Institute of Technology's grant LIFT\_C-2022-1-123 funded by the State of Upper Austria and Austria's Federal Ministry of Education, Science and Research.

\clearpage
\part*{Supplementary Materials}
\addcontentsline{toc}{part}{Supplementary Materials}

\setcounter{section}{0}
\setcounter{subsection}{0}
\setcounter{equation}{0}
\setcounter{figure}{0}
\setcounter{table}{0}
\setcounter{algorithm}{0}
\renewcommand{\thesection}{SM\arabic{section}}
\renewcommand{\thesubsection}{\thesection.\arabic{subsection}}
\renewcommand{\thesubsubsection}{\thesubsection.\arabic{subsubsection}}
\renewcommand{\theequation}{SM\arabic{section}.\arabic{equation}}
\renewcommand{\thefigure}{SM\arabic{figure}}
\renewcommand{\thetable}{SM\arabic{table}}
\renewcommand{\thealgorithm}{SM\arabic{algorithm}}
\renewcommand{\thetheorem}{SM\arabic{section}.\arabic{theorem}}
\renewcommand{\theHsection}{SM.\arabic{section}}
\renewcommand{\theHsubsection}{SM.\arabic{section}.\arabic{subsection}}
\renewcommand{\theHsubsubsection}{SM.\arabic{section}.\arabic{subsection}.\arabic{subsubsection}}
\renewcommand{\theHequation}{SM.\arabic{section}.\arabic{equation}}
\renewcommand{\theHfigure}{SM.\arabic{figure}}
\renewcommand{\theHtable}{SM.\arabic{table}}
\providecommand{\theHalgorithm}{}
\renewcommand{\theHalgorithm}{SM.\arabic{algorithm}}
\providecommand{\theHtheorem}{}
\renewcommand{\theHtheorem}{SM.\arabic{section}.\arabic{theorem}}
\crefname{section}{Section}{Sections}
\Crefname{section}{Section}{Sections}

\section{Proofs of Theoretical Results}
\label{supp:proofs}

\subsection{Proof of Theorem~3.1}
\label{supp:proof-mi-imspe}

\begin{proof}[Proof of Theorem~3.1]
Throughout the proof we condition on $\boldsymbol{\Omega}$. For brevity, write
\[
\boldsymbol{\Sigma}_0 := \boldsymbol{\Sigma}^*_{(b-1)}, 
\qquad
\boldsymbol{\Sigma}_1 := \boldsymbol{\Sigma}^*_{(b)},
\qquad
\Delta\boldsymbol{\Sigma}
:=
\Delta\boldsymbol{\Sigma}^*(\boldsymbol{\xi}_{(b)})
=
\boldsymbol{\Sigma}_0-\boldsymbol{\Sigma}_1
\succeq \mathbf 0.
\]
By definition,
\[
\Delta\mathrm{IMSPE}(\boldsymbol{\xi}_{(b)}\mid \boldsymbol{\Omega})
=
\frac{1}{n^*}\operatorname{tr}\bigl(\Delta\boldsymbol{\Sigma}\bigr).
\]
According to Assumption~(i), the mutual information formula has the closed-form expression
\[
U(\boldsymbol{\xi}_{(b)}\mid \boldsymbol{\Omega})
=
I\left(
\mathbf y^*;
\mathbf y^{\mathrm{new}}_{(b)}(\boldsymbol{\xi})
\mid
\boldsymbol{\Omega},
\mathcal D_{\rm sel}^{(b)}(\boldsymbol{\xi})
\right)
=
\frac{1}{2}
\left(
\log\det\boldsymbol{\Sigma}_0
-
\log\det\boldsymbol{\Sigma}_1
\right).
\]
According to Assumption~(ii), $\boldsymbol{\Sigma}_0$ is symmetric positive definite, so we can rewrite $\boldsymbol{\Sigma}_1$ and the standardized covariance reduction $\mathbf B$ as
\[
\boldsymbol{\Sigma}_1
=
\boldsymbol{\Sigma}_0-\Delta\boldsymbol{\Sigma}
=
\boldsymbol{\Sigma}_0^{1/2}
(\mathbf I_{n^*}-\mathbf B)
\boldsymbol{\Sigma}_0^{1/2},
\qquad
\mathbf B
:=
\boldsymbol{\Sigma}_0^{-1/2}
\Delta\boldsymbol{\Sigma}
\boldsymbol{\Sigma}_0^{-1/2}
\succeq \mathbf 0.
\]
Then
\[
U(\boldsymbol{\xi}_{(b)}\mid \boldsymbol{\Omega})
=
-\frac{1}{2}\log\det(\mathbf I_{n^*}-\mathbf B).
\]
Let $\tau_1,\ldots,\tau_{n^*}$ be the eigenvalues of $\mathbf B$.
Since $\mathbf B\succeq\mathbf 0$, $\tau_i\ge 0$ and
\[
\log\det(\mathbf I_{n^*}-\mathbf B)
=
\sum_{i=1}^{n^*}\log(1-\tau_i),
\]
so
\[
U(\boldsymbol{\xi}_{(b)}\mid \boldsymbol{\Omega})
=
-\frac{1}{2}\sum_{i=1}^{n^*}\log(1-\tau_i).
\]
Under the small-update regime,
\[
\|\mathbf B\|_2
=
\left\|
(\boldsymbol{\Sigma}^*_{(b-1)})^{-1/2}
\Delta\boldsymbol{\Sigma}^*(\boldsymbol{\xi}_{(b)})
(\boldsymbol{\Sigma}^*_{(b-1)})^{-1/2}
\right\|_2
\longrightarrow 0,
\]
and hence $\max_i\tau_i\to 0$.

For all sufficiently large $b$, choose $\|\mathbf B\|_2<1/2$, so that
$0\le \tau_i<1/2$ for all $i$. For each $i$, define
\[
r(\tau_i):=-\log(1-\tau_i)-\tau_i.
\]
The Taylor expansion of $-\log(1-\tau_i)$ at $0$ gives
\[
-\log(1-\tau_i)
=
\tau_i+\sum_{k=2}^{\infty}\frac{\tau_i^k}{k},
\]
so that
\[
r(\tau_i)
=
\sum_{k=2}^{\infty}\frac{\tau_i^k}{k}.
\]
Moreover,
\[
0\le r(\tau_i)
\le
\frac{1}{2}\sum_{k=2}^{\infty}\tau_i^k
=
\frac{1}{2}\frac{\tau_i^2}{1-\tau_i}
\le
\tau_i^2.
\]
Therefore, with
\[
R_b=\frac{1}{2}\sum_{i=1}^{n^*}r(\tau_i),
\]
we have
\[
|R_b|
\le
\frac{1}{2}\sum_{i=1}^{n^*}\tau_i^2.
\]
Using cyclic invariance of the trace,
\[
\sum_{i=1}^{n^*}\tau_i
=
\operatorname{tr}(\mathbf B)
=
\operatorname{tr}\bigl(\boldsymbol{\Sigma}_0^{-1}\Delta\boldsymbol{\Sigma}\bigr),
\qquad
\sum_{i=1}^{n^*}\tau_i^2
=
\|\mathbf B\|_F^2,
\]
we obtain
\begin{equation}
\label{eq:supp-U-main-rem}
U(\boldsymbol{\xi}_{(b)}\mid \boldsymbol{\Omega})
=
\frac{1}{2}
\operatorname{tr}\bigl(\boldsymbol{\Sigma}_0^{-1}\Delta\boldsymbol{\Sigma}\bigr)
+
R_b,
\qquad
|R_b|
\le
\frac{1}{2}\|\mathbf B\|_F^2.
\end{equation}

Dividing by
$\Delta\mathrm{IMSPE}(\boldsymbol{\xi}_{(b)}\mid \boldsymbol{\Omega})$ gives
\begin{equation}
\label{eq:supp-ratio-split}
\frac{U(\boldsymbol{\xi}_{(b)}\mid \boldsymbol{\Omega})}
     {\Delta\mathrm{IMSPE}(\boldsymbol{\xi}_{(b)}\mid \boldsymbol{\Omega})}
=
\frac{n^*}{2}
\frac{
\operatorname{tr}\bigl(\boldsymbol{\Sigma}_0^{-1}\Delta\boldsymbol{\Sigma}\bigr)
}{
\operatorname{tr}(\Delta\boldsymbol{\Sigma})
}
+
\frac{n^*}{\operatorname{tr}(\Delta\boldsymbol{\Sigma})}R_b.
\end{equation}
By the standard trace inequality for
$\boldsymbol{\Sigma}_0^{-1}\succ \mathbf 0$ and
$\Delta\boldsymbol{\Sigma}\succeq \mathbf 0$,
\[
\lambda_{\min}(\boldsymbol{\Sigma}_0^{-1})
\operatorname{tr}(\Delta\boldsymbol{\Sigma})
\le
\operatorname{tr}\bigl(\boldsymbol{\Sigma}_0^{-1}\Delta\boldsymbol{\Sigma}\bigr)
\le
\lambda_{\max}(\boldsymbol{\Sigma}_0^{-1})
\operatorname{tr}(\Delta\boldsymbol{\Sigma}).
\]
Since
$\lambda_{\min}(\boldsymbol{\Sigma}_0^{-1})=1/\lambda_{\max}$ and
$\lambda_{\max}(\boldsymbol{\Sigma}_0^{-1})=1/\lambda_{\min}$, we obtain
\begin{equation}
\label{eq:supp-main-bounds}
\frac{n^*}{2\lambda_{\max}}
\le
\frac{n^*}{2}
\frac{
\operatorname{tr}\bigl(\boldsymbol{\Sigma}_0^{-1}\Delta\boldsymbol{\Sigma}\bigr)
}{
\operatorname{tr}(\Delta\boldsymbol{\Sigma})
}
\le
\frac{n^*}{2\lambda_{\min}}.
\end{equation}

For the remainder term, cyclic invariance gives
\[
\Delta\boldsymbol{\Sigma}
=
\boldsymbol{\Sigma}_0^{1/2}\mathbf B\boldsymbol{\Sigma}_0^{1/2},
\qquad
\operatorname{tr}(\Delta\boldsymbol{\Sigma})
=
\operatorname{tr}(\mathbf B\boldsymbol{\Sigma}_0)
\ge
\lambda_{\min}\sum_{i=1}^{n^*}\tau_i,
\]
where the inequality follows from the trace inequality. Moreover,
\[
\sum_i\tau_i^2
\le
(\max_i\tau_i)\sum_i\tau_i
=
\|\mathbf B\|_2\sum_i\tau_i,
\]
where $\|\cdot\|_2$ denotes the spectral norm.

Thus,
\[
\frac{|R_b|}{\operatorname{tr}(\Delta\boldsymbol{\Sigma})}
\le
\frac{\frac{1}{2}\sum_i\tau_i^2}
     {\lambda_{\min}\sum_i\tau_i}
\le
\frac{1}{2\lambda_{\min}}\|\mathbf B\|_2
\longrightarrow 0
\quad\text{as}\quad
\Delta\boldsymbol{\Sigma}^*(\boldsymbol{\xi}_{(b)})\to \mathbf 0.
\]
Together with \cref{eq:supp-ratio-split,eq:supp-main-bounds},
\[
\frac{U(\boldsymbol{\xi}_{(b)}\mid \boldsymbol{\Omega})}
     {\Delta\mathrm{IMSPE}(\boldsymbol{\xi}_{(b)}\mid \boldsymbol{\Omega})}
=
\frac{n^*}{2}
\frac{
\operatorname{tr}\bigl(\boldsymbol{\Sigma}_0^{-1}\Delta\boldsymbol{\Sigma}\bigr)
}{
\operatorname{tr}(\Delta\boldsymbol{\Sigma})
}
+
o(1).
\]
If the limit
\[
C_{b-1}
:=
\lim_{\Delta\boldsymbol{\Sigma}^*(\boldsymbol{\xi}_{(b)})\to 0}
\frac{
U(\boldsymbol{\xi}_{(b)}\mid\boldsymbol{\Omega})
}{
\Delta\mathrm{IMSPE}(\boldsymbol{\xi}_{(b)}\mid\boldsymbol{\Omega})
}
\]
exists, then it is finite and satisfies
\[
\frac{n^*}{2\,\lambda_{\max}\!\bigl(\boldsymbol{\Sigma}^*_{(b-1)}\bigr)}
\le
C_{b-1}
\le
\frac{n^*}{2\,\lambda_{\min}\!\bigl(\boldsymbol{\Sigma}^*_{(b-1)}\bigr)}.
\]
\end{proof}

\subsection{Proof of Lemma~4.1}
\label{supp:proof-lemma1}

\begin{proof}[Proof of Lemma~4.1]
Let the argument of $f$ be $x$, where $x>0$. The first three derivatives of
$f(x)=\log x$ are
\[
f'(x)=\frac{1}{x},\qquad
f''(x)=-\frac{1}{x^2},\qquad
f^{(3)}(x)=\frac{2}{x^3}.
\]
On the restricted domain $[\tau,\infty)$ with $\tau>0$, each derivative is continuous and satisfies
\[
|f'(x)|\le \frac{1}{\tau},\qquad
|f''(x)|\le \frac{1}{\tau^2},\qquad
|f^{(3)}(x)|\le \frac{2}{\tau^3}.
\]
Hence $f\in C^3([\tau,\infty))$, and all derivatives up to order three are uniformly bounded by finite constants.
\end{proof}

\subsection{Proof of Corollary~4.2}
\label{supp:proof-corollary1}

The proofs of Corollary~4.2 and of Corollary~4.3 were largely taken from \cite{Rainforth2018} and are stated here for the reader's convenience.

\begin{proof}[Proof of Corollary~4.2]
For notational brevity, let $\mathbf z$ denote a generic random variable drawn from the predictive distribution in the MI-based utility function. It can represent either $\mathbf y^*$, $\mathbf y^{\mathrm{new}}_{(b)}(\boldsymbol{\xi})$, or their joint pair
$(\mathbf y^*,\mathbf y^{\mathrm{new}}_{(b)}(\boldsymbol{\xi}))$.

Consider the utility function and its NMC estimator,
\[
U=\mathbb E_{\mathbf z}[f(\gamma(\mathbf z))],
\qquad
\widehat U_{S,J}
=
\frac{1}{S}\sum_{s=1}^{S} f(\widehat\gamma_J(\mathbf z_s)),
\]
where
\[
\gamma(\mathbf z)
=
\int p(\mathbf z\mid\boldsymbol{\Omega})p(\boldsymbol{\Omega})\,d\boldsymbol{\Omega},
\qquad
\widehat\gamma_J(\mathbf z)
=
\frac{1}{J}\sum_{j=1}^{J}p(\mathbf z\mid\boldsymbol{\Omega}^{(j)}).
\]
The outer-layer function is $f(\cdot)=\log(\cdot)$.

Let
\[
\Delta(\mathbf z)=\widehat\gamma_J(\mathbf z)-\gamma(\mathbf z).
\]
For fixed $\mathbf z$, the second-order Taylor expansion of $f$ around $\gamma(\mathbf z)$ gives
\[
f(\widehat\gamma_J(\mathbf z))
=
f(\gamma(\mathbf z))
+
f'(\gamma(\mathbf z))\Delta
+
\frac{f''(\gamma(\mathbf z))}{2}\Delta^2
+
\frac{f^{(3)}(\eta)}{6}\Delta^3,
\]
where $\eta$ lies between $\widehat\gamma_J(\mathbf z)$ and $\gamma(\mathbf z)$.

By Lemma~4.1, all derivatives of $f$ up to order three are bounded on $[\tau,\infty)$. Under the standard moment assumptions for the inner Monte Carlo estimator,
\[
\mathbb E[\Delta\mid\mathbf z]=0,\qquad
\mathbb E[\Delta^2\mid\mathbf z]=\frac{\sigma^2(\mathbf z)}{J},\qquad
\mathbb E[|\Delta|^3\mid\mathbf z]\le \frac{C}{J^{3/2}},
\]
for some finite constant $C>0$ independent of $J$. Therefore,
\[
\left|
\mathbb E[f(\widehat\gamma_J(\mathbf z))\mid\mathbf z]
-
f(\gamma(\mathbf z))
\right|
\le
\frac{C_1}{J},
\qquad
\operatorname{Var}\!\left(f(\widehat\gamma_J(\mathbf z))\mid\mathbf z\right)
\le
\frac{C_2}{J}.
\]
The MSE decomposition gives
\[
\operatorname{MSE}(\widehat U_{S,J})
\le
\frac{V_0}{S}
+
\frac{C_2}{SJ}
+
\frac{C_1^2}{J^2},
\]
where $V_0$ is finite. Since the MI estimator is a linear combination of one joint and two marginal logarithmic terms, the same rate applies:
\[
\mathbb E\!\left[
(\widehat U_{S,J}-U)^2
\right]
=
O(S^{-1}+J^{-2}).
\]
\end{proof}

\subsection{Proof of Corollary~4.3}
\label{supp:proof-corollary2}

\begin{proof}[Proof of Corollary~4.3]
(i)~(\emph{Bias}).  
Recall that for each outer sample $\mathbf z_s$, the inner Monte Carlo estimate is
\[
\widehat{\gamma}_J(\mathbf z_s)
=
\frac{1}{J}\sum_{j=1}^{J}p(\mathbf z_s\mid\boldsymbol{\Omega}^{(j)}),
\qquad
\gamma(\mathbf z_s)
=
\mathbb E_{\boldsymbol{\Omega}\mid\mathcal D}
\!\left[
p(\mathbf z_s\mid\boldsymbol{\Omega})
\right].
\]
Then the NMC estimator is
\[
\widehat U_{S,J}
=
\frac{1}{S}\sum_{s=1}^{S}f(\widehat{\gamma}_J(\mathbf z_s)),
\qquad
U
=
\mathbb E_{\mathbf z}[f(\gamma(\mathbf z))].
\]
By Jensen's inequality, since $\log(\cdot)$ is strictly concave,
\[
\mathbb E
\left[
f(\widehat{\gamma}_J(\mathbf z))
\mid \mathbf z
\right]
=
\mathbb E
\left[
\log(\widehat{\gamma}_J(\mathbf z))
\mid \mathbf z
\right]
\le
\log
\left(
\mathbb E[
\widehat{\gamma}_J(\mathbf z)
\mid \mathbf z]
\right)
=
\log(\gamma(\mathbf z))
=
f(\gamma(\mathbf z)).
\]
Equality holds only in the degenerate case where
$\widehat{\gamma}_J(\mathbf z)$ is deterministic conditional on $\mathbf z$.
Thus, each log-density term has a downward bias. Since the MI estimator
contains one joint log-density term and two marginal log-density terms with
opposite signs, the overall NMC estimator is generally biased, but the sign
of the bias is not fixed. Under the stated regularity conditions,
\[
\mathbb E[\widehat U_{S,J}] - U = O(J^{-1}).
\]

\medskip
(ii)~(\emph{Almost sure convergence}).  
By the strong law of large numbers, the inner Monte Carlo estimator satisfies
\[
\widehat{\gamma}_J(\mathbf z)
=
\frac{1}{J}\sum_{j=1}^{J}
p(\mathbf z\mid\boldsymbol{\Omega}^{(j)},\mathcal D)
\xrightarrow{\rm a.s.}
\mathbb E_{\boldsymbol{\Omega}\mid\mathcal D}
\left[
p(\mathbf z\mid\boldsymbol{\Omega},\mathcal D)
\right]
=
\gamma(\mathbf z)
\]
as $J\to\infty$. Since $\gamma(\mathbf z)>0$ and $f(\cdot)=\log(\cdot)$ is continuous on
$(0,\infty)$, the continuous mapping theorem gives
\[
f(\widehat{\gamma}_J(\mathbf z))
\xrightarrow{\rm a.s.}
f(\gamma(\mathbf z)).
\]
For the outer layer,
\[
\widehat U_{S,J}
=
\frac{1}{S}\sum_{s=1}^{S}
f(\widehat{\gamma}_J(\mathbf z_s)).
\]
For each fixed $S$, as $J\to\infty$, we have pointwise convergence
\[
f(\widehat{\gamma}_J(\mathbf z_s))
\to
f(\gamma(\mathbf z_s)).
\]
Then, applying the strong law of large numbers again,
\[
\frac{1}{S}\sum_{s=1}^{S}f(\gamma(\mathbf z_s))
\xrightarrow{\rm a.s.}
\mathbb E_{\mathbf z}[f(\gamma(\mathbf z))]
=
U.
\]
Combining the two steps gives the iterated almost sure limit
\[
\lim_{S\to\infty}
\lim_{J\to\infty}
\widehat U_{S,J}
=
U
\qquad \text{a.s.}
\]
This completes the proof.
\end{proof}

\section{Computational Details of the Design Criteria}
\label{supp:comp-details}

This section provides the computational details for the strategies introduced in Section~4 of the main text, including the Gaussian mixture reduction procedure and the Schur complement update used for efficient covariance computation in sequential design.

\subsection{Gaussian Mixture Reduction}
\label{supp:gmm-reduction}

We first briefly review Runnalls' algorithm. Runnalls' method \cite{Runnalls2007} reduces a Gaussian mixture model (GMM) by greedily merging components based on minimizing Kullback--Leibler (KL) divergence. At each iteration, the algorithm evaluates and compares the divergence for candidate pairs, merges the optimal pair via moment matching, and repeats this process until the target number of components is reached. In addition, the $k$-means algorithm \cite{Lloyd1982} is a classical clustering method frequently employed for high-dimensional data classification. Our procedure differs from \cite{Schieferdecker2009} in the order of reduction. We first apply weighted \(k\)-means and then use Runnalls' algorithm to merge the components.\par
Firstly, we whiten the conditional prediction using the Cholesky decomposition of
\[
\boldsymbol{\Sigma}^*_{\rm avg}
=
\frac{1}{J}\sum_{j=1}^{J}\boldsymbol{\Sigma}^*_j
+
\frac{1}{J}\sum_{j=1}^{J}
(\boldsymbol{\mu}^*_j-\bar{\boldsymbol{\mu}}^*)
(\boldsymbol{\mu}^*_j-\bar{\boldsymbol{\mu}}^*)^\top,
\qquad
\bar{\boldsymbol{\mu}}^*
=
\frac{1}{J}\sum_{j=1}^{J}\boldsymbol{\mu}^*_j .
\]
This matrix combines the average covariance and the dispersion of the predictive means. For simplicity, we omit the outer sample index here and denote the whitened means by
\[
\mathbf u^*_j
=
\mathbf L^{-1}(\boldsymbol{\mu}_j^*-\bar{\boldsymbol{\mu}}^*),
\qquad
\mathbf L\mathbf L^\top=\boldsymbol{\Sigma}^*_{\rm avg},
\quad
j=1,\ldots,J.
\]

In the second step, we apply weighted $k$-means clustering in this whitened space to compress the original $J$ components into an initially smaller number with $J_0$ clusters. Inside each cluster, moment matching is performed to obtain a reduced GMM with $J_0$ components. Each component is characterized by a new mean $\boldsymbol{\mu}_k^\star$, covariance $\boldsymbol{\Sigma}_k^\star$, and weight $\pi_k^\star$, where $k=1,\ldots,J_0$, $\pi_k^\star\ge 0$, and $\sum_k\pi_k^\star=1$.\par
Finally, we refine the mixture of $J_0$ initial clusters by applying Runnalls' algorithm in a greedy way to further reduce the number of components to the final target value $J_{\rm target}$. The loss function is based on the KL divergence:
\[
D_{ij}
=
\frac{1}{2}
\left[
(\pi_i^\star+\pi_j^\star)\log|\boldsymbol{\Sigma}_{ij}^\star|
-
\pi_i^\star\log|\boldsymbol{\Sigma}_i^\star|
-
\pi_j^\star\log|\boldsymbol{\Sigma}_j^\star|
\right],
\]
where $\boldsymbol{\Sigma}_{ij}^\star$ is the covariance matrix obtained by moment matching between components $i$ and $j$ among the $J_0$ components. To reduce this cost, we use \(k\)-nearest neighbors (k-NN) \cite{CoverHart1967} to restrict each component to merging with nearby components. By selecting $n_{\rm knn}$ nearest neighbors per component, the pair-search cost is reduced from $O(J_0^2)$ to $O(J_0 n_{\rm knn})$.\par
The process is summarized in \cref{alg:gmm-reduction}.

\begin{algorithm}[htbp]
\caption{Gaussian mixture reduction with whitening, weighted $k$-means initialization, and Runnalls' merging}
\label{alg:gmm-reduction}
\begin{algorithmic}
\REQUIRE Original mixture $\{(\pi_j,\boldsymbol{\mu}^*_j,\boldsymbol{\Sigma}^*_j)\}_{j=1}^J$; final target size $J_{\rm target}$; initial cluster size $J_0$; k-NN size $n_{\rm knn}$; refresh period $r$
\ENSURE Reduced mixture $\{(\pi_k^\star,\boldsymbol{\mu}_k^\star,\boldsymbol{\Sigma}_k^\star)\}_{k=1}^{J_{\rm target}}$
\STATE Compute $\bar{\boldsymbol{\mu}}^*$ and $\boldsymbol{\Sigma}^*_{\rm avg}$.
\STATE Compute the Cholesky factor $\mathbf L\mathbf L^\top=\boldsymbol{\Sigma}^*_{\rm avg}$ and whiten the means $\mathbf u^*_j=\mathbf L^{-1}(\boldsymbol{\mu}_j^*-\bar{\boldsymbol{\mu}}^*)$.
\STATE Run weighted $k$-means on $\{\mathbf u^*_j\}_{j=1}^J$ to obtain $J_0$ clusters.
\STATE Perform moment matching within each cluster to form $\{(\pi_k^\star,\boldsymbol{\mu}_k^\star,\boldsymbol{\Sigma}_k^\star)\}_{k=1}^{J_0}$.
\STATE Re-whiten the $J_0$ reduced components and build k-NN neighborhoods of size $n_{\rm knn}$.
\WHILE{$J_0>J_{\rm target}$}
\STATE Evaluate $D_{ij}$ over the k-NN candidate pairs and merge $(i,j)=\arg\min D_{ij}$.
\STATE Update the merged component by moment matching.
\STATE Update $J_0\gets J_0-1$.
\STATE Every $r$ merges, re-whiten the remaining components and rebuild the k-NN neighborhoods.
\ENDWHILE
\RETURN $\{(\pi_k^\star,\boldsymbol{\mu}_k^\star,\boldsymbol{\Sigma}_k^\star)\}_{k=1}^{J_{\rm target}}$.
\end{algorithmic}
\end{algorithm}

\subsection{Schur Complement and Rank-One Update Details}
\label{supp:schur-updates}

We consider the covariance update from round $b-1$ to round $b$ after adding a candidate point $\boldsymbol{\xi}$. Let $N_o=n+m+b$ denote the size of the observation block after this candidate point is added. The enlarged observation covariance block can be written as
\[
\boldsymbol{\Sigma}_{oo}^{(b)}(\boldsymbol{\xi})
=
\begin{pmatrix}
B^{(b)}(\boldsymbol{\xi}) & \mathbf A^{(b)}(\boldsymbol{\xi})^\top\\
\mathbf A^{(b)}(\boldsymbol{\xi}) & \boldsymbol{\Sigma}_{oo}^{(b-1)}
\end{pmatrix},
\]
where $B^{(b)}(\boldsymbol{\xi})\in\mathbb R$ is the covariance of the new physical observation at $\boldsymbol{\xi}$ and
$\mathbf A^{(b)}(\boldsymbol{\xi})\in\mathbb R^{N_o-1}$ is its covariance vector with the existing observation block.

For the IMSPE and MI-based criteria, the core computation involves evaluating the predictive covariance $\boldsymbol{\Sigma}^*$. The computational bottleneck is the inversion of the observation block $\boldsymbol{\Sigma}_{oo}^{(b)}(\boldsymbol{\xi})$. For each posterior draw, each round, and each candidate design point, the computational complexity is
\[
O\!\left(
N_o^3
+
N_o^2 n^*
+
N_o {n^*}^2
\right)
\]
if one directly recomputes this inverse, where $N_o^3$ corresponds to the Cholesky decomposition of the observation covariance block and the other two terms correspond to the triangular solves.

To reduce this computational burden, we first precompute the full covariance matrix over all observation points, candidate points, and prediction inputs for all posterior draws. At each round, the corresponding submatrices required by the criterion are then obtained by slicing the precomputed kernel matrix. Since the covariance structure expands by adding one candidate input at each round, we apply the Schur complement together with a rank-one update \cite{ShermanMorrison1950,Gramacy2020Surrogates} to update the required matrices.

Assume that $\bigl(\boldsymbol{\Sigma}_{oo}^{(b-1)}\bigr)^{-1}$ is already available. Define
\[
\mathbf u^{(b)}(\boldsymbol{\xi})
=
\bigl(\boldsymbol{\Sigma}_{oo}^{(b-1)}\bigr)^{-1}
\mathbf A^{(b)}(\boldsymbol{\xi}),
\qquad
\lambda^{(b)}(\boldsymbol{\xi})
=
\left[
B^{(b)}(\boldsymbol{\xi})
-
\mathbf A^{(b)}(\boldsymbol{\xi})^\top
\mathbf u^{(b)}(\boldsymbol{\xi})
\right]^{-1}.
\]
Then the Schur complement formula gives
\[
\bigl(\boldsymbol{\Sigma}_{oo}^{(b)}(\boldsymbol{\xi})\bigr)^{-1}
=
\begin{pmatrix}
\lambda^{(b)}(\boldsymbol{\xi})
&
-\lambda^{(b)}(\boldsymbol{\xi})\mathbf u^{(b)}(\boldsymbol{\xi})^\top
\\
-\lambda^{(b)}(\boldsymbol{\xi})\mathbf u^{(b)}(\boldsymbol{\xi})
&
\bigl(\boldsymbol{\Sigma}_{oo}^{(b-1)}\bigr)^{-1}
+
\lambda^{(b)}(\boldsymbol{\xi})
\mathbf u^{(b)}(\boldsymbol{\xi})
\mathbf u^{(b)}(\boldsymbol{\xi})^\top
\end{pmatrix}.
\]
This update requires $O(N_o^2)$ operations for each posterior draw $\boldsymbol{\Omega}^{(j)}$.

\subsubsection{Rank-one update for the predictive covariance}
\label{supp:rankone-update}

At round $b-1$, the predictive covariance matrix is
\[
\boldsymbol{\Sigma}^*_{(b-1)}
=
\boldsymbol{\Sigma}_{**}
-
\boldsymbol{\Sigma}_{*o^{(b-1)}}
\bigl(\boldsymbol{\Sigma}_{oo}^{(b-1)}\bigr)^{-1}
\boldsymbol{\Sigma}_{o^{(b-1)}*}.
\]
After adding the candidate point $\boldsymbol{\xi}$, the predictive covariance becomes
\[
\boldsymbol{\Sigma}^*_{(b)}(\boldsymbol{\xi})
=
\boldsymbol{\Sigma}_{**}
-
\boldsymbol{\Sigma}_{*o^{(b)}}
\bigl(\boldsymbol{\Sigma}_{oo}^{(b)}(\boldsymbol{\xi})\bigr)^{-1}
\boldsymbol{\Sigma}_{o^{(b)}*}.
\]
Based on the block inverse derived above, this predictive covariance can be efficiently updated using a rank-one Sherman--Morrison update. Define
\[
\mathbf v^{(b)}(\boldsymbol{\xi})
=
\boldsymbol{\Sigma}_{*\xi}
-
\boldsymbol{\Sigma}_{*o^{(b-1)}}
\mathbf u^{(b)}(\boldsymbol{\xi}),
\]
where $\boldsymbol{\Sigma}_{*\xi}\in\mathbb R^{n^*}$ is the covariance vector between the prediction points and the candidate point. The vector $\mathbf v^{(b)}(\boldsymbol{\xi})$ is the conditional covariance between the prediction points and the new candidate point given the existing observation set. Substituting the block inverse into the predictive covariance formula yields
\[
\boldsymbol{\Sigma}^*_{(b)}(\boldsymbol{\xi})
=
\boldsymbol{\Sigma}^*_{(b-1)}
-
\lambda^{(b)}(\boldsymbol{\xi})
\mathbf v^{(b)}(\boldsymbol{\xi})
\mathbf v^{(b)}(\boldsymbol{\xi})^\top .
\]
This update requires computing the matrix--vector product
$\boldsymbol{\Sigma}_{*o^{(b-1)}}\mathbf u^{(b)}(\boldsymbol{\xi})$ and the outer product
$\mathbf v^{(b)}(\boldsymbol{\xi})\mathbf v^{(b)}(\boldsymbol{\xi})^\top$, so the cost for the rank-one update is
\[
O\bigl(N_o n^*+{n^*}^2\bigr)
\]
for each posterior draw $\boldsymbol{\Omega}^{(j)}$.

\subsubsection{Total computational complexity}
\label{supp:schur-total-complexity}

Combining the Schur complement update for
$\bigl(\boldsymbol{\Sigma}_{oo}^{(b)}(\boldsymbol{\xi})\bigr)^{-1}$
and the rank-one update for $\boldsymbol{\Sigma}^*_{(b)}(\boldsymbol{\xi})$, the computational cost for each candidate design point at each round under one posterior draw $\boldsymbol{\Omega}^{(j)}$ is
\[
O\!\left(
N_o^2
+
N_o n^*
+
{n^*}^2
\right).
\]

\raggedbottom

\section{Additional Results for Numerical Experiments}
\label{supp:additional-results}
The one-dimensional example was run on a local computer with an Intel
Core i7-11800H CPU, 16~GB RAM, and an NVIDIA RTX 3070 Ti GPU. The
two- and three-dimensional examples were run on the MUSICA
high-performance computing system equipped with an NVIDIA H100 GPU with 96~GB of memory.

\subsection{One-Dimensional Localized-Oscillation Example}
\label{supp:one-dimensional-example}

Let $x\in[0,1]$ denote the controllable input and let $\boldsymbol{\theta}=(a,b)^\top$ denote the calibration parameters. The computer model is defined as
\begin{equation}
\label{eq:one-d-simulator}
\eta(x,\boldsymbol{\theta})
=
1.15\tanh\!\left(
4\left[x-0.34-0.10(a-1.15)\right]
\right)
+
0.30\frac{b}{1.80}x
+
0.12\sin(2\pi x).
\end{equation}
We assign independent $\operatorname{LogNormal}(0,0.6)$ priors to $a$ and $b$. The simulator GP uses an ARD squared-exponential kernel with inverse squared length scales $\operatorname{LogNormal}(0,0.5)$ and variance $\operatorname{Gamma}(5,0.5)$. The discrepancy GP uses a locally periodic kernel with inverse-scale parameters $\operatorname{LogNormal}(0,0.6)$, variance $\operatorname{Gamma}(5,1.5)$, and period $\operatorname{Uniform}(0.25,4)$. We define the model discrepancy as 
\begin{equation}
\label{eq:one-d-discrepancy}
\delta(x)
=
0.72
\exp\!\left[
-\left(\frac{x-0.72}{0.17}\right)^8
\right]
\sin\!\left(18\pi(x-0.72)\right).
\end{equation}
Accordingly, the physical process is given
\begin{equation}
\label{eq:one-d-physical-response}
\zeta(x)
=
\eta(x,\boldsymbol{\theta}_0)+\delta(x),
\qquad
y(x)=\zeta(x)+\epsilon,
\qquad
\epsilon\sim\mathcal{N}(0,0.02^2).
\end{equation}
The discrepancy in \cref{eq:one-d-discrepancy} shows localized oscillations concentrated around $x=0.72$.\par

\begin{figure}[!htbp]
\centering

\begin{minipage}[t]{0.485\textwidth}
\centering
\textbf{(a) Sequential MI}\par\vspace{1pt}
\includegraphics[
  width=\linewidth,
  height=0.17\textheight,
  keepaspectratio
]{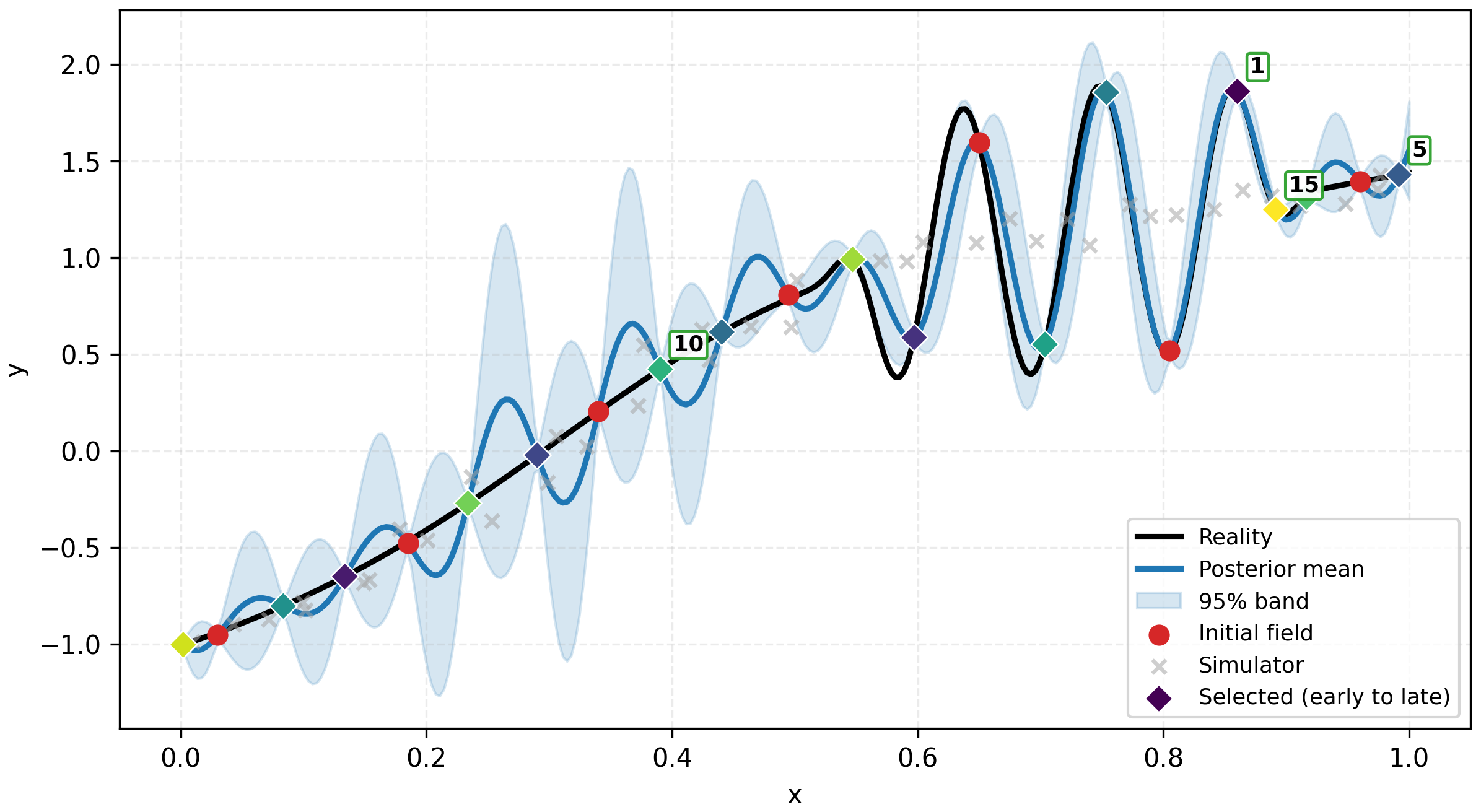}
\end{minipage}\hfill
\begin{minipage}[t]{0.485\textwidth}
\centering
\textbf{(b) Adaptive MI}\par\vspace{1pt}
\includegraphics[
  width=\linewidth,
  height=0.17\textheight,
  keepaspectratio
]{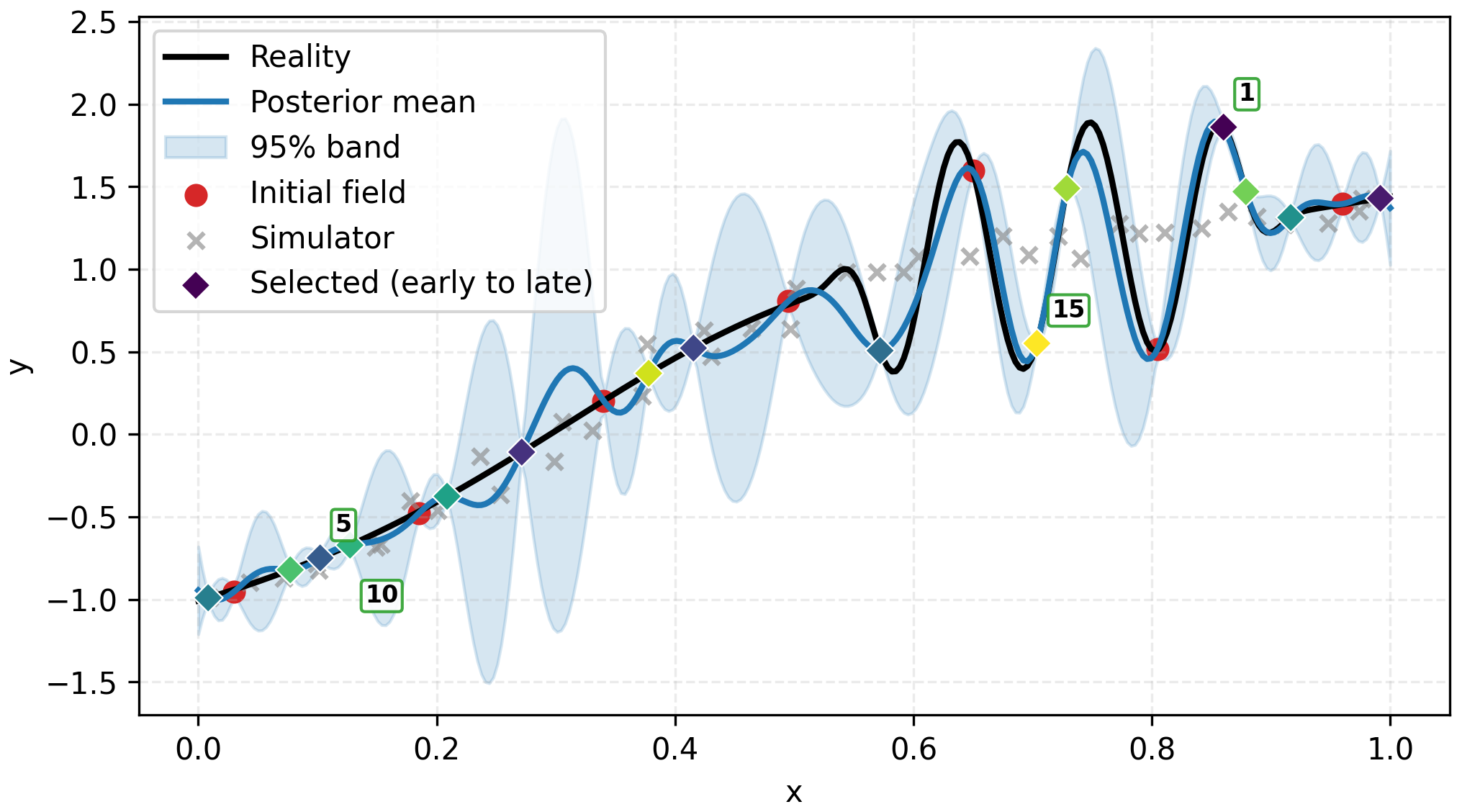}
\end{minipage}

\vspace{3pt}

\begin{minipage}[t]{0.485\textwidth}
\centering
\textbf{(c) Maximin distance}\par\vspace{1pt}
\includegraphics[
  width=\linewidth,
  height=0.17\textheight,
  keepaspectratio
]{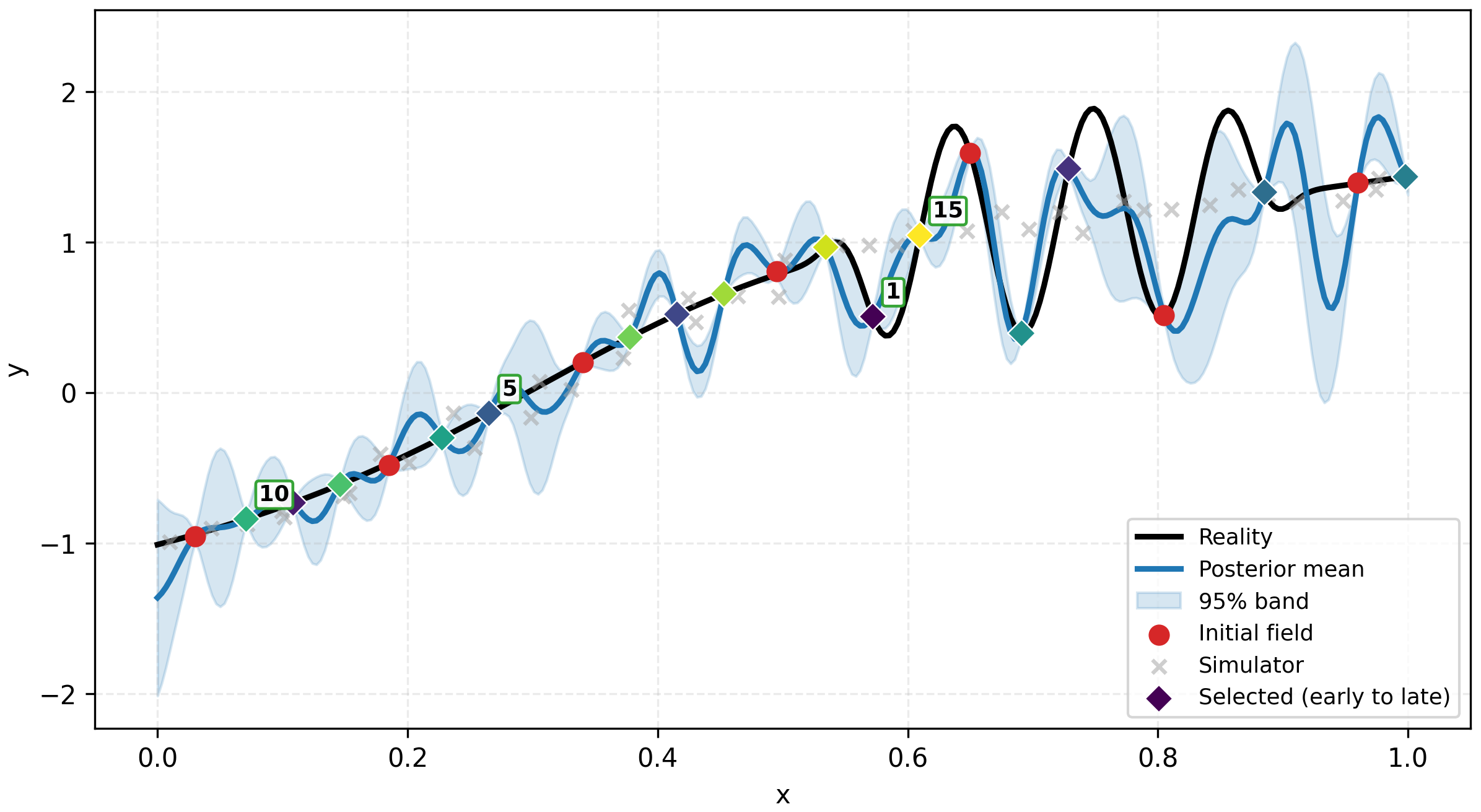}
\end{minipage}\hfill
\begin{minipage}[t]{0.485\textwidth}
\centering
\textbf{(d) Adaptive MI--CX}\par\vspace{1pt}
\includegraphics[
  width=\linewidth,
  height=0.17\textheight,
  keepaspectratio
]{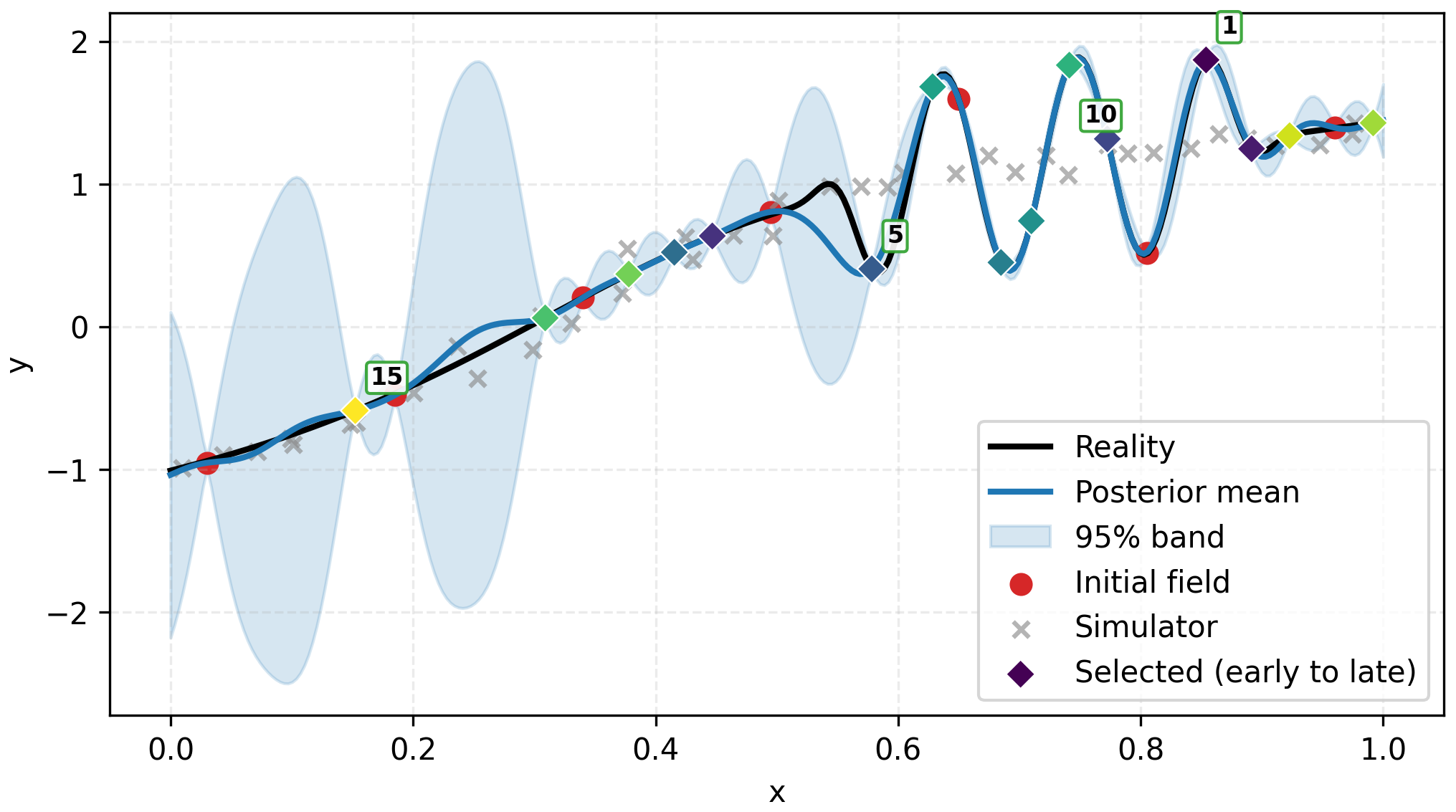}
\end{minipage}

\vspace{3pt}

\begin{minipage}[t]{0.485\textwidth}
\centering
\textbf{(e) Sequential IMSPE}\par\vspace{1pt}
\includegraphics[
  width=\linewidth,
  height=0.17\textheight,
  keepaspectratio
]{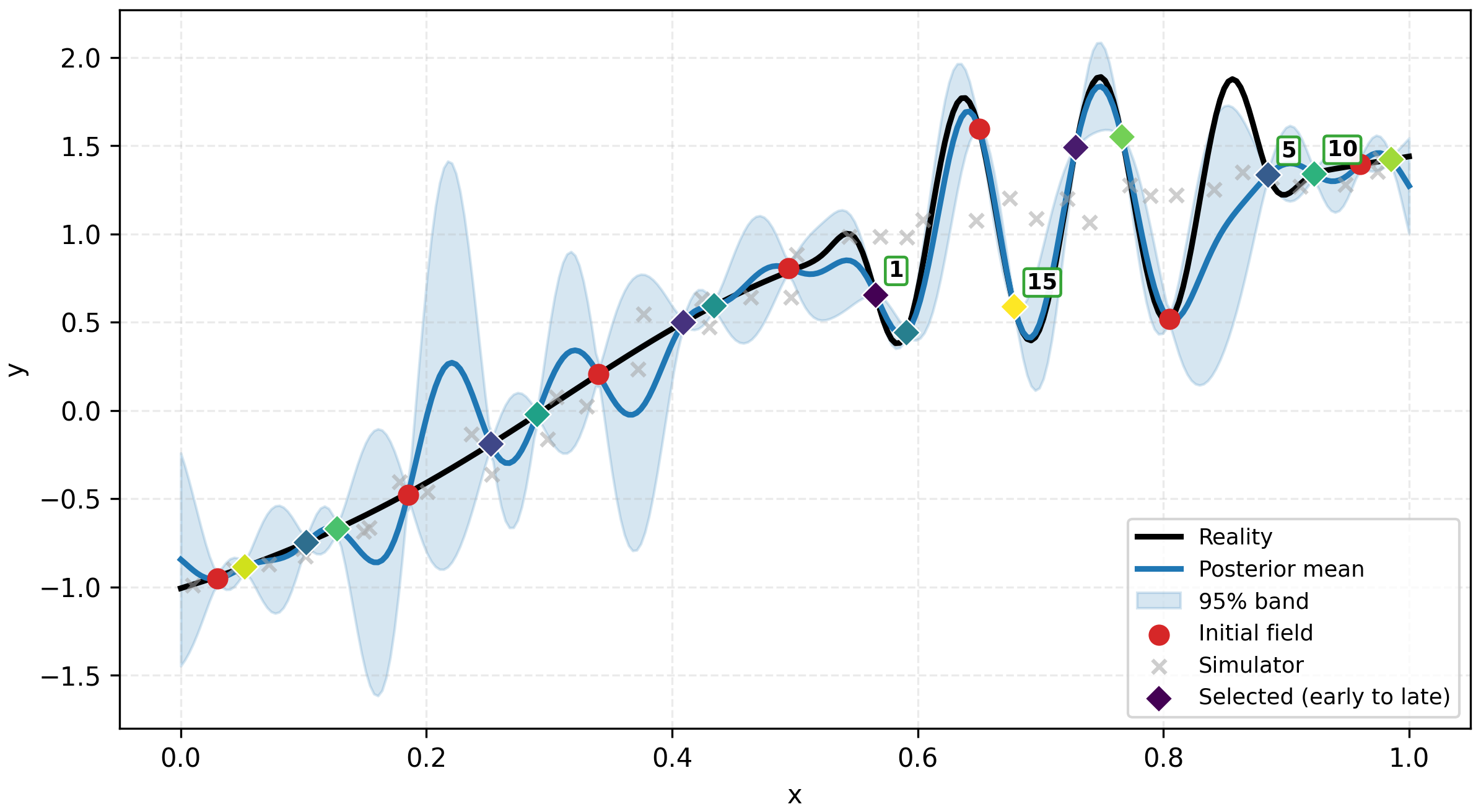}
\end{minipage}\hfill
\begin{minipage}[t]{0.485\textwidth}
\centering
\textbf{(f) Adaptive IMSPE}\par\vspace{1pt}
\includegraphics[
  width=\linewidth,
  height=0.17\textheight,
  keepaspectratio
]{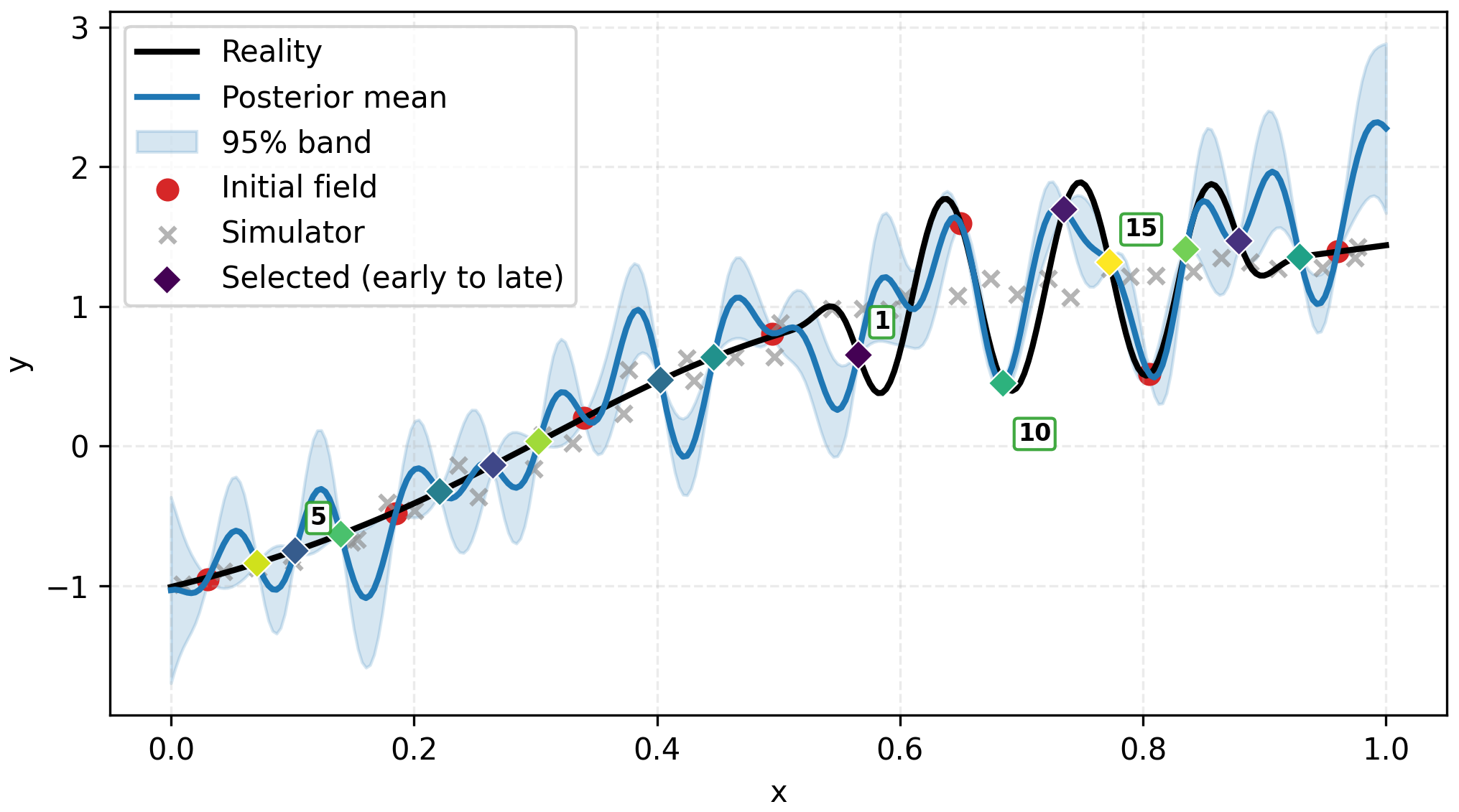}
\end{minipage}

\vspace{3pt}

\begin{minipage}[t]{0.485\textwidth}
\centering
\textbf{(g) Sequential D-optimality}\par\vspace{1pt}
\includegraphics[
  width=\linewidth,
  height=0.17\textheight,
  keepaspectratio
]{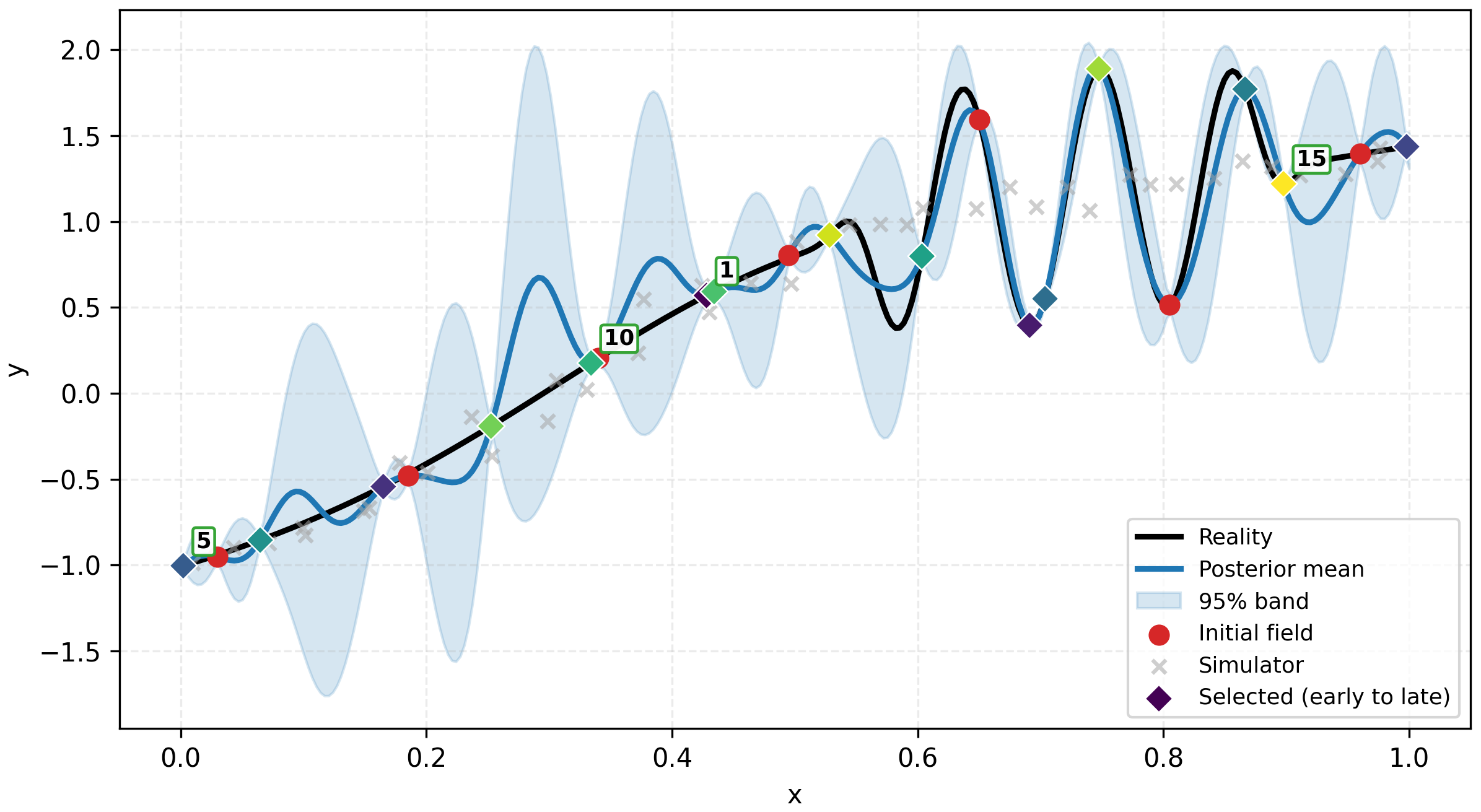}
\end{minipage}\hfill
\begin{minipage}[t]{0.485\textwidth}
\centering
\textbf{(h) Adaptive D-optimality}\par\vspace{1pt}
\includegraphics[
  width=\linewidth,
  height=0.17\textheight,
  keepaspectratio
]{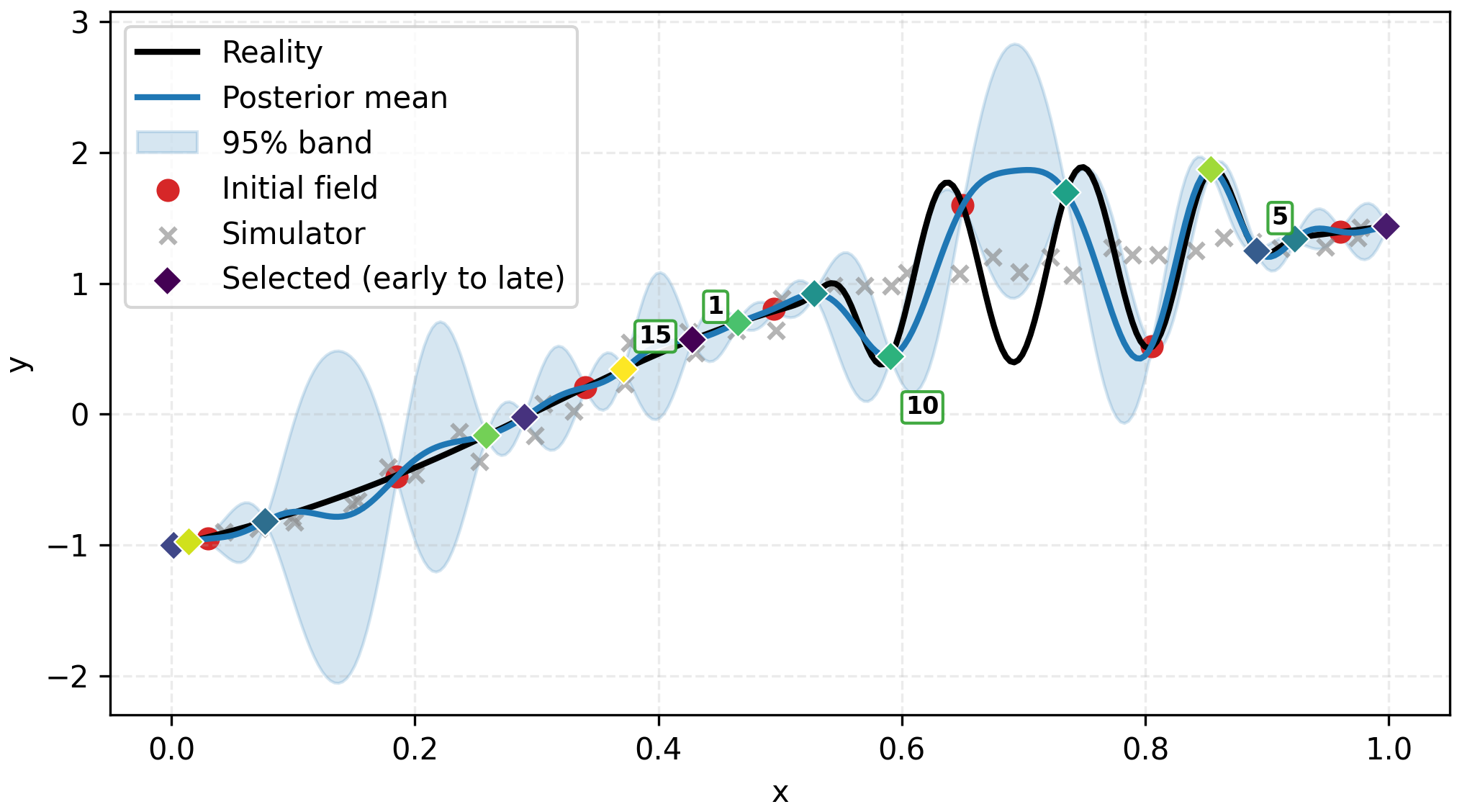}
\end{minipage}

\caption{Final posterior predictive fits of all design criteria in the
one-dimensional example.}
\label{fig:supp-one-d-final-fits}
\end{figure}

\clearpage

\begin{figure}[htbp]
\centering
\includegraphics[width=\textwidth]{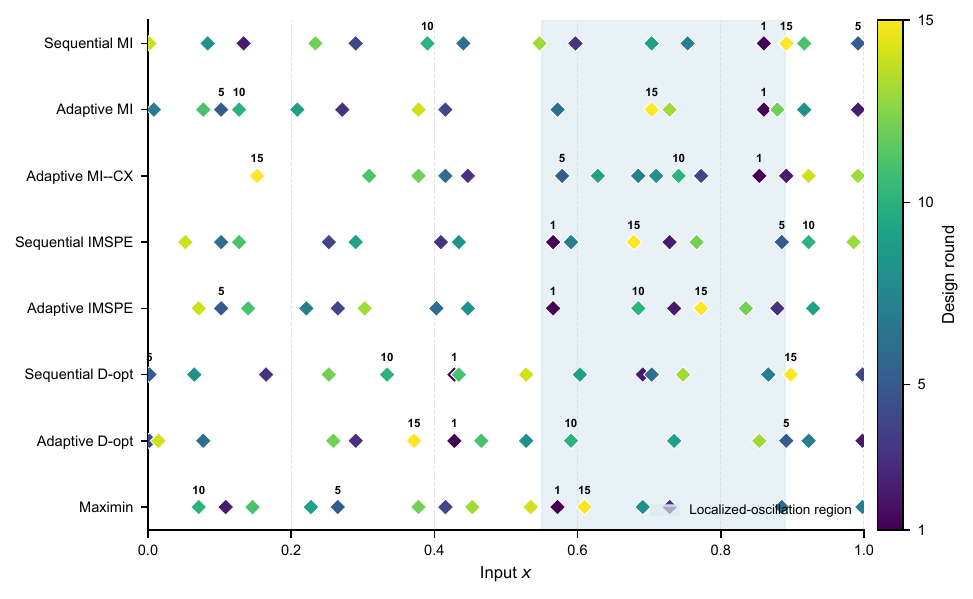}
\caption{Selected physical inputs in the
one-dimensional example. The colors indicate the
design round, and the shaded region indicates the region containing the localized oscillations.}
\label{fig:supp-one-d-selected-locations}
\end{figure}

\begin{table}[H]
\scriptsize
\caption{MSE at each design round in the one-dimensional example.}
\label{tab:supp-one-d-mse}
\begin{center}
\resizebox{\textwidth}{!}{%
\begin{tabular}{ccccccccc}
\hline
Round & S-MI & A-MI & A-MI--CX & S-IMSPE & A-IMSPE
& S-D-opt & A-D-opt & Maximin \\
\hline
0  & 0.1638 & 0.1638 & 0.1638 & 0.1638 & 0.1638 & 0.1638 & 0.1638 & 0.1638 \\
1  & 0.1285 & 0.1381 & 0.1475 & 0.1253 & 0.1300 & 0.1702 & 0.1685 & 0.1252 \\
2  & 0.1323 & 0.1581 & 0.1176 & 0.1441 & 0.1759 & 0.2694 & 0.1424 & 0.1140 \\
3  & 0.1381 & 0.1232 & 0.1249 & 0.1375 & 0.1456 & 0.2340 & 0.1638 & 0.1314 \\
4  & 0.1395 & 0.1142 & 0.0676 & 0.1566 & 0.1107 & 0.2754 & 0.1652 & 0.1076 \\
5  & 0.1217 & 0.1140 & 0.0901 & 0.1186 & 0.1027 & 0.2654 & 0.1571 & 0.1523 \\
6  & 0.1095 & 0.1157 & 0.0445 & 0.1178 & 0.1071 & 0.0779 & 0.1606 & 0.1105 \\
7  & 0.0873 & 0.1010 & 0.0454 & 0.1148 & 0.1149 & 0.0619 & 0.1645 & 0.1334 \\
8  & 0.0968 & 0.1086 & 0.0167 & 0.1269 & 0.1016 & 0.0766 & 0.1458 & 0.0852 \\
9  & 0.0226 & 0.1013 & 0.0188 & 0.1288 & 0.1288 & 0.0704 & 0.1523 & 0.1064 \\
10 & 0.0272 & 0.0925 & 0.0215 & 0.1433 & 0.0244 & 0.0322 & 0.1859 & 0.1277 \\
11 & 0.0171 & 0.0819 & 0.0141 & 0.1528 & 0.0394 & 0.0588 & 0.2157 & 0.1095 \\
12 & 0.0218 & 0.1095 & 0.0102 & 0.0897 & 0.0185 & 0.0925 & 0.1771 & 0.1146 \\
13 & 0.0573 & 0.0981 & 0.0099 & 0.0827 & 0.1003 & 0.0708 & 0.1106 & 0.1216 \\
14 & 0.0583 & 0.0750 & 0.0080 & 0.1016 & 0.0573 & 0.0141 & 0.1057 & 0.1224 \\
15 & 0.0288 & 0.0169 & \textbf{0.0085} & 0.0443 & 0.1084 & 0.0355 & 0.1099 & 0.0860 \\
\hline
\end{tabular}}
\end{center}
\end{table}

\begin{table}[H]
\scriptsize
\caption{CRPS at each design round in the one-dimensional example.}
\label{tab:supp-one-d-crps}
\begin{center}
\resizebox{\textwidth}{!}{%
\begin{tabular}{ccccccccc}
\hline
Round & S-MI & A-MI & A-MI--CX & S-IMSPE & A-IMSPE
& S-D-opt & A-D-opt & Maximin \\
\hline
0  & 0.2089 & 0.2089 & 0.2089 & 0.2089 & 0.2089 & 0.2089 & 0.2089 & 0.2089 \\
1  & 0.2007 & 0.2203 & 0.2227 & 0.1926 & 0.1987 & 0.1865 & 0.2076 & 0.1927 \\
2  & 0.2104 & 0.2176 & 0.1890 & 0.2193 & 0.2255 & 0.2608 & 0.1784 & 0.1745 \\
3  & 0.2201 & 0.1839 & 0.1934 & 0.1902 & 0.2087 & 0.2221 & 0.1984 & 0.1908 \\
4  & 0.2244 & 0.1671 & 0.1544 & 0.1961 & 0.1793 & 0.2355 & 0.1974 & 0.1670 \\
5  & 0.1728 & 0.1578 & 0.1904 & 0.1744 & 0.1650 & 0.2484 & 0.1906 & 0.1843 \\
6  & 0.1572 & 0.1931 & 0.1601 & 0.1673 & 0.1567 & 0.1585 & 0.1909 & 0.1584 \\
7  & 0.1467 & 0.1735 & 0.1569 & 0.1685 & 0.1518 & 0.1678 & 0.1960 & 0.1773 \\
8  & 0.1608 & 0.1949 & 0.1120 & 0.1768 & 0.1575 & 0.1537 & 0.1816 & 0.1811 \\
9  & 0.0992 & 0.1808 & 0.0930 & 0.1634 & 0.1579 & 0.1365 & 0.1807 & 0.1867 \\
10 & 0.1117 & 0.1705 & 0.1233 & 0.1739 & 0.0978 & 0.1290 & 0.2018 & 0.2129 \\
11 & 0.0952 & 0.1600 & 0.0873 & 0.1679 & 0.1265 & 0.1258 & 0.2151 & 0.1690 \\
12 & 0.0966 & 0.1535 & 0.0811 & 0.1374 & 0.0851 & 0.1311 & 0.1922 & 0.1880 \\
13 & 0.1489 & 0.1432 & 0.0767 & 0.1279 & 0.1882 & 0.1546 & 0.1652 & 0.2070 \\
14 & 0.1644 & 0.1105 & 0.0791 & 0.1419 & 0.1262 & 0.0936 & 0.1264 & 0.1895 \\
15 & 0.1092 & 0.0771 & \textbf{0.0640} & 0.0987 & 0.2215 & 0.0944 & 0.1423 & 0.1792 \\
\hline
\end{tabular}}
\end{center}
\end{table}

\begin{table}[H]
\footnotesize
\caption{Acquisition time of the SDE criteria over 15 design rounds.}
\label{tab:supp-one-d-runtime}
\begin{center}
\begin{tabular}{lcc}
\hline
Design criterion & Total acquisition time (s) & Time per round (s) \\
\hline
Sequential MI           & 275.817 & 18.388 \\
Sequential IMSPE        & 12.051  & 0.803  \\
Sequential D-optimality & 22.731  & 1.515  \\
Maximin distance        & 0.0021  & 0.00014 \\
\hline
\end{tabular}
\end{center}
\end{table}

\subsection{Two-Dimensional Plume Sensor-Placement Example}
\label{supp:two-dimensional-example}
Let
$\mathbf{x}=(x_1,x_2)^\top\in[0,1]^2$ denote a normalized
sensor location within the space
$[0.05,1.00]\times[0.25,0.75]$, with
\[
\widetilde{x}_1=0.05+0.95x_1,
\qquad
\widetilde{x}_2=0.25+0.50x_2.
\]
The normalized calibration vector
$\boldsymbol{\theta}=(\theta_1,\theta_2,\theta_3)^\top\in[0,1]^3$
determines the source location, wind direction, and spreading coefficient as
\[
y_{\mathrm{src}}=0.35+0.20\theta_1,
\qquad
\alpha=-0.15+0.30\theta_2,
\qquad
\kappa=0.07+0.10\theta_3.
\]
Define
\[
\begin{aligned}
s&=(\widetilde{x}_1-0.10)\cos\alpha
 +(\widetilde{x}_2-y_{\mathrm{src}})\sin\alpha,\\
r&=-(\widetilde{x}_1-0.10)\sin\alpha
 +(\widetilde{x}_2-y_{\mathrm{src}})\cos\alpha,\\
q(s)&=\frac{\log\{1+\exp(35s)\}}{35},
\qquad
g(s)=\{1+\exp(-45s)\}^{-1}.
\end{aligned}
\]
Following the standard Gaussian plume formulation in \cite{Stockie2011}, we define the computer model as
\begin{equation}
\label{eq:plume-simulator}
\eta(\mathbf{x},\boldsymbol{\theta})
=
g(s)\frac{\exp\{-q(s)/0.88\}}{1+1.20q(s)}
\exp\!\left[
-\frac{1}{2}
\left\{
\frac{r}{0.022+\kappa\sqrt{q(s)+0.012}}
\right\}^{2}
\right].
\end{equation}
The true calibration values are
$\boldsymbol{\theta}_0=(0.46,0.62,0.40)^\top$. Our physical model is a  simplified Gaussian plume containing a waving plume centerline, spatially varying plume width, and a weak secondary plume component. Specifically,
\[
\begin{aligned}
m(q)
&=0.060\sin\{2\pi(0.95q+0.04)\}
 +0.022\sin\{2\pi(2.85q+0.17)\},\\
w(q)
&=1+0.18\sin\{2\pi(0.72q+0.11)\},\\
\sigma_\zeta(q)
&=0.022+\kappa_0\sqrt{q+0.012}\,w(q),\\
c(q)
&=m(q)+0.105+0.020\sin(2\pi q),
\end{aligned}
\]
where $s$, $r$, $q$, and $g$ are evaluated at
$\boldsymbol{\theta}_0$, and $\kappa_0=0.11$. The physical process then gives
\begin{equation}
\label{eq:plume-physical-response}
\begin{aligned}
\zeta(\mathbf{x})
&=
g(s)\frac{\exp\{-q/(0.88\times0.97)\}}{1+1.18q}
\exp\!\left[-\frac{1}{2}
\left\{\frac{r-m(q)}{\sigma_\zeta(q)}\right\}^{2}\right]\\
&\quad
+0.13g(s)\exp(-q/0.55)
\exp\!\left[-\frac{1}{2}
\left\{\frac{r-c(q)}{1.18\sigma_\zeta(q)}\right\}^{2}\right],\\
y(\mathbf{x})
&=\zeta(\mathbf{x})+\epsilon,
\qquad
\epsilon\sim\mathcal{N}(0,0.015^2).
\end{aligned}
\end{equation}
We assign independent uniform priors on $[0,1]$ to the normalized calibration parameters. Both GPs use ARD squared-exponential kernels. For the simulator GP, the inverse squared length scales follow
$\operatorname{LogNormal}(2.5,0.8)$ and the variance follows
$\operatorname{Gamma}(3,3)$. For the discrepancy GP, the corresponding
priors are $\operatorname{LogNormal}(3,0.7)$ and
$\operatorname{Gamma}(2,8)$.

\begin{figure}[p]
\centering
\includegraphics[
  width=\textwidth,
  height=0.84\textheight,
  keepaspectratio
]{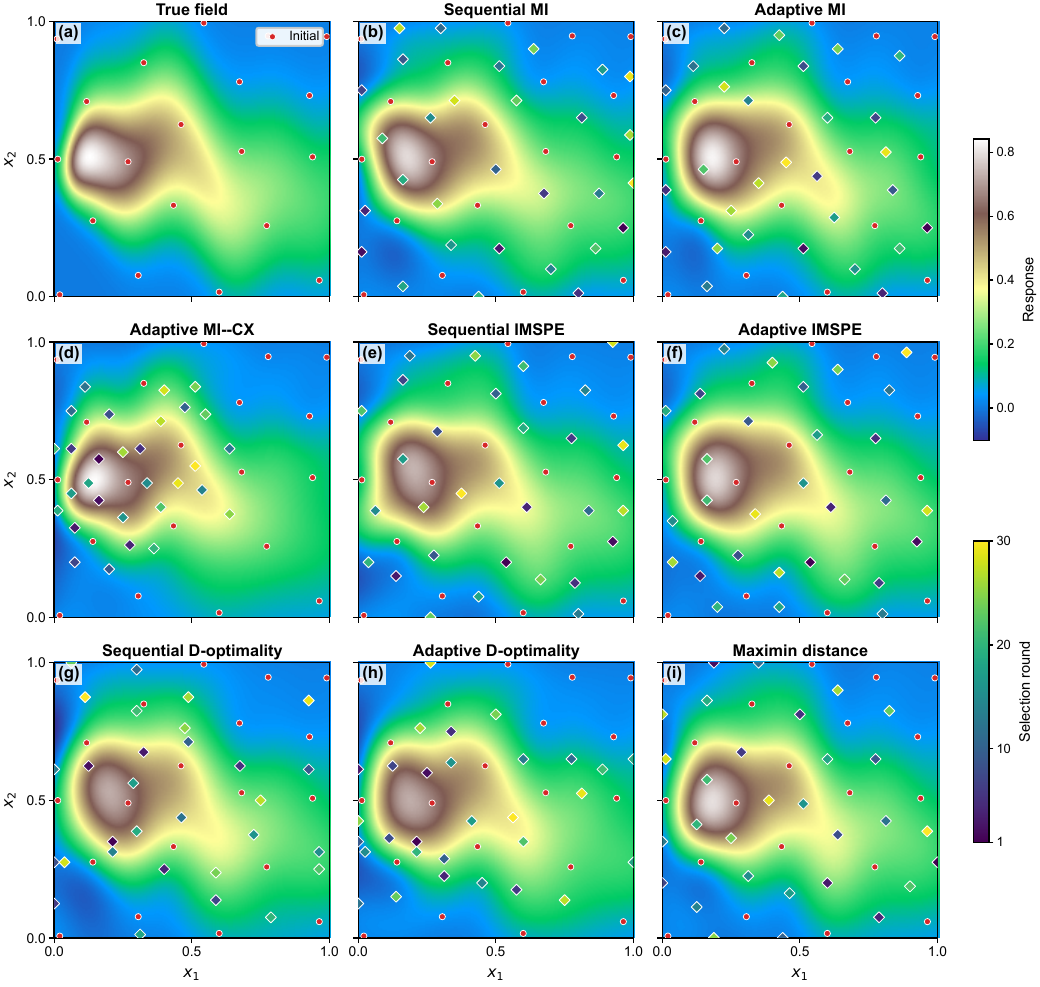}
\caption{True physical field and final posterior predictive means after
30 design rounds in the two-dimensional plume example. Red dots denote the initial physical observations,
and diamonds denote the selected physical inputs, colored by design round.}
\label{fig:supp-plume-final-means}
\end{figure}

\clearpage

\begin{table}[p]
\scriptsize
\caption{MSE at each design round in the two-dimensional plume example.}
\label{tab:supp-plume-mse}
\begin{center}
\resizebox{\textwidth}{!}{%
\begin{tabular}{ccccccccc}
\hline
Round & S-MI & A-MI & A-MI--CX & S-IMSPE & A-IMSPE
& S-D-opt & A-D-opt & Maximin \\
\hline
0  & 0.004726 & 0.004726 & 0.004726 & 0.004726 & 0.004726 & 0.004726 & 0.004726 & 0.004726 \\
1  & 0.004713 & 0.004713 & 0.004302 & 0.004626 & 0.004626 & 0.003801 & 0.003801 & 0.004723 \\
2  & 0.004814 & 0.004720 & 0.002180 & 0.004604 & 0.004604 & 0.003320 & 0.003916 & 0.004661 \\
3  & 0.004799 & 0.004779 & 0.002086 & 0.004584 & 0.004584 & 0.003289 & 0.003521 & 0.004423 \\
4  & 0.004879 & 0.004808 & 0.002050 & 0.004475 & 0.004475 & 0.003287 & 0.003487 & 0.004433 \\
5  & 0.004889 & 0.004694 & 0.002044 & 0.004493 & 0.004420 & 0.003150 & 0.003627 & 0.004423 \\
6  & 0.004688 & 0.004731 & 0.002060 & 0.004414 & 0.004428 & 0.003013 & 0.003505 & 0.004470 \\
7  & 0.004542 & 0.004683 & 0.001680 & 0.004411 & 0.004421 & 0.002864 & 0.001913 & 0.004553 \\
8  & 0.004462 & 0.004709 & 0.001639 & 0.004519 & 0.004417 & 0.002923 & 0.001812 & 0.004324 \\
9  & 0.004470 & 0.004643 & 0.001574 & 0.004428 & 0.004269 & 0.002915 & 0.001655 & 0.004300 \\
10 & 0.004455 & 0.004438 & 0.001557 & 0.004422 & 0.004309 & 0.002898 & 0.001554 & 0.004253 \\
11 & 0.004460 & 0.004436 & 0.001551 & 0.004385 & 0.004316 & 0.002919 & 0.001568 & 0.004249 \\
12 & 0.004479 & 0.004374 & 0.001476 & 0.004319 & 0.004319 & 0.002940 & 0.001329 & 0.004254 \\
13 & 0.004496 & 0.004468 & 0.001348 & 0.004303 & 0.004285 & 0.001857 & 0.001271 & 0.004291 \\
14 & 0.004365 & 0.004308 & 0.001254 & 0.004304 & 0.004304 & 0.001806 & 0.001099 & 0.004313 \\
15 & 0.004069 & 0.004161 & 0.001257 & 0.003620 & 0.004137 & 0.001763 & 0.001110 & 0.004343 \\
16 & 0.004016 & 0.004046 & 0.001245 & 0.003586 & 0.004025 & 0.001658 & 0.001103 & 0.004258 \\
17 & 0.003995 & 0.003633 & 0.001028 & 0.003032 & 0.004018 & 0.001521 & 0.001062 & 0.003905 \\
18 & 0.004032 & 0.003615 & 0.000687 & 0.002969 & 0.003963 & 0.001512 & 0.001029 & 0.003897 \\
19 & 0.001731 & 0.003648 & 0.000630 & 0.002932 & 0.003649 & 0.001503 & 0.001017 & 0.003888 \\
20 & 0.001681 & 0.003660 & 0.000597 & 0.003152 & 0.003644 & 0.001509 & 0.001013 & 0.001342 \\
21 & 0.001371 & 0.003652 & 0.000612 & 0.003185 & 0.001740 & 0.001542 & 0.000995 & 0.000971 \\
22 & 0.001261 & 0.000827 & 0.000582 & 0.002559 & 0.001127 & 0.001543 & 0.000903 & 0.000999 \\
23 & 0.001228 & 0.000803 & 0.000532 & 0.002360 & 0.001057 & 0.001536 & 0.000848 & 0.000977 \\
24 & 0.001192 & 0.000777 & 0.000497 & 0.002101 & 0.000993 & 0.001487 & 0.000839 & 0.000935 \\
25 & 0.001121 & 0.000734 & 0.000432 & 0.002134 & 0.001054 & 0.001440 & 0.000803 & 0.000929 \\
26 & 0.001138 & 0.000674 & 0.000408 & 0.001592 & 0.000958 & 0.001439 & 0.000795 & 0.000850 \\
27 & 0.001137 & 0.000669 & 0.000395 & 0.001714 & 0.000921 & 0.001416 & 0.000786 & 0.000842 \\
28 & 0.000987 & 0.000619 & 0.000394 & 0.001735 & 0.000891 & 0.001425 & 0.000782 & 0.000569 \\
29 & 0.001032 & 0.000614 & 0.000403 & 0.001574 & 0.000880 & 0.001427 & 0.000746 & 0.000562 \\
30 & 0.000969 & 0.000597 & \textbf{0.000397} & 0.001513 & 0.000884 & 0.001395 & 0.000745 & 0.000546 \\
\hline
\end{tabular}}
\end{center}
\end{table}

\begin{table}[p]
\scriptsize
\caption{CRPS at each design round in the two-dimensional plume example.}
\label{tab:supp-plume-crps}
\begin{center}
\resizebox{\textwidth}{!}{%
\begin{tabular}{ccccccccc}
\hline
Round & S-MI & A-MI & A-MI--CX & S-IMSPE & A-IMSPE
& S-D-opt & A-D-opt & Maximin \\
\hline
0  & 0.0356 & 0.0356 & 0.0356 & 0.0356 & 0.0356 & 0.0356 & 0.0356 & 0.0356 \\
1  & 0.0363 & 0.0363 & 0.0376 & 0.0365 & 0.0365 & 0.0342 & 0.0342 & 0.0361 \\
2  & 0.0370 & 0.0368 & 0.0368 & 0.0364 & 0.0364 & 0.0351 & 0.0351 & 0.0361 \\
3  & 0.0366 & 0.0364 & 0.0349 & 0.0353 & 0.0353 & 0.0329 & 0.0348 & 0.0344 \\
4  & 0.0365 & 0.0361 & 0.0347 & 0.0345 & 0.0345 & 0.0328 & 0.0352 & 0.0340 \\
5  & 0.0370 & 0.0363 & 0.0362 & 0.0348 & 0.0340 & 0.0346 & 0.0374 & 0.0348 \\
6  & 0.0336 & 0.0346 & 0.0341 & 0.0324 & 0.0321 & 0.0318 & 0.0337 & 0.0339 \\
7  & 0.0352 & 0.0355 & 0.0332 & 0.0336 & 0.0333 & 0.0321 & 0.0333 & 0.0357 \\
8  & 0.0333 & 0.0355 & 0.0339 & 0.0333 & 0.0315 & 0.0313 & 0.0341 & 0.0323 \\
9  & 0.0339 & 0.0341 & 0.0321 & 0.0317 & 0.0306 & 0.0304 & 0.0328 & 0.0318 \\
10 & 0.0329 & 0.0329 & 0.0296 & 0.0316 & 0.0315 & 0.0291 & 0.0308 & 0.0307 \\
11 & 0.0325 & 0.0328 & 0.0309 & 0.0317 & 0.0310 & 0.0300 & 0.0302 & 0.0312 \\
12 & 0.0328 & 0.0321 & 0.0293 & 0.0309 & 0.0305 & 0.0297 & 0.0294 & 0.0306 \\
13 & 0.0329 & 0.0330 & 0.0284 & 0.0310 & 0.0306 & 0.0294 & 0.0276 & 0.0312 \\
14 & 0.0328 & 0.0331 & 0.0271 & 0.0307 & 0.0307 & 0.0284 & 0.0270 & 0.0312 \\
15 & 0.0297 & 0.0305 & 0.0266 & 0.0306 & 0.0292 & 0.0275 & 0.0264 & 0.0300 \\
16 & 0.0297 & 0.0306 & 0.0269 & 0.0305 & 0.0294 & 0.0258 & 0.0263 & 0.0313 \\
17 & 0.0297 & 0.0278 & 0.0313 & 0.0296 & 0.0294 & 0.0267 & 0.0266 & 0.0281 \\
18 & 0.0287 & 0.0265 & 0.0304 & 0.0288 & 0.0280 & 0.0252 & 0.0241 & 0.0272 \\
19 & 0.0238 & 0.0268 & 0.0304 & 0.0282 & 0.0275 & 0.0249 & 0.0263 & 0.0276 \\
20 & 0.0225 & 0.0264 & 0.0285 & 0.0295 & 0.0269 & 0.0243 & 0.0238 & 0.0234 \\
21 & 0.0227 & 0.0265 & 0.0300 & 0.0299 & 0.0224 & 0.0238 & 0.0249 & 0.0231 \\
22 & 0.0219 & 0.0205 & 0.0283 & 0.0270 & 0.0216 & 0.0238 & 0.0224 & 0.0215 \\
23 & 0.0216 & 0.0204 & 0.0275 & 0.0253 & 0.0203 & 0.0240 & 0.0221 & 0.0219 \\
24 & 0.0202 & 0.0188 & 0.0276 & 0.0237 & 0.0196 & 0.0230 & 0.0219 & 0.0207 \\
25 & 0.0204 & 0.0185 & 0.0273 & 0.0237 & 0.0189 & 0.0237 & 0.0216 & 0.0207 \\
26 & 0.0197 & 0.0177 & 0.0256 & 0.0212 & 0.0182 & 0.0222 & 0.0204 & 0.0200 \\
27 & 0.0199 & 0.0180 & 0.0269 & 0.0210 & 0.0181 & 0.0215 & 0.0209 & 0.0190 \\
28 & 0.0191 & 0.0168 & 0.0252 & 0.0208 & 0.0172 & 0.0215 & 0.0199 & 0.0176 \\
29 & 0.0184 & 0.0168 & 0.0244 & 0.0210 & 0.0178 & 0.0210 & 0.0198 & 0.0179 \\
30 & 0.0193 & \textbf{0.0161} & 0.0253 & 0.0211 & 0.0176 & 0.0212 & 0.0197 & 0.0177 \\
\hline
\end{tabular}}
\end{center}
\end{table}

\begin{table}[htbp]
\footnotesize
\caption{Acquisition time of the SDE criteria over 30 design rounds in the
two-dimensional plume example.}
\label{tab:supp-plume-runtime}
\begin{center}
\begin{tabular}{lcc}
\hline
Design criterion & Total acquisition time (s) & Time per round (s) \\
\hline
Sequential MI           & 233.392 & 7.780 \\
Sequential IMSPE        & 2.597   & 0.0866 \\
Sequential D-optimality & 131.536 & 4.385 \\
Maximin distance        & 0.0192  & 0.000641 \\
\hline
\end{tabular}
\end{center}
\end{table}

\clearpage

\subsection{Three-Dimensional Duffing-Oscillator Example}
\label{supp:duffing-example}
The controllable input is
$\mathbf{x}=(\omega,F,k)^\top$, where $\omega$ is the forcing frequency, $F$
is the forcing amplitude, and $k$ is the linear stiffness. The physical
response is generated by the hardening Duffing oscillator
\begin{equation}
\label{eq:supp-duffing-physical}
\ddot q(t)+2\theta_0\sqrt{k}\,\dot q(t)+kq(t)+\beta q(t)^3
=F\cos(\omega t),
\end{equation}
with $\beta=1.5$, and true damping ratio
$\theta_0=0.10$. The system is initialized from rest and integrated by a fixed-step fourth-order Runge--Kutta method with time step $0.04$ up to $T=240$. The physical response is the peak-to-peak displacement
over the final 60 time units,
\begin{equation}
\label{eq:supp-duffing-qoi}
\zeta(\mathbf{x})
=\max_{t\in[180,240]}q(t)-\min_{t\in[180,240]}q(t).
\end{equation}
The noisy physical observation is
\begin{equation}
\label{eq:supp-duffing-observation}
y(\mathbf{x})=\zeta(\mathbf{x})+\epsilon,
\qquad \epsilon\sim\mathcal{N}(0,0.02^2).
\end{equation}
The computer model omits the cubic stiffness term and uses the damping ratio
$\theta$ as its calibration parameter. This omission gives the following analytical expression of the steady-state peak-to-peak response:
\begin{equation}
\label{eq:supp-duffing-simulator}
\eta(\mathbf{x},\theta)
=\frac{2F}{\sqrt{(k-\omega^2)^2+
\left(2\theta\sqrt{k}\,\omega\right)^2}}.
\end{equation}
We assign a $\operatorname{Uniform}(0,1)$ prior to the normalized calibration
parameter, corresponding to a uniform prior over $[0.06,0.14]$ on the physical
scale. The simulator GP uses an ARD squared-exponential kernel with inverse
squared length scales $\operatorname{LogNormal}(2.2,0.8)$ and variance
$\operatorname{Gamma}(2,5)$. For the discrepancy GP, the corresponding priors
are $\operatorname{LogNormal}(3,0.8)$ and $\operatorname{Gamma}(2,10)$. In addition, we evaluate the response over $t\in[180,240]$, after the system
has reached steady state, as shown in \Cref{fig:supp-duffing-time-series}.

\begin{figure}[!htbp]
\centering
\includegraphics[width=0.98\textwidth]{duffing\_qt\_full\_240.pdf}

\vspace{0.5em}

\includegraphics[width=0.98\textwidth]{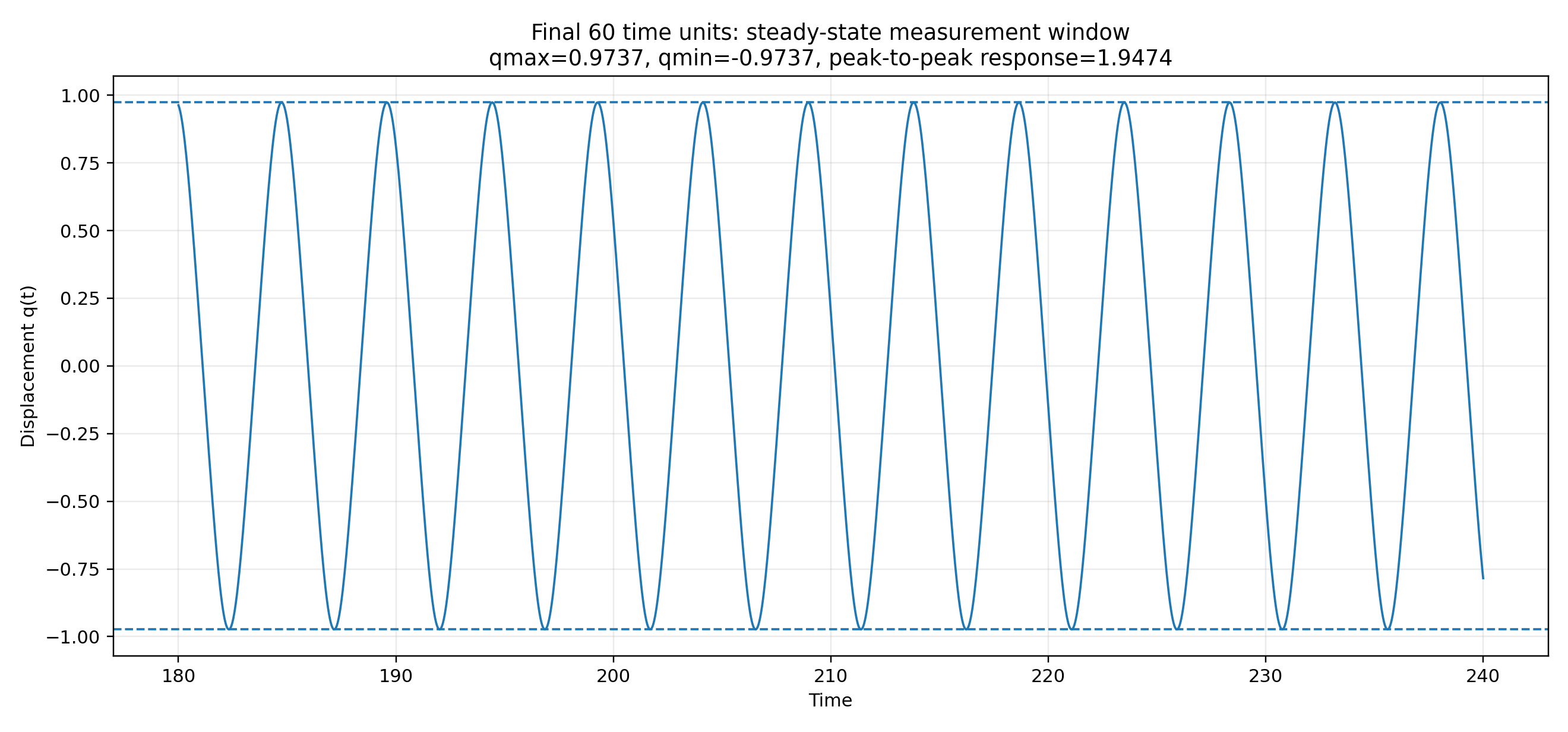}

\caption{Displacement trajectory of the physical Duffing oscillator over
the full integration interval (top) and the steady-state response over
$t\in[180,240]$ used to compute the peak-to-peak displacement (bottom).}
\label{fig:supp-duffing-time-series}
\end{figure}

\begin{figure}[!htbp]
\centering
\includegraphics[width=\textwidth,height=0.88\textheight,keepaspectratio]
{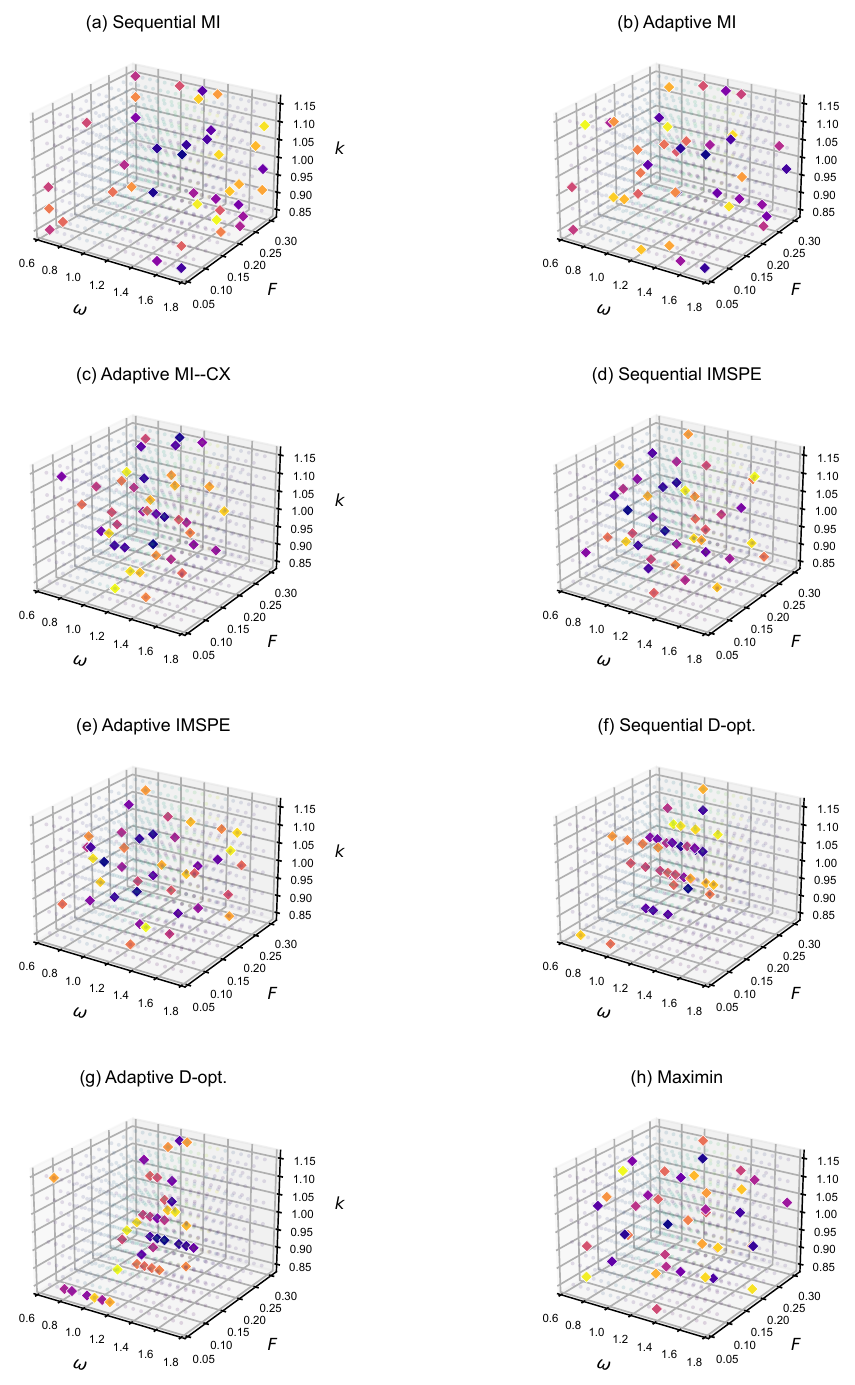}
\caption{Selected physical inputs in the three-dimensional
Duffing-oscillator example.}
\label{fig:supp-duffing-selected-inputs}
\end{figure}

\newcommand{\duffingsliceimage}[1]{%
\includegraphics[
width=\textwidth,
trim={0 0 0 58bp},
clip
]{#1}
}

\begin{figure}[!htbp]
\centering
\includegraphics[
  width=0.96\textwidth,
  height=0.82\textheight,
  keepaspectratio
]{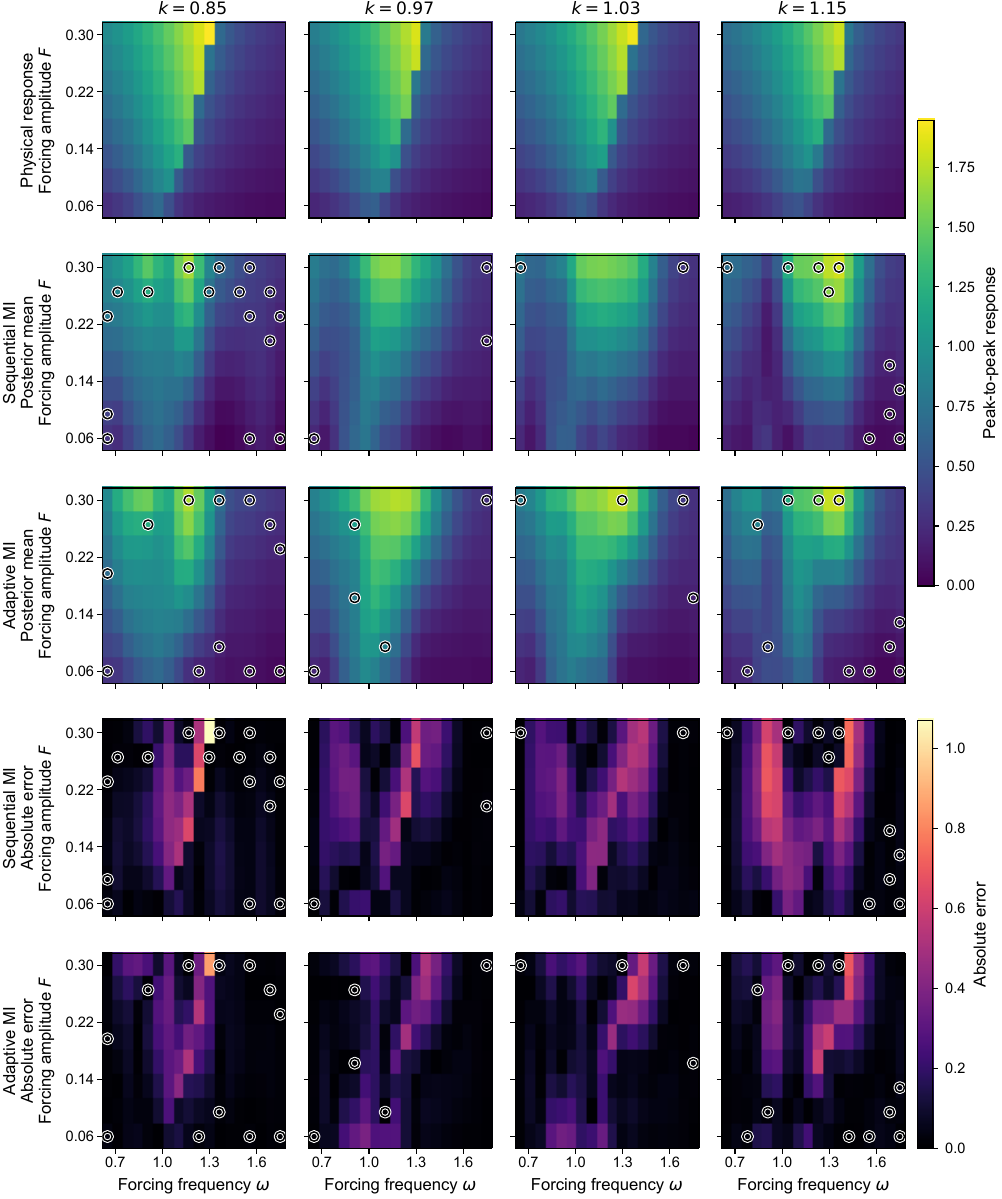}
\caption{Fixed-stiffness predictive slices at four representative stiffness
levels for Sequential MI and Adaptive MI. The rows show the physical
response, the two posterior means, and their corresponding absolute errors.
Open circles mark selected physical experiments.}
\label{fig:supp-duffing-mi-slices}
\end{figure}

\begin{figure}[!htbp]
\centering
\includegraphics[
  width=0.96\textwidth,
  height=0.82\textheight,
  keepaspectratio
]{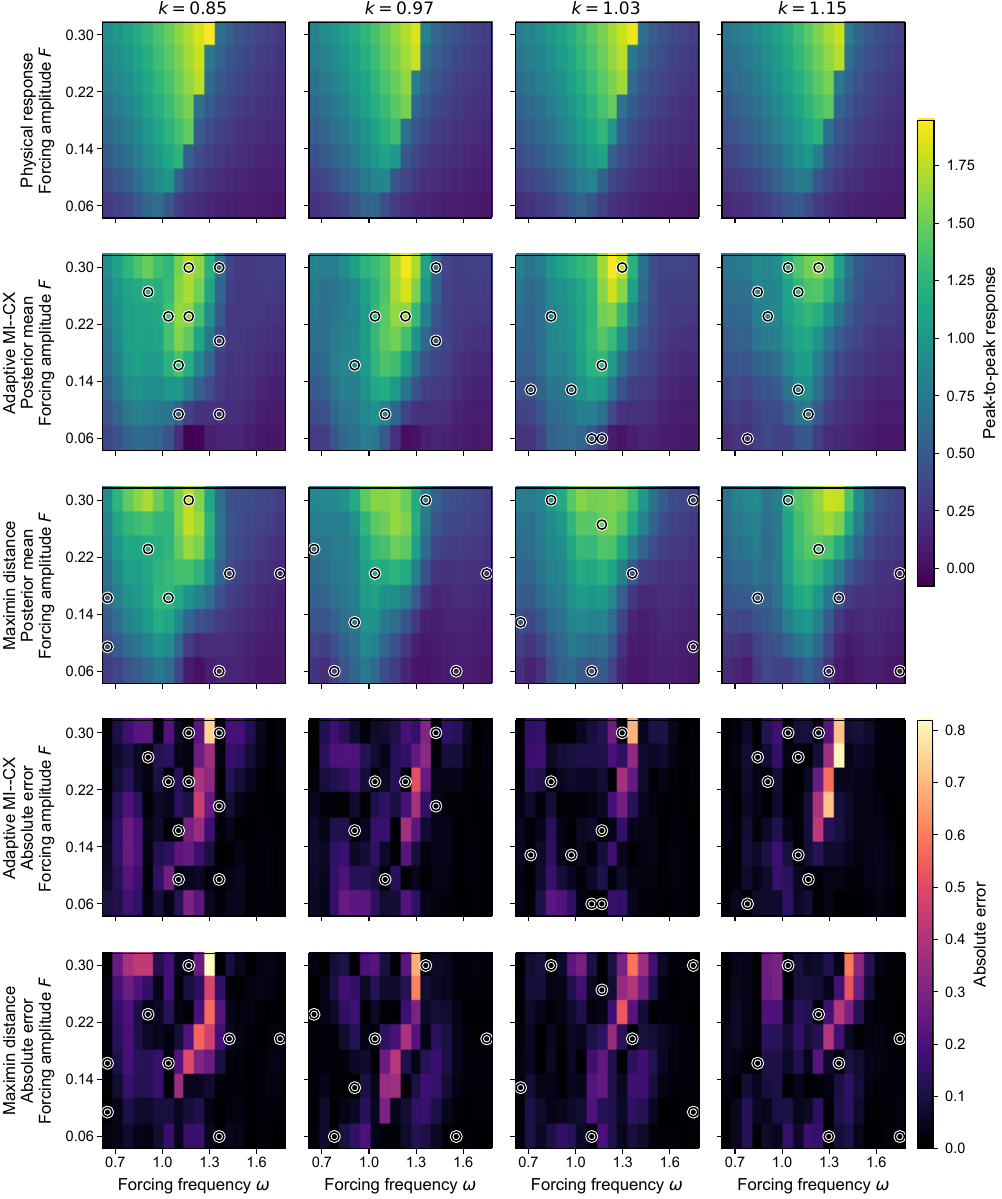}
\caption{Fixed-stiffness predictive slices at four representative stiffness
levels for Adaptive MI--CX and maximin distance. The rows show the physical
response, the two posterior means, and their corresponding absolute errors.
Open circles mark selected physical experiments.}
\label{fig:supp-duffing-micx-maximin-slices}
\end{figure}

\begin{figure}[!htbp]
\centering
\includegraphics[
  width=0.96\textwidth,
  height=0.82\textheight,
  keepaspectratio
]{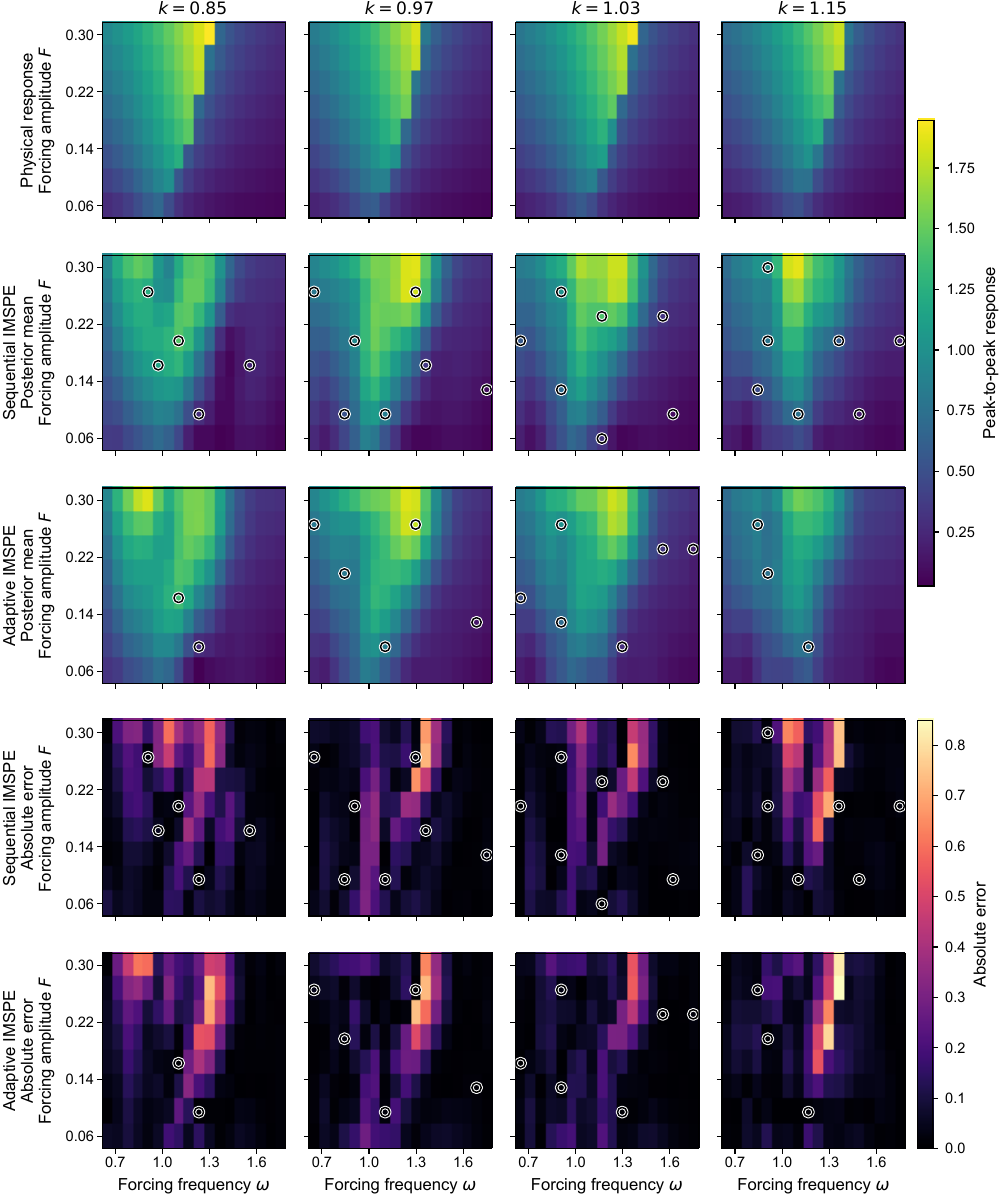}
\caption{Fixed-stiffness predictive slices at four representative stiffness
levels for Sequential IMSPE and Adaptive IMSPE. The rows show the physical
response, the two posterior means, and their corresponding absolute errors.
Open circles mark selected physical experiments.}
\label{fig:supp-duffing-imspe-slices}
\end{figure}

\begin{figure}[!htbp]
\centering
\includegraphics[
  width=0.96\textwidth,
  height=0.82\textheight,
  keepaspectratio
]{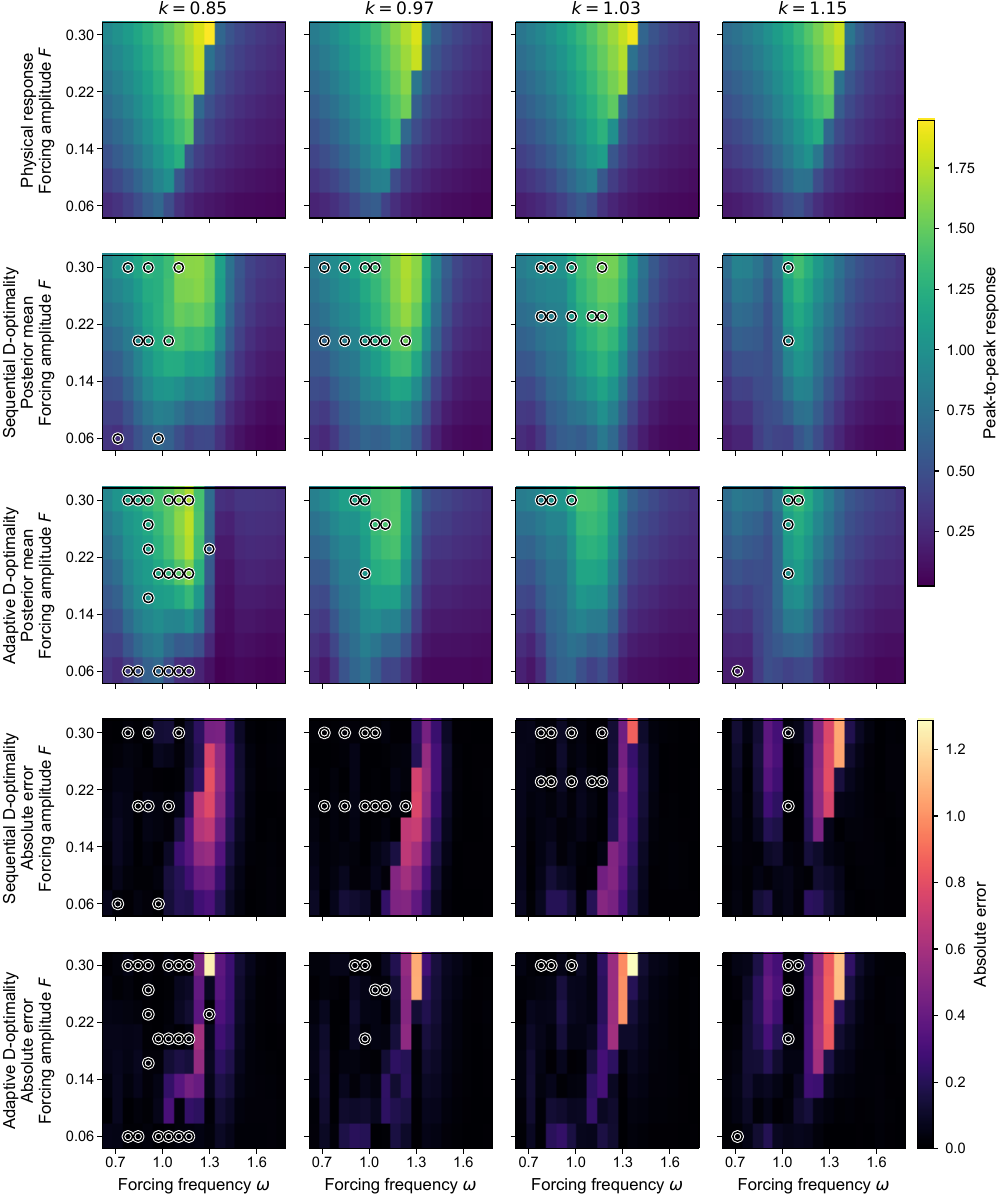}
\caption{Fixed-stiffness predictive slices at four representative stiffness
levels for Sequential and Adaptive D-optimality. The rows show the physical
response, the two posterior means, and their corresponding absolute errors.
Open circles mark selected physical experiments.}
\label{fig:supp-duffing-dopt-slices}
\end{figure}

\begin{table}[!htbp]
\scriptsize
\caption{MSE at each design round in the three-dimensional Duffing-oscillator example.}
\label{tab:supp-duffing-mse}
\centering
\resizebox{\textwidth}{!}{%
\begin{tabular}{ccccccccc}
\hline
Round & S-MI & A-MI & A-MI--CX & S-IMSPE & A-IMSPE & S-D-opt & A-D-opt & Maximin \\
\hline
0 & 0.1694 & 0.1694 & 0.1694 & 0.1694 & 0.1694 & 0.1694 & 0.1694 & 0.1694 \\
1 & 0.1636 & 0.1636 & 0.1018 & 0.1507 & 0.1507 & 0.1006 & 0.1006 & 0.1260 \\
2 & 0.1025 & 0.1025 & 0.1008 & 0.1036 & 0.1130 & 0.0959 & 0.1004 & 0.1427 \\
3 & 0.1030 & 0.1030 & 0.0971 & 0.0958 & 0.1011 & 0.0978 & 0.0910 & 0.1397 \\
4 & 0.1021 & 0.1021 & 0.0988 & 0.0929 & 0.0998 & 0.0833 & 0.0895 & 0.1386 \\
5 & 0.1026 & 0.1024 & 0.0971 & 0.0916 & 0.0995 & 0.0810 & 0.0784 & 0.1378 \\
6 & 0.1026 & 0.1015 & 0.0960 & 0.0911 & 0.0965 & 0.0793 & 0.0777 & 0.1376 \\
7 & 0.0995 & 0.1004 & 0.0925 & 0.0917 & 0.0970 & 0.0663 & 0.0751 & 0.1373 \\
8 & 0.0999 & 0.0987 & 0.0964 & 0.0903 & 0.0858 & 0.0652 & 0.0703 & 0.1349 \\
9 & 0.0991 & 0.0823 & 0.0952 & 0.0887 & 0.0857 & 0.0645 & 0.0620 & 0.1210 \\
10 & 0.1000 & 0.0803 & 0.0962 & 0.0864 & 0.0627 & 0.0633 & 0.0569 & 0.1207 \\
11 & 0.0982 & 0.0775 & 0.0949 & 0.0867 & 0.0589 & 0.0456 & 0.0576 & 0.1189 \\
12 & 0.0983 & 0.0752 & 0.0694 & 0.0861 & 0.0614 & 0.0456 & 0.0578 & 0.1272 \\
13 & 0.0930 & 0.0757 & 0.0511 & 0.0842 & 0.0572 & 0.0418 & 0.0569 & 0.1228 \\
14 & 0.0921 & 0.0775 & 0.0519 & 0.0812 & 0.0549 & 0.0383 & 0.0571 & 0.1242 \\
15 & 0.0740 & 0.0729 & 0.0489 & 0.0808 & 0.0517 & 0.0386 & 0.0574 & 0.1200 \\
16 & 0.0708 & 0.0708 & 0.0455 & 0.0795 & 0.0505 & 0.0381 & 0.0572 & 0.1234 \\
17 & 0.0664 & 0.0688 & 0.0341 & 0.0789 & 0.0482 & 0.0378 & 0.0546 & 0.1191 \\
18 & 0.0636 & 0.0623 & 0.0294 & 0.0795 & 0.0483 & 0.0375 & 0.0546 & 0.1193 \\
19 & 0.0610 & 0.0650 & 0.0283 & 0.0495 & 0.0482 & 0.0354 & 0.0537 & 0.0762 \\
20 & 0.0586 & 0.0617 & 0.0283 & 0.0468 & 0.0468 & 0.0355 & 0.0540 & 0.0770 \\
21 & 0.0541 & 0.0588 & 0.0276 & 0.0439 & 0.0460 & 0.0359 & 0.0522 & 0.0743 \\
22 & 0.0410 & 0.0471 & 0.0270 & 0.0455 & 0.0424 & 0.0373 & 0.0519 & 0.0730 \\
23 & 0.0382 & 0.0392 & 0.0252 & 0.0446 & 0.0356 & 0.0368 & 0.0512 & 0.0754 \\
24 & 0.0349 & 0.0362 & 0.0263 & 0.0440 & 0.0358 & 0.0398 & 0.0518 & 0.0723 \\
25 & 0.0352 & 0.0391 & 0.0233 & 0.0403 & 0.0353 & 0.0410 & 0.0469 & 0.0650 \\
26 & 0.0346 & 0.0385 & 0.0229 & 0.0389 & 0.0356 & 0.0403 & 0.0428 & 0.0575 \\
27 & 0.0351 & 0.0376 & 0.0233 & 0.0379 & 0.0343 & 0.0372 & 0.0431 & 0.0576 \\
28 & 0.0367 & 0.0376 & 0.0235 & 0.0372 & 0.0330 & 0.0363 & 0.0452 & 0.0389 \\
29 & 0.0365 & 0.0298 & 0.0239 & 0.0287 & 0.0327 & 0.0366 & 0.0469 & 0.0388 \\
30 & 0.0391 & 0.0300 & 0.0210 & 0.0285 & 0.0328 & 0.0377 & 0.0473 & 0.0378 \\
31 & 0.0469 & 0.0288 & 0.0221 & 0.0288 & 0.0328 & 0.0404 & 0.0461 & 0.0344 \\
32 & 0.0484 & 0.0294 & 0.0220 & 0.0285 & 0.0321 & 0.0453 & 0.0464 & 0.0267 \\
33 & 0.0501 & 0.0280 & 0.0182 & 0.0285 & 0.0317 & 0.0450 & 0.0468 & 0.0227 \\
34 & 0.0505 & 0.0281 & 0.0181 & 0.0283 & 0.0274 & 0.0427 & 0.0455 & 0.0224 \\
35 & 0.0535 & 0.0278 & 0.0198 & 0.0278 & 0.0262 & 0.0438 & 0.0446 & 0.0218 \\
36 & 0.0502 & 0.0278 & 0.0193 & 0.0280 & 0.0260 & 0.0434 & 0.0447 & 0.0210 \\
37 & 0.0502 & 0.0264 & 0.0154 & 0.0265 & 0.0260 & 0.0436 & 0.0446 & 0.0210 \\
38 & 0.0513 & 0.0264 & 0.0151 & 0.0263 & 0.0264 & 0.0448 & 0.0416 & 0.0205 \\
39 & 0.0502 & 0.0265 & 0.0153 & 0.0259 & 0.0259 & 0.0451 & 0.0421 & 0.0206 \\
40 & 0.0517 & 0.0268 & \textbf{0.0158} & 0.0261 & 0.0257 & 0.0444 & 0.0421 & 0.0201 \\
\hline
\end{tabular}}
\end{table}

\begin{table}[!htbp]
\scriptsize
\caption{CRPS at each design round in the three-dimensional Duffing-oscillator example.}
\label{tab:supp-duffing-crps}
\centering
\resizebox{\textwidth}{!}{%
\begin{tabular}{ccccccccc}
\hline
Round & S-MI & A-MI & A-MI--CX & S-IMSPE & A-IMSPE & S-D-opt & A-D-opt & Maximin \\
\hline
0 & 0.2027 & 0.2027 & 0.2027 & 0.2027 & 0.2027 & 0.2027 & 0.2027 & 0.2027 \\
1 & 0.1956 & 0.1956 & 0.1962 & 0.1859 & 0.1859 & 0.2090 & 0.2090 & 0.1788 \\
2 & 0.1981 & 0.1981 & 0.2075 & 0.1783 & 0.1798 & 0.2044 & 0.2095 & 0.2010 \\
3 & 0.2056 & 0.2056 & 0.2179 & 0.1886 & 0.1888 & 0.2163 & 0.2200 & 0.2084 \\
4 & 0.2007 & 0.1999 & 0.2130 & 0.1833 & 0.1862 & 0.2058 & 0.2151 & 0.2047 \\
5 & 0.1875 & 0.1872 & 0.1968 & 0.1735 & 0.1759 & 0.1982 & 0.1970 & 0.1932 \\
6 & 0.1751 & 0.1765 & 0.1871 & 0.1686 & 0.1684 & 0.1879 & 0.1899 & 0.1910 \\
7 & 0.1843 & 0.1770 & 0.1844 & 0.1723 & 0.1709 & 0.1794 & 0.1905 & 0.1902 \\
8 & 0.1792 & 0.1753 & 0.1914 & 0.1642 & 0.1608 & 0.1809 & 0.1820 & 0.1856 \\
9 & 0.1903 & 0.1844 & 0.1963 & 0.1702 & 0.1659 & 0.1885 & 0.1975 & 0.1906 \\
10 & 0.1908 & 0.1838 & 0.1922 & 0.1636 & 0.1716 & 0.1749 & 0.1885 & 0.1926 \\
11 & 0.1771 & 0.1732 & 0.1746 & 0.1488 & 0.1521 & 0.1589 & 0.1779 & 0.1801 \\
12 & 0.1676 & 0.1656 & 0.1563 & 0.1454 & 0.1489 & 0.1556 & 0.1699 & 0.1829 \\
13 & 0.1726 & 0.1644 & 0.1556 & 0.1524 & 0.1624 & 0.1483 & 0.1735 & 0.1840 \\
14 & 0.1698 & 0.1713 & 0.1635 & 0.1621 & 0.1546 & 0.1449 & 0.1794 & 0.1863 \\
15 & 0.1700 & 0.1700 & 0.1619 & 0.1630 & 0.1466 & 0.1457 & 0.1831 & 0.1833 \\
16 & 0.1563 & 0.1561 & 0.1494 & 0.1529 & 0.1415 & 0.1294 & 0.1727 & 0.1828 \\
17 & 0.1441 & 0.1462 & 0.1345 & 0.1408 & 0.1333 & 0.1382 & 0.1604 & 0.1728 \\
18 & 0.1380 & 0.1359 & 0.1278 & 0.1419 & 0.1313 & 0.1398 & 0.1652 & 0.1708 \\
19 & 0.1371 & 0.1405 & 0.1215 & 0.1327 & 0.1345 & 0.1386 & 0.1644 & 0.1548 \\
20 & 0.1428 & 0.1584 & 0.1275 & 0.1494 & 0.1502 & 0.1477 & 0.1664 & 0.1740 \\
21 & 0.1295 & 0.1442 & 0.1306 & 0.1390 & 0.1415 & 0.1389 & 0.1543 & 0.1640 \\
22 & 0.1178 & 0.1346 & 0.1212 & 0.1296 & 0.1224 & 0.1387 & 0.1669 & 0.1550 \\
23 & 0.1292 & 0.1361 & 0.1243 & 0.1298 & 0.1177 & 0.1495 & 0.1703 & 0.1683 \\
24 & 0.1237 & 0.1411 & 0.1189 & 0.1389 & 0.1228 & 0.1531 & 0.1725 & 0.1716 \\
25 & 0.1179 & 0.1221 & 0.1138 & 0.1358 & 0.1233 & 0.1516 & 0.1614 & 0.1514 \\
26 & 0.1210 & 0.1214 & 0.1136 & 0.1261 & 0.1142 & 0.1426 & 0.1549 & 0.1458 \\
27 & 0.1207 & 0.1290 & 0.1232 & 0.1266 & 0.1192 & 0.1547 & 0.1742 & 0.1566 \\
28 & 0.1268 & 0.1323 & 0.1289 & 0.1226 & 0.1123 & 0.1553 & 0.1830 & 0.1380 \\
29 & 0.1219 & 0.1171 & 0.1212 & 0.1183 & 0.1107 & 0.1536 & 0.1778 & 0.1310 \\
30 & 0.1364 & 0.1292 & 0.1164 & 0.1274 & 0.1131 & 0.1492 & 0.1881 & 0.1324 \\
31 & 0.1363 & 0.1208 & 0.1124 & 0.1255 & 0.1102 & 0.1453 & 0.1769 & 0.1280 \\
32 & 0.1396 & 0.1101 & 0.0960 & 0.1087 & 0.1025 & 0.1334 & 0.1538 & 0.1197 \\
33 & 0.1568 & 0.1021 & 0.1062 & 0.1087 & 0.1024 & 0.1388 & 0.1636 & 0.1169 \\
34 & 0.1559 & 0.1061 & 0.1074 & 0.1118 & 0.1030 & 0.1422 & 0.1648 & 0.1159 \\
35 & 0.1586 & 0.1034 & 0.0991 & 0.1080 & 0.0959 & 0.1362 & 0.1439 & 0.1111 \\
36 & 0.1411 & 0.0973 & 0.1012 & 0.1074 & 0.0931 & 0.1393 & 0.1479 & 0.1090 \\
37 & 0.1339 & 0.1022 & 0.1013 & 0.0996 & 0.0915 & 0.1430 & 0.1513 & 0.1083 \\
38 & 0.1391 & 0.0998 & 0.0909 & 0.1009 & 0.0901 & 0.1260 & 0.1374 & 0.1009 \\
39 & 0.1284 & 0.1039 & 0.0966 & 0.0919 & 0.0885 & 0.1326 & 0.1382 & 0.1096 \\
40 & 0.1473 & 0.0997 & 0.0974 & 0.0982 & \textbf{0.0842} & 0.1354 & 0.1466 & 0.1041 \\
\hline
\end{tabular}}
\end{table}

\begin{table}[!htbp]
\footnotesize
\caption{Acquisition time of the SDE criteria over 40 design rounds in the three-dimensional Duffing-oscillator example.}
\label{tab:supp-duffing-runtime}
\centering
\begin{tabular}{lcc}
\hline
Design criterion & Total acquisition time (s) & Time per round (s) \\
\hline
Sequential MI & 204.897 & 5.1224 \\
Sequential IMSPE & 0.865 & 0.0216 \\
Sequential D-opt. & 273.280 & 6.8320 \\
Maximin & 0.036 & 0.0009 \\
\hline
\end{tabular}
\end{table}

\clearpage
\subsection{Additional Ablation Results}
\label{supp:ablation}

\Cref{fig:supp-ablation-trajectories} compares the complete MSE and CRPS
trajectories of Adaptive MI, CX-only, and Adaptive MI--CX in the three
examples.

\begin{figure}[H]
\centering
\includegraphics[
  width=\textwidth,
  height=0.78\textheight,
  keepaspectratio
]{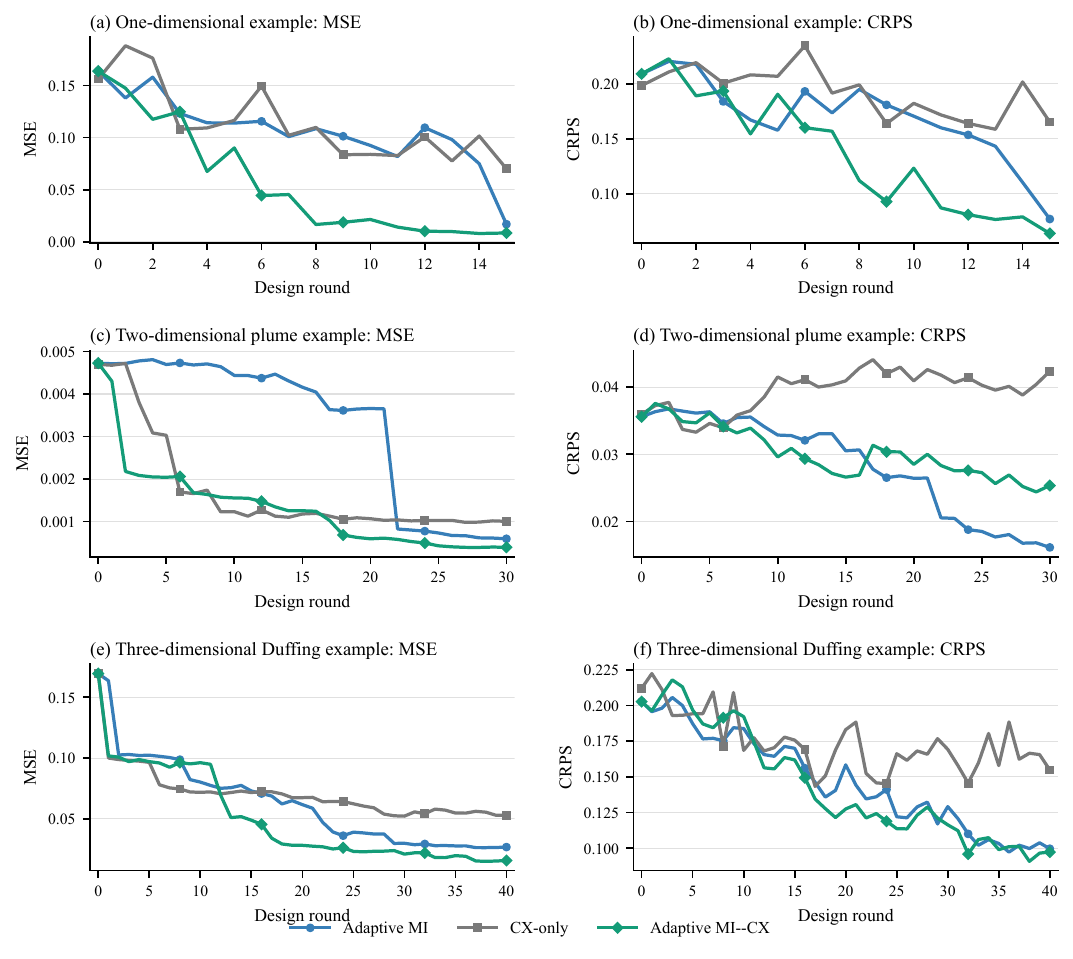}
\caption{Ablation results for the one-dimensional, two-dimensional plume,
and three-dimensional Duffing examples. The left and right columns show
MSE and CRPS, respectively.}
\label{fig:supp-ablation-trajectories}
\end{figure}

\FloatBarrier

\FloatBarrier
\bibliographystyle{siamplain}
\bibliography{references}

\end{document}